\documentclass[twoside,11pt]{article}

\usepackage[preprint]{jmlr2e}
\usepackage{mathtools}
\usepackage{bbm}
\usepackage{enumitem}

\newtheorem{method}{Method} 

\newcommand{\indep}{{\perp\!\!\!\perp}}
\newcommand{\notindep}{{\not\!\perp\!\!\!\perp}}
\newcommand{\Ex}{\mathbb{E}}
\newcommand{\Prob}{\mathbb{P}}
\newcommand{\upind}[1]{^{(#1)}}

\newcommand{\todistunifPW}[0]{\stackrel{\mathcal D; \mathcal P, \mathcal W}{\to}}
\DeclareMathOperator{\var}{var}
\DeclareMathOperator{\sign}{sign}

\DeclareMathOperator{\Cov}{Cov}

\graphicspath{{Pictures/}}

\ShortHeadings{The Weighted Generalised Covariance Measure}{Scheidegger, H\"orrmann and B\"uhlmann}
\firstpageno{1}

\begin{document}

\title{The Weighted Generalised Covariance Measure}

\author{\name Cyrill Scheidegger \email cyrill.scheidegger@stat.math.ethz.ch \\
       \addr Seminar for Statistics\\
       ETH Z\"urich\\
       8092 Z\"urich, Switzerland
       \AND
       \name Julia H\"orrmann \email jhoerrmann@ethz.ch \\
       \addr Department of Computer Science\\
       Seminar for Statistics\\
       ETH Z\"urich\\
       8092 Z\"urich, Switzerland
	\AND
	\name Peter B\"uhlmann \email peter.buehlmann@stat.math.ethz.ch \\
       \addr Seminar for Statistics\\
       ETH Z\"urich\\
       8092 Z\"urich, Switzerland}

\editor{Mladen Kolar}

\maketitle

\begin{abstract}
We introduce a new test for conditional independence which is based on what we call the \textit{weighted generalised covariance measure} (WGCM). It is an extension of the recently introduced \textit{generalised covariance measure} (GCM). To test the null hypothesis of $X$ and $Y$ being conditionally independent given $Z$, our test statistic is a weighted form of the sample covariance between the residuals of nonlinearly regressing  $X$ and $Y$ on $Z$. We propose different variants of the test for both univariate and multivariate $X$ and $Y$. We give conditions under which the tests yield the correct type I error rate. Finally, we compare our novel tests to the original GCM using simulation and on real data sets. Typically, our tests have power against a wider class of alternatives compared to the GCM. This comes at the cost of having less power against alternatives for which the GCM already works well. In the special case of binary or categorical $X$ and $Y$, one of our tests has power against all alternatives. 
\end{abstract}


\section{Introduction}
Conditional independence is a key concept for statistical inference. Where already \cite{DawidCondInd} argued that different important statistical concepts can be unified using conditional independence, it has received more attention during the last years. This is mainly because conditional independence plays a prominent role in the context of graphical models and causal inference. As a consequence, conditional independence tests form the basis of many algorithms for causal structure learning, see for example  \cite{PearlCausality} or \cite{ElementsCausInf}.

In contrast to unconditional independence (see for example \citealp{JosseHolmesMMA} for an overview), testing conditional independence is a hard statistical problem. In fact, it was recently proven by \citet{ShahPetersCondInd} that conditional independence is not a testable hypothesis. If the joint distribution of $(X,Y,Z)$ is absolutely continuous with respect to Lebesgue measure, then there is no test for the null hypothesis of $X$ and $Y$ being conditionally independent given $Z$ that has power against any alternative and at the same time controls the level for all distributions in the null hypothesis. A test for conditional independence therefore needs to make some assumptions to restrict the space of possible null distributions. In the following, we write $X\indep Y|Z$ for conditional independence of $X$ and $Y$ given $Z$. 

The review over nonparametric conditional independence tests for continuous variables by \citet{LiFanNonParCI} groups the tests into the following categories. 
\begin{description}
	\item[Discretization-based tests:] For discrete $Z$, testing $X\indep Y|Z$ reduces to unconditional independence testing. By discretising $Z$, the case of continuous $Z$ can also be treated in this way. Such approaches are being followed for example in \citet{HuangCIMaxNCC} and \cite{MargaritisDistributionFreeLearning}.
	\item[Metric-based tests:]
	\citet{SuWhiteTestingCIEmpLik} construct tests using the smoothed empirical likelihood ratio. \citet{SuWhiteConsCharFunTestCI} and \citet{WangCondDistCor} propose tests based on conditional characteristic functions. \citet{SuWhiteNonParHellingerMetricCI} introduce a test based on the weighted Hellinger distance between two conditional densities. A test proposed by \citet{RungeCITestNNMI} is based on conditional mutual information.
	\item[Permutation-based two-sample tests:] This category of tests reduces the problem of conditional independence testing to a two-sample test by permuting the sample in a way that leaves the joint distribution unchanged under the null hypothesis, see for example \citet{DoranPermutationBasedKCIT} and \citet{SenModelPoweredCITest}.
	\item[Kernel-based tests:] Tests in this category extend the Hilbert-Schmidt independence criterion to the conditional setting, see  \citet{KernelMeasCondInd}, \citet{KernelBasedCondIndCaus} and \citet{StroblApproxKernelBasedCI}.
	\item[Regression-based tests:] Many tests related to causal inference assume an additive noise model of the form
$$X=f(Z)+\eta_X, \quad Y=g(Z)+\eta_Y,$$
where $\eta_X$ and $\eta_Y$ are independent of $Z$ with mean zero. In this case, testing $X\indep Y|Z$ is equivalent to testing $\eta_X\indep \eta_Y$. Hence, a reasonable approach is to regress $X$ on $Z$ and $Y$ on $Z$ and then test (unconditional) independence of the residuals, see for example \citet{HoyerNonlinearCausalDiscoveryANM}, \citet{PetersCausalDiscoveryContinuousANM}, \citet{ZhangMeasuringCIByIndependentResiduals}, \citet{ZhangCausalDiscoveryUsingRegressionBasedCITests} and  \citet{RamseyScalCondInd}.  Instead of testing independence of the residuals, \citet{ShahPetersCondInd} introduce a test based on the sample covariance of the residuals. In view of the hardness of conditional independence testing, an advantage of regression based tests is that they convert the problem of restricting the null hypothesis to choosing appropriate regression or machine learning methods, which may be more accessible in practice.
  \item[Other tests:] Under the assumption that the conditional distribution of $X|Z$ is known at least approximately, it is possible to restore type I error control, see \citet{CondPermTest} and \citet{CondRandTest}. \cite{ICPNonlinear} propose some additional tests in the setting of nonlinear invariant causal prediction. \citet{AzadkiaSimple} introduce a new non-parametric coefficient of conditional dependence, based on which they construct a new variable selection algorithm.
\end{description}

\subsection{Our Contribution}\label{Sec_Contribution}
We introduce a new regression-based conditional independence test, which is a non-trivial and often more powerful extension of the \textit{generalised covariance measure} (GCM) introduced by \citet{ShahPetersCondInd}. We call it the \textit{weighted generalised covariance measure} (WGCM). For simplicity, assume that the random variables $X$ and $Y$ take values in $\mathbb R$. The test statistic of the GCM is a normalised sum of the product of the residuals of (nonlinearly) regressing $X$ on $Z$ and $Y$ on $Z$. Hence, the GCM essentially tests if $\Ex[\epsilon\xi]\neq 0$, where
$$\epsilon=X-\Ex[X|Z],\quad \xi =Y-\Ex[Y|Z].$$
For a more thorough treatment, see Section \ref{SubSec_Prerequisites}. Note that $\Ex[\epsilon\xi]=0$ under the null hypothesis of $X\indep Y|Z$, but we can also have $\Ex[\epsilon\xi]=0$ under an alternative.
Our WGCM however introduces an additional weight function. Instead of using the sum of the products of the residuals from the regression of $X$ on $Z$ and $Y$ on $Z$ as the basis of the test statistic, we weight this sum with an additional weight function depending on $Z$. Thus, the idea of the WGCM is to test $\Ex[\epsilon \xi w(Z)]\neq 0$ for some suitable weight function $w$ from the domain of $Z$ to $\mathbb R$. We will propose two different methods of the WGCM. WGCM.fix tests $\Ex[\epsilon \xi w(Z)]\neq 0$ for several fixed weight functions and aggregates the results. WGCM.est performs sample splitting and estimates a promising weight function on one part of the data and calculates the test statistic using this weight function on the other part of the data. To give conditions for the correct type I error rate of our tests, we can rely on the work of \citet{ShahPetersCondInd} and largely follow the proofs given there.

A bounded weight function $w$ satisfying $\Ex[\epsilon\xi w(Z)]\neq 0$ exists if and only if $\Ex[\epsilon\xi|Z]$ is not almost surely equal to $0$: If $\Ex[\epsilon\xi|Z]=0$ a.s., then for every bounded weight function $w: \mathbb R^{d_Z}\to\mathbb R$,
$$\Ex[\epsilon\xi w(Z)]=\Ex[\Ex[\epsilon\xi|Z]w(Z)]=0.$$
Conversely, if $\Ex[\epsilon\xi|Z]$ is not almost surely equal to $0$, then for example 
\begin{equation}\label{eq_WeightFunctionSign}
w(Z)=\sign(\Ex[\epsilon\xi|Z])
\end{equation}
 satisfies $\Ex[\epsilon\xi w(Z)] > 0$.

The two methods WGCM.fix and WGCM.est have power against a wider class of alternatives than the GCM, at the cost of typically being less powerful in situations where the GCM already works well. In view of the above considerations, we expect WGCM.fix and WGCM.est to be superior to GCM, when $\Ex[\epsilon\xi]$ is equal to $0$ or close to $0$, but $\Ex[\epsilon\xi |Z]$ is not. This can be supported by simulations.

Our main contributions are:
\begin{itemize}
	\item We introduce the two tests WGCM.fix and WGCM.est for conditional independence of two univariate random variables $X$ and $Y$ given a random vector $Z$.
	\item For both tests, we prove asymptotic guarantees for the level of the tests under appropriate conditions.
	\item For both tests, we prove theorems justifying that WGCM.fix and WGCM.est have full asymptotic power against a wider class of alternatives compared to the GCM. In particular, for categorical $X$ and $Y$ but continuous $Z$, WGCM.est leads to full asymptotic power against all alternatives.
	\item We introduce extensions mWGCM.fix and mWGCM.est for the case of multivariate $X$ and $Y$ and derive the corresponding results.
\end{itemize}
As in \citet{ShahPetersCondInd}, our theoretical results allow for high-dimensional conditioning variables Z.

In the following, we give a more precise overview about which methods have power in which situations. Let $\mathcal E_0$ be the collection of all distributions for $(X,Y,Z)$ that are absolutely continuous with respect to Lebesgue measure. For a distribution $P\in \mathcal E_0$, let $\Ex_P[\cdot]$ denote the expectation with respect to $P$ and consider the following subsets of distributions:
\begin{align}
\mathcal P_{GCM}&=\{P\in \mathcal E_0 | \Ex_P[\epsilon\xi]\neq 0\}\label{eq_Pgcm}\\
\mathcal P_{est}&=\{P\in \mathcal E_0 | \neg (\Ex_P[\epsilon\xi|Z]=0\, a.s.)\}\label{eq_Pest}\\
\mathcal P_{alt}&=\{P\in \mathcal E_0 | X \notindep Y | Z\}. \label{eq_Palt}
\end{align}
Moreover, for a fixed collection of bounded weight functions $\mathbf W = \{w_1,\ldots, w_K\}$ from $\mathbb R^{d_Z}$ to $\mathbb R$, let
\begin{equation}
\mathcal P_{\mathbf W}=\{P\in \mathcal E_0 | \exists k\in \{1,\ldots,k\} : \Ex_P[\epsilon\xi w_k(Z)]\neq 0\}. \label{eq_PW}
\end{equation}
By the argument prior to Equation (\ref{eq_WeightFunctionSign}), it follows that for a fixed $\mathbf W = \{w_1,\ldots, w_K\}$ with $w_1 = 1$, we have
$$\mathcal P_{GCM}\subset \mathcal P_{\mathbf W} \subset \mathcal P_{est} \subset \mathcal P_{alt}.$$

We will give conditions under which
\begin{itemize}
	\item GCM has asymptotic power against alternatives in $\mathcal P_{GCM}$ (see Theorem 8 in \citealp{ShahPetersCondInd}),
	\item WGCM.fix with weight functions $\mathbf W = \{w_1,\ldots, w_K\}$ has asymptotic power against alternatives in $\mathcal P_{\mathbf W}$ (see Corollary \ref{cor_PowerWGCMfix}),
	\item WGCM.est has asymptotic power against alternatives in $\mathcal P_{est}$ (see Corollary \ref{cor_PowerWGCMest}).
\end{itemize}

\subsubsection{Categorical Variables}
Our tests are also applicable in the setting of categorical $X$ and $Y$ and continuous $Z$. For simplicity, assume that $X$ and $Y$ take values in $\{0,1\}$ and $Z$ takes values in $\mathbb R^{d_Z}$. In this case,
\begin{align}
\Ex_P[\epsilon\xi|Z]&=\Ex_P[(X-E_P[X|Z])(Y-E_P[Y|Z])] \nonumber\\
&=\Ex_P[XY]-\Ex_P[X|Z]\Ex_P[Y|Z] \nonumber\\
&=\Prob_P(XY=1|Z)-\Prob_P(X=1|Z)\Prob_P(Y=1|Z).\label{eq_ExBinary}
\end{align}
It follows that a distribution $P$ satisfies $X\indep Y|Z$ if and only if $\Ex_P[\epsilon\xi|Z]=0, \, a.s.$ or using the notation (\ref{eq_Pest}) and (\ref{eq_Palt}), $\mathcal P_{est}=\mathcal P_{alt}$. Hence, for binary variables, we will give conditions under which the test WGCM.est has asymptotic power against all alternatives (see Section \ref{Sec_WGCMDiscrete}). Using dummy coding, this result can also be extended to arbitrary categorical $X$ and $Y$ (see Appendix \ref{Sec_MultDiscrete}).

\subsubsection{A Concrete Example} \label{Sec_ConcreteExample}
For $\lambda \in [0,1]$, let $h_\lambda(x)=\lambda x+0.5(1-\lambda)x^2$. Consider the setting
\begin{align}
Z&\sim \mathcal N(0,1), \quad \eta_X\sim 0.3\mathcal N(0,1), \quad \eta_Y\sim 0.3 \mathcal N(0,1),\nonumber\\
X&= Z+\eta_X, \label{eq_MotEx}\\
Y&=Z+\eta_Y+0.3 h_\lambda(X),\nonumber
\end{align}
where $Z$, $\eta_X$ and $\eta_Y$ are jointly independent.
Clearly, $X$ and $Y$ are not conditionally independent given $Z$. In Figure \ref{fig_MotEx}, we plot the rejection rates at level $\alpha=0.05$ for the three methods GCM, WGCM.est and WGCM.fix for different values of $\lambda \in [0,1]$. The rejection rates are calculated from $1000$ simulation runs with $n=200$ samples. We see that GCM outperforms the other methods for high values of $\lambda$, which corresponds to settings where the function $h_\lambda(x)$ introducing the dependence is nearly linear. However, for small and moderate values of $\lambda$, the methods WGCM.est and WGCM.fix still have considerable power, whereas the power of the GCM rapidly decreases.

\begin{figure}
	\begin{center}
	\includegraphics[width=0.5\textwidth]{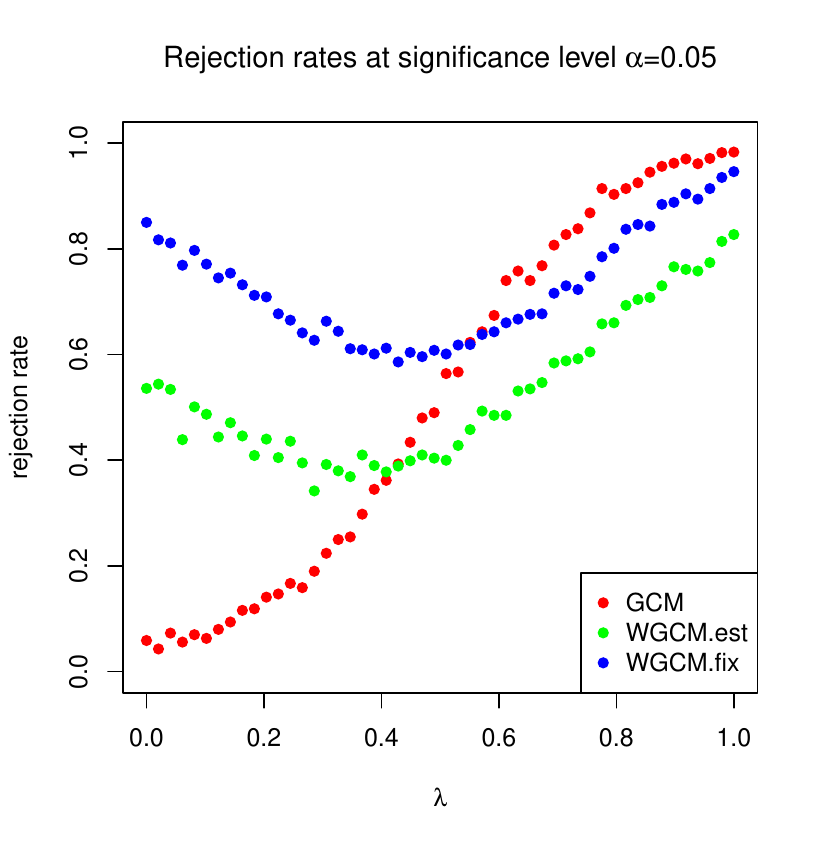}
	\end{center}
	\caption{The plot shows the rejection rates at level $\alpha=0.05$ from $1000$ simulation runs of (\ref{eq_MotEx}) with sample size $n=200$ for a range of $\lambda \in [0,1]$. The three methods GCM, WGCM.est and WGCM.fix are based on regressions using splines. For WGCM.est, we use $30\%$ of the samples to estimate the weight function $\sign(\Ex[\epsilon\xi|Z])$ and for WGCM.fix, we use $8$ fixed weight functions using our default choice described in Section \ref{sec_ChoiceWeightFun}.}
	\label{fig_MotEx}
\end{figure}

To understand the phenomenon, let us consider the cases of $\lambda=1$ and $\lambda =0$ more closely. Using the notation from before, we have
$$\epsilon =X-\Ex[X|Z]=\eta_X$$
and

$$\xi=Y-\Ex[Y|Z]=\begin{cases}
\eta_Y+0.3\eta_X,  & \lambda = 1, \\
\eta_Y+0.15(2\eta_X Z+\eta_X^2-1), & \lambda =0.
\end{cases}$$

If $\lambda=1$, we have $\Ex[\epsilon \xi] = 0.3$, so we expect the GCM to have power in this case, which is also supported by our simulation. If $\lambda=0$, however, $\Ex[\epsilon \xi]=0$, since $\Ex[\eta_X \eta_Y]$, $\Ex[\eta_X^2 Z]$ and $\Ex[\eta_X^3]$ are all equal to zero. Therefore, we do not expect the GCM to have power larger than the significance level $\alpha=0.05$ in this case, which is also visible in the plot. On the other hand, in the case of $\lambda = 0$, we have that 
$$\Ex[\epsilon\xi|Z]=0.3\Ex[\eta_X^2 Z|Z]=0.3 Z,$$
which is not almost surely equal to zero. Thus, we expect the two variants of the WGCM to have power. This is supported by the simulation.

\subsection{Conventions and Notation}
As most of the arguments are along the same lines as in \cite{ShahPetersCondInd}, we also largely use the same notation.

We will use the following setting:
Let $X$, $Y$ and $Z$ be random vectors taking values in $\mathbb R^{d_X}$, $\mathbb R^{d_Y}$ and $\mathbb R^{d_Z}$. For $i=1,\ldots, n$, let $(x_i,y_i,z_i)$ be i.i.d. copies of $(X,Y,Z)$. Let $\mathbf X\upind n\in \mathbb R^{n\times d_X}$, $\mathbf Y\upind n\in \mathbb R^{n\times d_Y}$ and $\mathbf Z\upind n\in \mathbb R^{n\times d_Z}$ be the data matrices with rows $x_i$, $y_i$ and $z_i$, respectively.

Whereas at some places, we will also consider categorical $X$ and $Y$, for most of the time, we assume that the joint distribution of $(X,Y,Z)$ is absolutely continuous with respect to Lebesgue measure. Let $p(x,y,z)$ be the joint Lebesgue density of $(X,Y,Z)$. 

We say that $X$ and $Y$ are conditionally independent given $Z$, written as 
$$X\indep Y|Z,$$
if one of the following two equivalent characterizations holds for Lebesgue almost all $(x,y,z)\in\mathbb R^{d_X}\times \mathbb R^{d_Y}\times \mathbb R^{d_Z}$ with $p(z)>0$:
\begin{enumerate}
	\item $p(x,y|z)=p(x|z)p(y|z)$;
	\item $p(x|y,z)=p(x|z)$,
\end{enumerate}
see for example Section 3.1 in \cite{DawidCondInd}.

Let $\mathcal E_0$ be the collection of all distributions for $(X,Y,Z)$ that are absolutely continuous with respect to Lebesgue measure. Let $\mathcal P_0\subseteq \mathcal E_0$ be the set of distributions in $\mathcal E_0$ such that $X$ and $Y$ are conditionally independent given $Z$.

We use $\Ex_P[\cdot]$ to denote the expectation of random variables under the probability distribution $P$ and we write $\Prob_P(\cdot)$ for the corresponding probability measure $\Ex_P\left[\mathbbm 1{\{\cdot\}}\right]$. We denote indicator functions with $\mathbbm 1\{\cdot\}$ or $\mathbbm 1_A$.

We write $\Phi$ for the cumulative distribution function of a standard normal random variable, that is, for all $x\in\mathbb R$, we have $\Phi(x)=\Prob(X\leq x)$, where $X\sim \mathcal N(0,1)$.

Let $\mathcal P$ be a collection of probability distributions. For a sequence of random variables $\left(V_{P,n}\right)_{n\in\mathbb N, P\in\mathcal P}$ with laws determined by $P\in\mathcal P$, we write $V_{P,n}=o_{\mathcal P}(1)$ if for all $\epsilon>0$
$$\lim_{n\to\infty}\sup_{P\in\mathcal P}\Prob_P\left(|V_n|>\epsilon\right)=0.$$
We write $V_{P,n}=O_{\mathcal P}(1)$ if
$$\limsup_{M\to\infty}\sup_{P\in\mathcal P}\sup_{n\in\mathbb N}\Prob_P\left(|V_n|>M\right)=0.$$
For another sequence $\left(W_{P,n}\right)_{n\in\mathbb N, P\in\mathcal P}$ of strictly positive random variables, we write $V_{P,n}=o_{\mathcal P}(W_{P,n})$ and $V_{P,n}=O_{\mathcal P}(W_{P,n})$ if $V_{P,n}/W_{P,n}=o_{\mathcal P}(1)$ and $V_{P,n}/W_{P,n}=O_{\mathcal P}(1)$, respectively.

For a random variable $Z$ taking values in (a subset of) $\mathbb R^{d_Z}$, we typically say \textit{for all $z\in\mathbb R^{d_Z}$} instead of \textit{for all $z$ in the support of $Z$}.

\subsection{Outline}
In Section \ref{Sec_UnivariateWGCM}, we describe the WGCM for univariate $X$ and $Y$. After a motivation as an extension of the GCM, we present two variants of the WGCM. The first variant WGCM.est is based on sample splitting to estimate a weight function, whereas the second variant WGCM.fix uses multiple fixed weight functions and aggregates the results. We provide Theorems justifying level and power of the methods. In Section \ref{Sec_Experiments}, we compare our methods to the original GCM using simulation and also apply the methods to some real data sets in the context of variable selection. In the appendix, we show, how the WGCM can be extended to the case of multivariate $X$ and $Y$ and present the proofs of our various results.

\section{The Univariate Weighted Generalised Covariance Measure}\label{Sec_UnivariateWGCM}
In this section, we introduce the \textit{weighted generalised covariance measure} (WGCM), which is an extension of the \textit{generalised covariance measure} (GCM) recently introduced by \citet{ShahPetersCondInd}. We first treat the case of $d_X=d_Y=1$ and present different variants of the WGCM in this case. In Appendix \ref{sec_MultWGCM}, we give extensions for multivariate $X$ and $Y$.

\subsection{Prerequisites} \label{SubSec_Prerequisites}
We consider the same setup as in Section 3.1 of \citet{ShahPetersCondInd}.

Let $d_X=d_Y=1$ and $d_Z\geq 1$. For any distribution $P\in\mathcal E_0$, let $\epsilon_P=X-\Ex_P[X|Z]$ and $\xi_P=Y-\Ex_P[Y|Z]$ and for $z\in \mathbb R^{d_Z}$, define the functions $f_P(z)=\Ex_P[X|Z=z]$ and $g_P(z)=\Ex_P[Y|Z=z]$. Then, we can write
$$X=f_P(Z)+\epsilon_P,\quad Y=g_P(Z)+\xi_P.$$
For $i=1,\ldots,n$, let $\epsilon_{P,i}=x_i-f_P(z_i)$ and $\xi_{P,i}=y_i-g_P(z_i)$. Moreover, let $u_P(z)=\Ex_P[\epsilon_P^2|Z=z]$ and $v_P(z)=\Ex_P[\xi_P^2|Z=z]$.

Let $\hat f^{(n)}$ and $\hat g^{(n)}$ be estimates of $f_P$ and $g_P$, obtained by regression of $\mathbf{X}^{(n)}$, respectively $\mathbf{Y}^{(n)}$ on $\mathbf{Z}^{(n)}$.

In the following, the dependence on $n$ and $P$ is sometimes omitted. 
The GCM by \citet{ShahPetersCondInd} uses the products of the residuals
$$R_i=\left(x_i-\hat f(z_i)\right)\left(y_i-\hat g(z_i)\vphantom{x_i-\hat f(z_i)}\right),\quad i=1,\ldots,n$$
as the basis of the test statistic
\begin{equation}\label{eq_DefStatGCM}
T^{(n)}=\frac{\frac{1}{\sqrt n}\sum_{i=1}^n R_i}{\left(\frac{1}{n}\sum_{i=1}^n R_i^2-\left(\frac{1}{n}\sum_{r=1}^n R_r\right)^2\right)^{1/2}}\eqqcolon \frac{\tau_N^{(n)}}{\tau_D^{(n)}},
\end{equation}
which is a normalised sum of the $R_i$. Under the null hypothesis $X\indep Y| Z$ and suitable conditions on the convergence of $\hat f$ and $\hat g$ to $f$ and $g$, the test statistic $T\upind n$ converges in distribution to a standard normal random variable, see Theorem 6 in \citet{ShahPetersCondInd}.

The test statistic of the GCM is a normalised estimate of $\Ex_P[\epsilon_P\xi_P]$. Under the null hypothesis of $X\indep Y|Z$, we have
$$\Ex_P[\epsilon_P\xi_P|Z]=\Ex_P[\Ex_P[\xi_P|Z,X]\epsilon_P]=0,$$
using $\Ex_P[\xi|X,Z]=\Ex_P[\xi|X]$ a.s. However, it is also possible to have $\Ex_P[\epsilon_P\xi_P]=0$ and $X\notindep Y|Z$. Therefore, the GCM only has power against alternatives with  $\Ex_P[\epsilon_P\xi_P]\neq 0$, see Theorem 8 in \citet{ShahPetersCondInd}.

By introducing an additional weight function in the test statistic, one can modify the set of alternatives against which the test has power. Let $w:\mathbb R^{d_Z}\to \mathbb R$ be a bounded measurable function. If we use the weighted product of the residuals
$$ R_i=\left(x_i-\hat f(z_i)\right)\left(y_i-\hat g(z_i)\vphantom{x_i-\hat f(z_i)}\right) w(z_i),\quad i=1,\ldots,n $$
as the basis of the test statistic defined in (\ref{eq_DefStatGCM}), $T\upind n$ is now a normalised estimate of $\Ex_P[\epsilon_P\xi_Pw(Z)]$. Under the null hypothesis of $X\indep Y|Z$, we still have
$$\Ex_P[\epsilon_P\xi_P w(Z)]= \Ex_P[\Ex_P[\xi_P|Z,X]\epsilon_P w(Z)]=0,$$
but now, the test has power against alternatives with $\Ex_P[\epsilon_P \xi_P w(Z)]\neq 0$. A weight function $w$ with $\Ex_P[\epsilon_P\xi_Pw(Z)]\neq 0$ exists if and only if $\Ex_P[\epsilon_P\xi_P|Z]$ is not almost surely equal to $0$, see Equation (\ref{eq_WeightFunctionSign}). Typically, $\Ex_P[\epsilon_P\xi_P|Z]$ is unknown, so we do not know a suitable weight function. A possible strategy is to perform sample splitting and estimate a weight function on the first part of the sample (WGCM.est). This approach is treated next. Another approach is to calculate the test statistic for multiple weight functions and then to aggregate the results (WGCM.fix). This will be treated in Section \ref{SubSec_WGCM1dMultFix}.

\subsection{WGCM with Estimated Weight Function (WGCM.est)}\label{SubSec_WGCMEst}
We have seen in Equation (\ref{eq_WeightFunctionSign}) that a desirable weight function for the WGCM is for example $w(z)=\sign\left(\Ex[\epsilon\xi|Z=z]\right)$. We do not know $\Ex[\epsilon\xi|Z]$ in practice, but we can estimate it. We propose the following procedure:
\begin{method}[WGCM.est]\label{meth_WGCMEst}
Using one random sample split, create two independent data sets $(\mathbf X\upind n,\mathbf Y\upind n,\mathbf Z\upind n)$ and $\mathbf A=(\mathbf X_A,\mathbf Y_A,\mathbf Z_A)$. We use the data set $\mathbf A$ to estimate a weight function and calculate the test statistic on the data set $(\mathbf X\upind n,\mathbf Y\upind n,\mathbf Z\upind n)$ as in Section \ref{SubSec_Prerequisites}. For ease of notation, we still assume that $(\mathbf X\upind n,\mathbf Y\upind n,\mathbf Z\upind n)$ consists of $n$ samples, whereas the size $a_n$ of the data set $\mathbf A=(\mathbf X_A,\mathbf Y_A,\mathbf Z_A)$ is arbitrary (but depends on $n$). The ratio between the sizes of the data sets $(\mathbf X\upind n,\mathbf Y\upind n,\mathbf Z\upind n)$ and $\mathbf A$ is difficult to choose in practice. We propose to estimate a weight function as follows:
	\begin{enumerate}
		\item (Nonlinearly) regress $\mathbf X_A$ on $\mathbf Z_A$ to get $\hat f_A$ and $\mathbf Y_A$ on $\mathbf Z_A$ to get $\hat g_A$. Let $\hat\epsilon_{A,i}=x_{A,i}-\hat f_{A}(z_{A,i})$ and $\hat\xi_{A,i}=y_{A,i}-\hat g_A(z_{A,i})$.
	\item (Nonlinearly) regress $\left(\hat\epsilon_{A,i}\hat\xi_{A,i}\right)_{i=1}^{a_n}$ on $\mathbf Z_A$ to get $\hat h$ which is an estimate of $h(\cdot)=\Ex_P[\epsilon\xi|Z=\cdot]$.
	\item Set $\hat w\upind n(\cdot)=\sign(\hat h(\cdot))$.
	\end{enumerate}
\end{method}

The following two theorems do not assume this particular choice of $\hat w\upind n$ based on the $\sign$, but only require it to be estimated on a data set $\mathbf A$ independent of $(\mathbf X\upind n,\mathbf Y\upind n,\mathbf Z\upind n)$.
Note that the estimated weight function $\hat w\upind n$ is not required to converge for $n\to\infty$. This is important, since under the null hypothesis, we have $h(Z)=\Ex_P[\epsilon\xi|Z]=0$ a.s. Thus, an estimate of $w$ of the form $\hat w\upind n (\cdot)=\sign(\hat h (\cdot))$ will typically not converge under the null hypothesis.

We consider the setting of Section \ref{SubSec_Prerequisites} with the difference of replacing $R_i$  by
\begin{equation}\label{eq_DefRWGCM1DEst}
R_i\upind n =\left(x_i-\hat f\upind n(z_i)\right)\left(y_i-\hat g\upind n(z_i)\right) \hat w\upind n(z_i), \quad i=1,\ldots, n,
\end{equation}
and redefine the test statistic (\ref{eq_DefStatGCM}) by
\begin{equation}\label{eq_DefStatWGCMEst}
T^{(n)}=\frac{\frac{1}{\sqrt n}\sum_{i=1}^n R_i\upind n}{\left(\frac{1}{n}\sum_{i=1}^n {R_i\upind n}^2-\left(\frac{1}{n}\sum_{r=1}^n R_r\upind n\right)^2\right)^{1/2}}\eqqcolon \frac{\tau_N^{(n)}}{\tau_D^{(n)}}.
\end{equation}

\subsubsection{Distribution Under the Null Hypothesis}\label{Sec_WGCMEstH0}
We will repeatedly use the quantities
\begin{align}
A_f&=\frac{1}{n}\sum_{i=1}^n(f_P(z_i)-\hat f(z_i))^2, &A_g&=\frac{1}{n}\sum_{i=1}^n(g_P(z_i)-\hat g(z_i))^2, \label{eq_DefAfAg}\\
B_f&=\frac{1}{n}\sum_{i=1}^n(f_P(z_i)-\hat f(z_i))^2v_P(z_i), &B_g&=\frac{1}{n}\sum_{i=1}^n(g_P(z_i)-\hat g(z_i))^2u_P(z_i). \label{eq_DefBfBg}
\end{align}
The following theorem gives conditions under which the test statistic $T\upind n$ is asymptotically standard normal. In the case of constant weight function $w(z)=1$ (and without sample splitting), it reduces to Theorem 6 in \citet{ShahPetersCondInd}.
\begin{theorem}[WGCM.est]\label{thm_WGCM1DEst}
Let $A_f$, $A_g$, $B_f$ and $B_g$ be defined as in (\ref{eq_DefAfAg}) and (\ref{eq_DefBfBg}). Assume that the weight function $\hat w\upind n$ is estimated on a data set $\mathbf A$ independent of $(\mathbf X\upind n,\mathbf Y\upind n,\mathbf Z\upind n)$ and let $T\upind n$ be defined as in (\ref{eq_DefStatWGCMEst}). Assume that there exists $C>0$ such that for all $n\in\mathbb N$ we have $|\hat w\upind n(z)|\leq C$ for all $z\in\mathbb R^{d_Z}$. Define 
$$\bar\sigma_n^2=\bar\sigma_{n,P}^2=\Ex_P\left[\hat w\upind n (Z)^2\epsilon_P^2\xi_P^2|\mathbf A\right].$$
\begin{enumerate}
	\item Let $P\in\mathcal P_0$. Assume that $A_fA_g=o_P(n^{-1})$, $B_f=o_P(1)$ and $B_g= o_P(1)$. If there exists $\eta > 0$ such that $\Ex_P\left[|\epsilon_P \xi_P |^{2+\eta}\right]<\infty$ and if $P$-almost surely there exists $c>0$ such that $\inf_{n\in\mathbb N} \bar\sigma_n^2 \geq c$, then
	$$\sup_{t\in\mathbb R} |\Prob_P(T^{(n)}\leq t)-\Phi(t)|\to 0.$$
	\item Let $\mathcal P\subset \mathcal P_0$. Assume that $A_f A_g=o_\mathcal P(n^{-1})$, $B_f=o_\mathcal P(1)$ and $B_g= o_\mathcal P(1)$. If there exists $\eta>0$ such that $\sup_{P\in\mathcal P}\Ex_P\left[|\epsilon_P\xi_P|^{2+\eta}\right]<\infty$ and if there exists $c>0$ such that for all $P\in\mathcal P$, we have $P$-almost surely $\inf_{n\in\mathbb N}\bar\sigma_n^2 \geq c$, then
	$$\sup_{P\in\mathcal P}\sup_{t\in\mathbb R} |\Prob_P(T^{(n)}\leq t)-\Phi(t)|\to 0.$$
\end{enumerate}
\end{theorem}

A proof can be found in Appendix \ref{App_ProofUnivariateWGCM}.

\begin{remark}
\begin{enumerate}
\item An estimate of the form $\hat w\upind n (z)=\sign(\hat h\upind n (z))$ satisfies
$$\Ex_P\left[\hat w\upind n(Z)^2\epsilon_P^2\xi_P^2|\mathbf A\right]=\Ex_P\left[\epsilon_P^2\xi_P^2\right] \text{ a.s.},$$
so $\bar \sigma_n^2\geq c$ is satisfied.

\item It would be desirable to do both the estimation of the weight function and calculating the test statistic on the full sample. However, the condition that the weight functions are estimated on an independent data set cannot be removed in general. If the weight function was allowed to depend on $(\mathbf X\upind n,\mathbf Y\upind n,\mathbf Z\upind n)$, one could take functions $\hat w\upind n$ with $\hat w\upind n(z_i)=\sign\left((x_i-\hat f(z_i))(y_i-\hat g(z_i))\right)$. This would always lead to a positive test statistic, which would contradict the asymptotic normality under the null hypothesis.

\item The conditions on the quantities $A_f$, $A_g$, $B_f$ and $B_g$ are reasonably weak. It is for example enough if $\hat f$ and $\hat g$ have mean squared prediction error (MSPE) of order $o(n^{-1/2})$, see Remark 7 in \citet{ShahPetersCondInd}. An MSPE of order $o(n^{-1/2})$ can for example be obtained for real-valued $Z$ with bounded support if $f$ and $g$ are Lipschitz continuous, see for example \cite{DistFreeTheoNonparReg}, or for high-dimensional $Z$ with sparse and linear $f$ and $g$, see for example \cite{BuehlHDStat}. Moreover, the rates can also be satisfied using kernel ridge regression, see Section 4 of \cite{ShahPetersCondInd}.
\end{enumerate}
\end{remark}

\subsubsection{Power}
In order to state a general power result, we need to make additional assumptions. We follow the path of Theorem 8 in \citet{ShahPetersCondInd} and assume that $\hat f$ and $\hat g$ have been estimated on another additional data set independent of $(\mathbf X^{(n)}, \mathbf Y^{( n)},\mathbf Z^{( n)})$ and $\mathbf A$. This means that we have three independent data sets involved.
\begin{enumerate}
	\item An auxiliary data set to estimate $\hat f$ and $\hat g$.
	\item The data set $(\mathbf X\upind n,\mathbf Y\upind n,\mathbf Z\upind n)$ to calculate the test statistic.
	\item The data set $\mathbf A$ to estimate the weight function $\hat w \upind n$.
\end{enumerate}
Whereas the sample splitting, that is, the independence of the $\mathbf A$ and $(\mathbf X\upind n,\mathbf Y\upind n,\mathbf Z\upind n)$ is important in practice, the auxiliary data set to estimate $\hat f$ and $\hat g$ is mainly needed for technical reasons. In practice, it is usually recommended to use the original version of WGCM.est used in Section \ref{Sec_WGCMEstH0}, see also Section 3.1.1. in \citet{ShahPetersCondInd} for a more detailed discussion.

\begin{theorem}[WGCM.est]\label{thm_PowerWGCM1DEst}
Consider the setup of Theorem \ref{thm_WGCM1DEst}.
Let $A_f$, $A_g$, $B_f$ and $B_g$ be defined as in (\ref{eq_DefAfAg}) and (\ref{eq_DefBfBg}) with the difference that $\hat f$ and $\hat g$ have been estimated on an auxiliary data set independent of $(\mathbf X\upind n, \mathbf Y\upind n,\mathbf Z\upind n)$ and $\mathbf A$. Assume that there exists $C>0$ such that for all $n\in\mathbb N$ we have $|\hat w\upind n(z)|\leq C$ for all $z\in\mathbb R^{d_Z}$. For $P\in\mathcal E_0$, let 
$$\bar\rho_{P,n}=\Ex_P\left[\epsilon_P\xi_P \hat w\upind n(Z)|\mathbf A\right]\quad \text{and} \quad \bar \sigma_n^2=\bar\sigma_{P,n}^2=\var_P\left(\epsilon_P\xi_P \hat w\upind n(Z)|\mathbf A\right).$$
Then, the following holds:
\begin{enumerate}
	\item Let $P\in \mathcal E_0$. Assume that $A_f A_g=o_P(n^{-1})$, $B_f=o_P(1)$ and $B_g=o_P(1)$. Assume that there exists $\eta >0$ such that $\Ex_P\left[|\epsilon_P\xi_P|^{2+\eta}\right]<\infty$ and that $P$-almost surely there exists $c>0$ such that $\inf_{n\in\mathbb N}\bar\sigma_n^2 \geq c$.
	Then, we have
	$$\sup_{t\in \mathbb R} \left|\Prob_P\left(T\upind n - \sqrt n\frac{\bar \rho_{P,n}}{\tau_D\upind n}\leq t\right)-\Phi(t)\right| \to 0,$$
	where $\tau_D\upind n$ is defined in (\ref{eq_DefStatGCM}).
	\item Let $\mathcal P\subset \mathcal E_0$. Assume that $A_f A_g=o_\mathcal P(n^{-1})$, $B_f=o_\mathcal P(1)$ and $B_g= o_\mathcal P(1)$. If there exists $\eta>0$ such that $\sup_{P\in\mathcal P}\Ex_P\left[|\epsilon_P\xi_P|^{2+\eta}\right]<\infty$ and if there exists $c>0$ such that for all $P\in\mathcal P$, we have $P$-almost surely $\inf_{n\in \mathbb N}\bar\sigma_n^2 \geq c$, then
	$$\sup_{P\in\mathcal P}\sup_{t\in \mathbb R} \left|\Prob_P\left(T\upind n - \sqrt n\frac{\bar \rho_{P,n}}{\tau_D\upind n}\leq t\right)-\Phi(t)\right| \to 0.$$
\end{enumerate}

\end{theorem}
A proof can be found in Appendix \ref{App_ProofUnivariateWGCM}.
\begin{remark}\label{Rmk_PowerWGCMFixed1}
If instead of estimating the weight function, one uses a fixed weight function $w$, Theorem \ref{thm_PowerWGCM1DEst} implies that if $\rho_P=\Ex_P[\epsilon_P\xi_Pw(Z)]\neq 0$, the test statistic $T \upind n$ is of order $\sqrt n\frac{\rho_P}{\sigma_P}$. That is, if $\rho_P \neq 0$, the WGCM with fixed weight function $w$ has asymptotic power $1$ against alternative $P$.
\end{remark}

We recommend to obtain $\hat w \upind n$ by estimating $w(z)=\sign(\Ex_P[\epsilon\xi|Z=z])$, for example using Method \ref{meth_WGCMEst}. Recall the notation $\mathcal P_{est}$ from (\ref{eq_Pest}). Fix $P\in\mathcal P_{est}$, that is, $\Ex_P[\epsilon\xi|Z]$ is not almost surely equal to $0$ and assume that we can consistently estimate $w(z)=\sign(\Ex_P[\epsilon\xi|Z=z])$, wherever $\Ex_P[\epsilon\xi|Z=z]\neq 0$, that is
$$\Ex_P\left[(\hat w\upind n(Z)-w(Z))^2\mathbbm 1\{\Ex_P[\epsilon\xi|\mathbf Z]\neq 0\}|\mathbf A\right]\to 0 \text{ in probability}.$$
Defining
$$\rho_P=\Ex_P[\epsilon\xi w(Z)]=\Ex_P[|\Ex_P[\epsilon\xi|Z]|]>0,$$
it follows that $\bar \rho_{P,n}-\rho_P=o_P(1)$, because by the Cauchy-Schwarz inequality
\begin{align*}
|\bar \rho_{P,n}- \rho_P|&=\left|\Ex_P\left[\epsilon\xi(\hat w\upind n(Z)- w(Z))|\mathbf A\right]\right|\\
&\leq\Ex_P\left[\epsilon^2\xi^2\right]^{1/2}\Ex_P\left[(\hat w\upind n(Z)-w(Z))^2\mathbbm 1\{\Ex_P[\epsilon\xi|\mathbf Z]\neq 0\}|\mathbf A\right]^{1/2}.
\end{align*}
Hence, with high probability, $\bar \rho_{P,n}$ is bounded away from $0$ and we arrive at the following corollary.
\begin{corollary}[WGCM.est]\label{cor_PowerWGCMest}
Let $P\in\mathcal P_{est}$, that is $\Ex_P[\epsilon_P\xi_P|Z]$ is not almost surely equal to $0$. In the setting of Theorem \ref{thm_PowerWGCM1DEst}, assertion 1., assume that 
$$\Ex_P\left[(\hat w\upind n(Z)-w(Z))^2\mathbbm 1\{\Ex_P[\epsilon\xi|\mathbf Z]\neq 0\}|\mathbf A\right]\to 0\text{ in probability.}$$
Then, for all $M>0$,
$$\Prob_P(T\upind n \geq M)\to 1,$$
that is, WGCM.est has asymptotic power $1$ against alternative $P$ for any significance level $\alpha\in (0,1)$.
\end{corollary}

\begin{remark}\label{rmk_PowerWGCM1DEst}
Under the conditions of Corollary \ref{cor_PowerWGCMest}, WGCM.est has asymptotic power $1$ for alternatives $P \in \mathcal P_{est}$. This is a larger class compared to the GCM, which has power against alternatives in the class $\mathcal P_{GCM}$, that is $\Ex_P[\epsilon\xi]\neq 0$.
\begin{enumerate}
	\item However, with Method \ref{meth_WGCMEst} our additional requirement that $w(z)=\sign(\Ex_P[\epsilon\xi|Z])$ can be consistently estimated is not straightforward to verify, since there are two regressions involved. Intuitively, we will have that for the regression in step 1 of Method \ref{meth_WGCMEst}, $\hat\epsilon_A\hat\xi_A$ is close to $\epsilon\xi$. Even if the regression method in 2 is consistent, we still need that it is also not too much affected by the difference between $\hat\epsilon_A\hat\xi_A$ and $\epsilon\xi$. This will depend on the regression method applied in step 2.
	\item One could use the following observation to obtain an alternative method to estimate $\hat w\upind n$. By definition,
	$$\Ex_P[\epsilon \xi|Z]=\Ex_P[(X-f(Z))(Y-g(Z)|Z]=\Ex_P[XY|Z]-f(Z)g(Z).$$
	Hence, additionally to the functions $\hat f$ and $\hat g$, one could try to also estimate the function $\Ex_P[XY|Z=z]$ using a regression of $XY$ on $Z$. The consistency condition in Corollary \ref{cor_PowerWGCMest} could then be replaced by consistency conditions for estimating $\hat f$, $\hat g$ and $\Ex[XY|Z=z]$ which might be easier to justify than for the two-step approach of Method \ref{meth_WGCMEst}. However, the estimation of $\sign(\Ex_P[\epsilon\xi|Z])$ seems to be less reliable for finite samples using this method compared to using Method \ref{meth_WGCMEst}, which also leads to reduced power.
\end{enumerate}
\end{remark}

\subsubsection{Binary Variables}\label{Sec_WGCMDiscrete}
If $X$ and $Y$ are binary, we can use the same methodology to obtain a test that has asymptotic power 1 against any alternative, provided that the conditions of Theorem 5 and Corollary \ref{cor_PowerWGCMest} hold.

Assume that $X$ and $Y$ take values in $\{0,1\}$. By (\ref{eq_ExBinary}), we have that $\Ex_P[\epsilon\xi|Z]=0$ a.s. if and only if $X\indep Y|Z$. We thus obtain the following result for binary $X$ and $Y$.

\begin{corollary}[WGCM.est, binary case]
Let $X$ and $Y$ be binary and assume that the distribution $P$ of $(X,Y,Z)$ satisfies $X\notindep Y| Z$. In the setting of Theorem \ref{thm_PowerWGCM1DEst}, assertion 1., assume that $\Ex_P\left[(\hat w\upind n(Z)-w(Z))^2\mathbbm 1\{\Ex_P[\epsilon\xi|\mathbf Z]\neq 0\}|\mathbf A\right]\to 0$ in probability. Then, for all $M>0$,
$$\Prob_P(T\upind n \geq M)\to 1,$$
that is, WGCM.est has asymptotic power $1$ against alternative $P$ for any significance level $\alpha\in (0,1)$.
\end{corollary}

Using dummy coding, this methodology can be extended to arbitrary categorical $X$ and $Y$ variables, see Appendix \ref{Sec_MultDiscrete}.

\subsection{WGCM With Several Fixed Weight Functions (WGCM.fix)}\label{SubSec_WGCM1dMultFix}
We consider an alternative to the sample splitting approach. We calculate the test statistic for several fixed weight functions and aggregate the results. Consider the same setting as in Section \ref{SubSec_Prerequisites}. Let $\{w_k\}_{k=1}^K$ be bounded functions from $\mathbb R^{d_Z}\to \mathbb R$. For $k=1,\ldots, K$, let $\mathbf R_k\in \mathbb R^n$ be the vector of products of the residuals weighted by $w_k$, that is,
$$\mathbf R_k\upind n=\left(\begin{array}{c}(x_1-\hat f(z_1)) (y_1-\hat g(z_1))w_k(z_1)\\ \vdots\\ (x_n-\hat f(z_n)) (y_n-\hat g(z_n))w_k(z_n)\end{array}\right).$$
Let $T_k\upind n$ be the test statistic of the WGCM based on the vector $\mathbf R_k\upind n$, that is,
\begin{equation}\label{eq_DefStatMultFixWOneDim}
T_k\upind n=\frac{\sqrt n \bar{\mathbf R}_k}{\left(\frac{1}{n}\|\mathbf R_k\|_2^2-\bar{\mathbf R}_k^2\right)^{1/2}}\eqqcolon \frac{\tau_{N,k}\upind n}{\tau_{D,k}\upind n},
\end{equation}
where $\bar{\mathbf R}_k$ is the sample average of the coordinates of ${\mathbf R_k}$. Finally, let
$$\mathbf T\upind n=\left(T_1\upind n,\ldots, T_K\upind n\right)^T.$$
For a fixed number $K$ of weight functions, the simplest approach would be to perform an individual test for each weight function $w_k$ and use Bonferroni correction to aggregate the $K$ $p$-values. With the aggregated test statistic
$$S_n=\max_{k=1,\ldots K} |T_k\upind n|,$$
the Bonferroni corrected $p$-value is therefore
$$p_\textup{Bon}\upind n=K\cdot 2\left(1-\Phi(S_n)\right).$$
In this case, it is straightforward to obtain a variant of Theorem \ref{thm_WGCM1DEst} (stating that $p_{Bon}\upind n$ is a conservative $p$-value) and a variant of Theorem \ref{thm_PowerWGCM1DEst} (stating that the method has asymptotic power 1 if there exists $k\in\{1,\ldots, K\}$ with $\Ex_P[\epsilon\xi w_k(Z)]\neq 0$).

However, we can also use more sophisticated methods to calculate a $p$-value for $S_n=\max_{k=1,\ldots K} |T_k\upind n|$. With $K$ fixed, it is possible to show that under the null hypothesis of $P\in \mathcal P_0$ (i.e. $X\indep Y|Y$), the vector $\mathbf T\upind n$ converges to a multivariate Gaussian distribution. In fact, we go one step further and propose the same procedure as in Section 3.2 of \citet{ShahPetersCondInd}, where the multivariate case of the (unweighted) GCM is treated. For technical reason, it is assumed that $K\geq 3$. We do not assume that $K$ is fixed, but it is allowed to grow with $n$. Define 
$$\hat\Sigma_{kl}=\frac{\frac{1}{n}\mathbf R_k^T \mathbf R_l-\bar{\mathbf R}_k\bar {\mathbf R}_l}{\left(\frac{1}{n}\|\mathbf R_k\|_2^2-\bar{\mathbf R}_k^2\right)^{1/2}\left(\frac{1}{n}\|\mathbf R_l\|_2^2-\bar{\mathbf R}_l^2\right)^{1/2}}.$$

Let $\hat {\mathbf T}\upind n \in \mathbb R^{K}$ have a multivariate normal distribution with covariance matrix $\hat \Sigma$ and mean $0$. Let
$$\hat S_n=\max_{k=1,\ldots, K} |\hat T\upind n_k|$$
and let $\hat G_n$ be the quantile function of $\hat S_n$ given $\hat \Sigma$. Note that $\hat G_n$ is random and depends on the data. $\hat G_n$ can be approximated by simulation. Recent results by \cite{DistMaxMultNorm2019} also allow to calculate $\hat G_n$ analytically. For a significance level $\alpha\in (0,1)$ we propose the following test.
\begin{method}[WGCM.fix]\label{meth_WGCMFix}
Reject the null hypothesis $X\indep Y|Z$ if $S_n > \hat G_n (1-\alpha)$. The corresponding $p$-value is given by 
$$p\upind n_\textup{WGCM.fix}=1-\hat G_n^{-1}(S_n)=\Prob_P(\hat S_n > s|\hat \Sigma)\big|_{s=S_n}.$$
\end{method}

We need the following conditions on the errors $\epsilon_P$ and $\xi_P$. Let $B\geq 1$.
\begin{itemize}[leftmargin=15mm]
	\item[(A1a)] $\max_{r=1,2}\Ex_P\left[|\epsilon_P\xi_P|^{2+r}/B^r\right]+\Ex_P\left[\exp(|\epsilon_P\xi_P|/B)\right]\leq 4$;
	\item[(A1b)] $\max_{r=1,2}\Ex_P\left[|\epsilon_P\xi_P|^{2+r}/B^{r/2}\right]+\Ex_P\left[|\epsilon_P\xi_P|^4/B^2\right]\leq 4$;
	\item[(A2)] $B^2\left(\log(Kn)\right)^7/n\leq C n^{-c}$ for some constants $C,c>0$ that do not depend on $P\in\mathcal P$. 
\end{itemize}

We obtain the following theorem. It is the adaptation of Theorem 9 in \cite{ShahPetersCondInd} to our setting. 
\begin{theorem}[WGCM.fix]\label{thm_MultFixedWOneDim}
Let $\mathcal P\subset \mathcal P_0$ and let $A_f$, $A_g$, $B_f$ and $B_g$ be defined as in (\ref{eq_DefAfAg}) and (\ref{eq_DefBfBg}). Assume that there exist $C, c\geq 0$ such that for all $n\in\mathbb N$ and $P\in\mathcal P$ there exists $B\geq 1$ such that either (A1a) and (A2) or (A1b) and (A2) hold. Furthermore, assume that there exist $C_1, c_1>0$ (independent of $n$) such that for all $k=1,\ldots, K$ and $P\in\mathcal P$ we have $|w_k|\leq C_1$ and $\Ex_P[\epsilon_P^2\xi_P^2 w_k(Z)^2]\geq c_1$. Assume that
\begin{align}
A_f A_g=o_\mathcal P\left(n^{-1}\log(K)^{-2}\right), \label{eq_MultFixedWOneDimCondA0}\\
B_f=o_\mathcal P\left(\log(K)^{-4}\right), \quad B_g=o_\mathcal P\left(\log(K)^{-4}\right).\label{eq_MultFixedWOneDimCondB0}
\end{align}
Assume that there exist sequences $\left(\tau_{f,n}\right)_{n\in\mathbb N}$ and $\left(\tau_{g,n}\right)_{n\in\mathbb N}$ of real numbers such that
\begin{align}
\max_{i=1,\ldots, n}|\epsilon_{P,i}|=O_\mathcal P(\tau_{g,n}),\quad A_{g}=o_\mathcal P\left(\tau_{g,n}^{-2}\log(K)^{-2}\right), \label{eq_MultFixedWOneDimCondC0}\\
\max_{i=1,\ldots, n}|\xi_{P,i}|=O_\mathcal P(\tau_{f,n}),\quad A_{f}=o_\mathcal P\left(\tau_{f,n}^{-2}\log(K)^{-2}\right). \label{eq_MultFixedWOneDimCondD0}
\end{align}
Then,
$$\sup_{P\in\mathcal P}\sup_{\alpha\in (0,1)}|\Prob_P(S_n\leq \hat G_n(\alpha))-\alpha|\to 0.$$
\end{theorem}
The proof can be found in Appendix \ref{App_ProofsWGCMMultFix}.

\begin{remark}\phantomsection\label{rmk_MultFixedWOneDim}
\begin{enumerate}
	\item If the errors $\epsilon_P$ and $\xi_P$ are sub-Gaussian with parameters bounded by some $M>0$ independent of $P\in\mathcal P$, then by Lemma \ref{lem_ProdSubGauss} in Appendix \ref{App_SubGaussExp}, their product $\epsilon_P\xi_P$ has a sub-exponential distribution, with parameters bounded independent of $P\in\mathcal P$, see also Remark 10 in \cite{ShahPetersCondInd}. A summary of results on sub-Gaussian and sub-exponential distributions can be found in Appendix \ref{App_SubGaussExp}. If $\epsilon_P\xi_P$ is sub-exponential with bounded parameters, then condition (A1a) is satisfied: By Definition \ref{def_SubExp}, 2., there exists $K_2>0$ such that for all $P\in\mathcal P$, 
	$$\Ex_P\left[\exp(|\epsilon\xi|/K_2)\right]\leq 2.$$
	For $B_0\geq 1$, we get that
	$$\Ex_P[|\epsilon\xi|^{2+r}/B_0^r]\leq 2 (2+r)! K_2^{2+r}/B_0^r.$$ 
Thus, we can choose $B_0\geq 1$ such that for $r=1,2$ we have $\Ex_P[|\epsilon\xi|^{2+r}/B_0^r]\leq 2$ and set $B=\max(K_2, B_0)$.
	
If $\epsilon_P$ and $\xi_P$ are sub-Gaussian with bounded parameters, by Corollary \ref{cor_MaxSubGauss}, we also have
$$\max_{i=1,\ldots, n} |\epsilon_i|=O_P\left(\sqrt{\log(n)}\right)\text{ and } \max_{i=1,\ldots, n}|\xi_i|=O_P\left(\sqrt{\log(n)}\right).$$
If we for example have that both $A_f, A_g=o_\mathcal P(n^{-1/2}\log(K)^{-2})$, then (\ref{eq_MultFixedWOneDimCondA0}), (\ref{eq_MultFixedWOneDimCondC0}) and (\ref{eq_MultFixedWOneDimCondD0}) are satisfied. Alternatively, we can also work under the conditions of Theorem \ref{thm_MultFixedWMultDim}. Then, the corresponding Remark \ref{rmk_ThmMultFixedWMultDim} holds.
	\item Compared to Theorem 9 in \cite{ShahPetersCondInd}, we are still in a simplified setting, because due to the boundedness of the $w_k$, we only need bounds on $\epsilon_P\xi_P$ instead of $\epsilon_{P,j}\xi_{P,l}$ for all combinations of $j=1,\ldots, d_X$ and $l=1,\ldots, d_Y$. This will change when
we consider the multivariate case, see Theorem \ref{thm_MultFixedWMultDim}.
\end{enumerate}
\end{remark}
The approach WGCM.fix always yields a lower $p$-value than using Bonferroni correction. To see this, let $X=(X_1, \ldots, X_k)^T\sim N(0, \Sigma)$, for a covariance matrix $\Sigma$ with $\Sigma_{kk}=1$ for all $k=1,\ldots, K$. Then, we have for $s\geq 0$
$$\Prob_P\left(\max_{k=1,\ldots, K} |X_k|> s\right)= \Prob_P\left(\bigcup_{k=1}^K \{|X_k|> s\}\right)\leq \sum_{k=1}^K\Prob_P(|X_k|> s)=K\cdot 2\left(1-\Phi(s)\right).$$
By replacing $\Sigma$ with $\hat \Sigma$, and $s$ by $S_n$, it follows that $p\upind n_\textup{WGCM.fix}\leq p\upind n_\textup{Bon}$. As an immediate consequence of Theorem \ref{thm_PowerWGCM1DEst} and Remark \ref{Rmk_PowerWGCMFixed1}, we get the following result on the power of WGCM.fix for a fixed number $K$ of weight functions. Recall the set $\mathcal P_{\mathbf W}$ from (\ref{eq_PW}).
\begin{corollary}[WGCM.fix]\label{cor_PowerWGCMfix}
Let $P\in\mathcal E_0$. Let $A_f$, $A_g$, $B_f$ and $B_g$ be defined as in (\ref{eq_DefAfAg}) and (\ref{eq_DefBfBg}) with the difference that $\hat f$ and $\hat g$ have been estimated on an auxiliary data set independent of $(\mathbf X^{(n)}, \mathbf Y^{( n)},\mathbf Z^{( n)})$. Let $K\geq 1$ be fixed. Assume that there exists $C>0$ such that for all $z\in \mathbb R^{d_Z}$ and all $k=1,\ldots, K$, we have $|w_k(z)|\leq C$. Assume that $A_f A_g=o_P(n^{-1})$, $B_f=o_P(1)$ and $B_g=o_P(1)$ as well as $\Ex_P[\epsilon_P^2\xi_P^2w_k(Z)^2]>0$ for all $k=1,\ldots, K$ and $\Ex_P\left[\epsilon_P^2\xi_P^2\right]<\infty$. If $P\in \mathcal P_{\mathbf W}$ for $\mathbf W = \{w_1,\ldots, w_K\}$, that is if there exists $k\in\{1,\ldots, K\}$ such that $\Ex_P[\epsilon_P\xi_Pw_k(Z)]\neq 0$, then for all $M>0$,
$$\Prob_P(S_n \geq M)\to 1,$$
that is, WGCM.fix with fixed number $K$ of weight functions has asymptotic power $1$ against alternative $P$ for any significance level $\alpha\in (0,1)$.
\end{corollary}

\subsubsection{Choice of Weight Functions}\label{sec_ChoiceWeightFun}
We give a few heuristics on how to choose the weight functions in practice.

For a fixed alternative $P\in \mathcal E_0$, a promising weight function $w:\mathbb R^{d_Z}\to \mathbb R$ satisfies $\Ex_P[\epsilon \xi w(Z)]\neq 0$. We know, that such a $w$ exists if and only if $\Ex_P[\epsilon\xi|Z]$ is not almost surely equal to $0$.

Let us first assume that $d_Z=1$. For $a\in \mathbb R$, define the functions
$$w_a(z)= \sign(z-a)=\begin{cases}
-1, & z < a \\
1, & z\geq a.
\end{cases}$$
Then, we have that $\Ex_P[\epsilon\xi |Z]=0$ a.s. if and only if $\Ex_P[\epsilon \xi w_a(Z)]=0$ for all $a$ in the support of $Z$. To see this, define $h(z)=\Ex_P[\epsilon \xi|Z=z]$. Then,
$$\Ex_P[\epsilon \xi w_a(Z)]=\Ex_P[h(Z) w_a(Z)]\\
=\int_a^\infty h(z) p_Z(z)\text{d}z-\int_{-\infty}^a h(z) p_Z(z)\text{d}z.$$
If $\Ex_P[\epsilon \xi w_a(Z)]=0$ for all $a$ in the support of $Z$, then taking the derivative with respect to $a$ yields $h(z)=0$ for all $z$ in the support of $Z$.

Therefore, it intuitively seems a good idea (in addition to the constant weight function $w(z)=1$) to use the functions $w_{a_1}, \ldots, w_{a_{K-1}}$ for some $a_1,\ldots, a_{K-1}\in \mathbb{R}$. In practice, one can for example take $a_1, \ldots, a_{K-1}$ at the empirical $\frac{1}{K}, \ldots, \frac{K-1}{K}$-quantiles of $Z$. However, the choice of $K$ is a difficult problem. Theorem \ref{thm_MultFixedWOneDim} allows for $K$ to be large compared to $n$. However, if $K$ is too large, one is in danger of performing too many tests and loosing power again. This tradeoff makes the choice of $K$ difficult. In practice, we have experienced that a small number of weight functions is usually sufficient. For the experiments in Section \ref{Sec_Experiments}, we will use $K=8$ weight functions.

For $d_Z>1$, one can for example take functions
$$w_{d,k}(\mathbf z)= \sign(z_d -a_{d,k}),\, d=1,\ldots, d_Z, \, k=1,\ldots, k_0,$$
where $a_{d,k}$ is the empirical $\frac{k}{k_0+1}$-quantile of $Z_j$. This means that including the constant weight function $w(z)=1$, we have $K=k_0\cdot d_Z+1$ weight functions in total. For the experiments in Section \ref{Sec_Experiments}, we will use $k_0=7$ weight functions per dimension of $Z$.

\section{Experiments}\label{Sec_Experiments}
We implement all our methods in R. Our implementations are based on the functions from the package \texttt{GeneralisedCovarianceMeasure}, see \cite{GCMPackage}. Our code is available as the R-package \texttt{weightedGCM} on CRAN.

\subsection{Detailed Comparison}
We have seen that for the introductory example in Section \ref{Sec_ConcreteExample}, the power of the GCM heavily depends on the function introducing the dependence, whereas WGCM.est and WGCM.fix are more stable in this respect. To investigate the observed effect more systematically, we consider a family of functions $h_{b_1,b_2}(t)$ indexed by two parameters $b_1,b_2\in [0,1]$, where $b_1$ is a parameter for symmetry and $b_2$ is a parameter for wiggliness. Define
\begin{align*}
h_1(t,b_1,b_2)&= (b_1\cos(3 b_2 t)+(1-b_1)\sin(3 b_2 t))\exp(-t^2/2),\\
h_2(t,b_1)&=0.3 (b_1 |t|+(1-b_1) t)
\end{align*}
and let
$$h_{b_1,b_2}(t)=(1-b_2) h_2(t,b_1)-b_2 h_1(t,b_1,b_2).$$
Plots of the functions $h_{b_1,b_2}$ for various values of $b_1$ and $b_2$ can be found in Figure \ref{fig_FunSymmWigg}. 

\begin{figure}
	\begin{center}
	\includegraphics[width=0.95\textwidth]{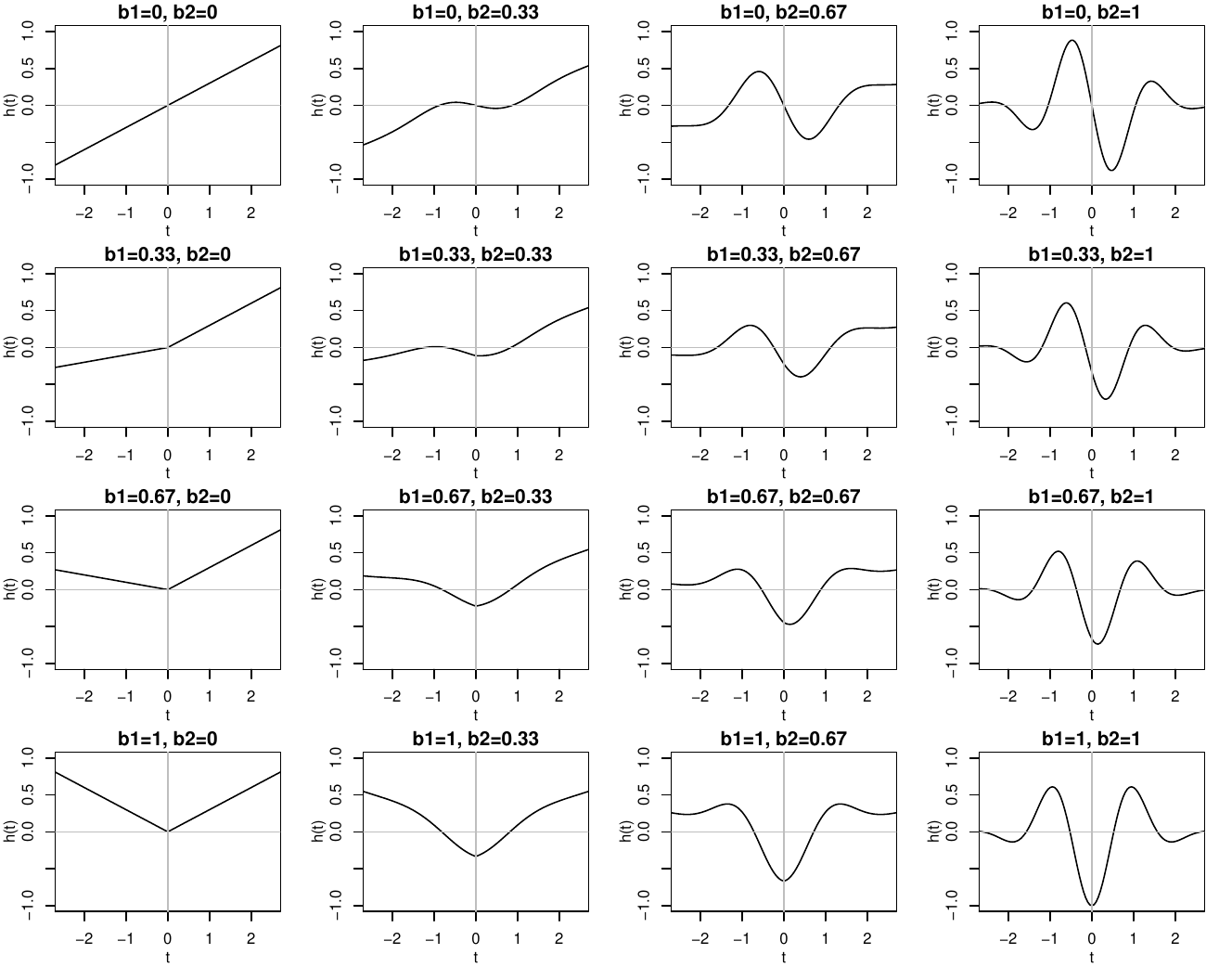}
	\end{center}
	\caption{Plots of $h_{b_1,b_2}$ for $b_1,b_2\in \left\{0,\frac{1}{3},\frac{2}{3},1\right\}$.}
	\label{fig_FunSymmWigg}
\end{figure}

\subsubsection{Null Hypothesis}\label{SubSec_ExpH0}
We first look at the level of the tests under the null hypothesis in the following three settings, which are similar to Section 5.2 in \cite{ShahPetersCondInd}:
\begin{itemize}[leftmargin=30mm]
	\item[(1D)] $Z\sim \mathcal N(0,1),\, X= h_{b_1, b_2}(Z)+0.3\mathcal N(0,1),\, Y=h_{b_1, b_2}(Z)+0.3 \mathcal N(0,1)$;
	\item[(10D.Add)] $Z_1,\ldots, Z_{10} \sim \mathcal N(0,1), \, X=h_{b_1, b_2}(Z_1)- h_{b_1, b_2}(Z_2)+0.3\mathcal N(0,1), \, Y=h_{b_1, b_2}(Z_1)+h_{b_1, b_2}(Z_2)+0.3\mathcal N(0,1)$;
	\item[(10D.NonAdd)]  $Z_1,\ldots, Z_{10} \sim \mathcal N(0,1), \, X=\sign\left(h_{b_1, b_2}(Z_1)+ h_{b_1, b_2}(Z_2)\right)+0.3\mathcal N(0,1), \, Y=\sign\left(h_{b_1, b_2}(Z_1)- h_{b_1, b_2}(Z_2)\right)+0.3\mathcal N(0,1)$; 
\end{itemize}

For every combination of $b_1, b_2\in \left\{0,\frac{1}{3},\frac{2}{3}, 1\right\}$, we simulate $500$ data sets with $n$ samples for each setting and perform the following tests:

\begin{itemize}[leftmargin=30mm]
	\item[(GCM)] The (unweighted) GCM by \citet{ShahPetersCondInd}.
	\item[(WGCM.est)]The WGCM with one single estimated weight function, where $30 \%$ of the samples are used to estimate the weight function.
	\item[(WGCM.fix)]  The WGCM with fixed weight functions
	$$w_{j,l}(z)= \sign(z_j-a_{j,l}), \, j=1,\ldots, d_Z, \, l=1,\ldots, k_0,$$
	where $a_{j,l}$ is the empirical $\frac{l}{k_0+1}$-quantile of $Z_j$. Additionally, we take the constant weight function $w_0(z)=1$. We will use $k_0=7$. This means we have a total of $8$ weight functions in the setting (1D) and $71$ weight functions in the settings (10D.Add) and (10D.NonAdd).
\end{itemize}

We perform all three tests both with regression splines using \texttt{gam} from the R package \texttt{mgcv} (see \citealp{MgcvPackage}) and with boosted regression trees using the package \texttt{xgboost} (see \citealp{ChenXgboostPaper} and \citealp{ChenXgboostPackage}). For WGCM.est, we always use the same regression method both for step 1 and step 2 of Method \ref{meth_WGCMEst} and for the calculation of the test statistic. Plots of the rejection rates at level $\alpha=0.05$ can be found in Figure \ref{fig_RejRatesH0}. For each combination of setting and method, we plot the rejection rate for all 16 combinations of $b_1$ and $b_2$. The lower horizontal gray line denotes the individual one-sided test region at level $0.05$ for each dot based on a $Bin(500,0.95)$ distribution. The upper horizontal gray line denotes the joint one-sided test region at level $0.05$ for each group of $16$ dots based on $16$ i.i.d. $Bin(500, 0.95)$ variables.

We see that with a sample size of $n=400$ and in the settings (1d) and (10dAdd), all the methods seem to perform reasonably well in the sense that the rejection rates are not significantly larger than $0.05$. In the setting (10dNonAdd), the methods based on \texttt{gam} seem to reject too often. This was to be expected since $\texttt{gam}$ assumes an additive structure. For lower sample size, some of the methods seem to reject too often also in the settings (1d) and (10dAdd). The guarantees for the level of the tests heavily rely on the quality of the approximation of $\Ex[X|Z]$ and $\Ex[Y|Z]$. Therefore, we suspect that the large rejection rates in some settings with $n=100$ and $n=200$ are due to the fact that some of the functions used in the simulation are too complex to be well approximated with the smaller sample sizes.

\begin{figure}
	\begin{center}
	\includegraphics[width=0.99\textwidth]{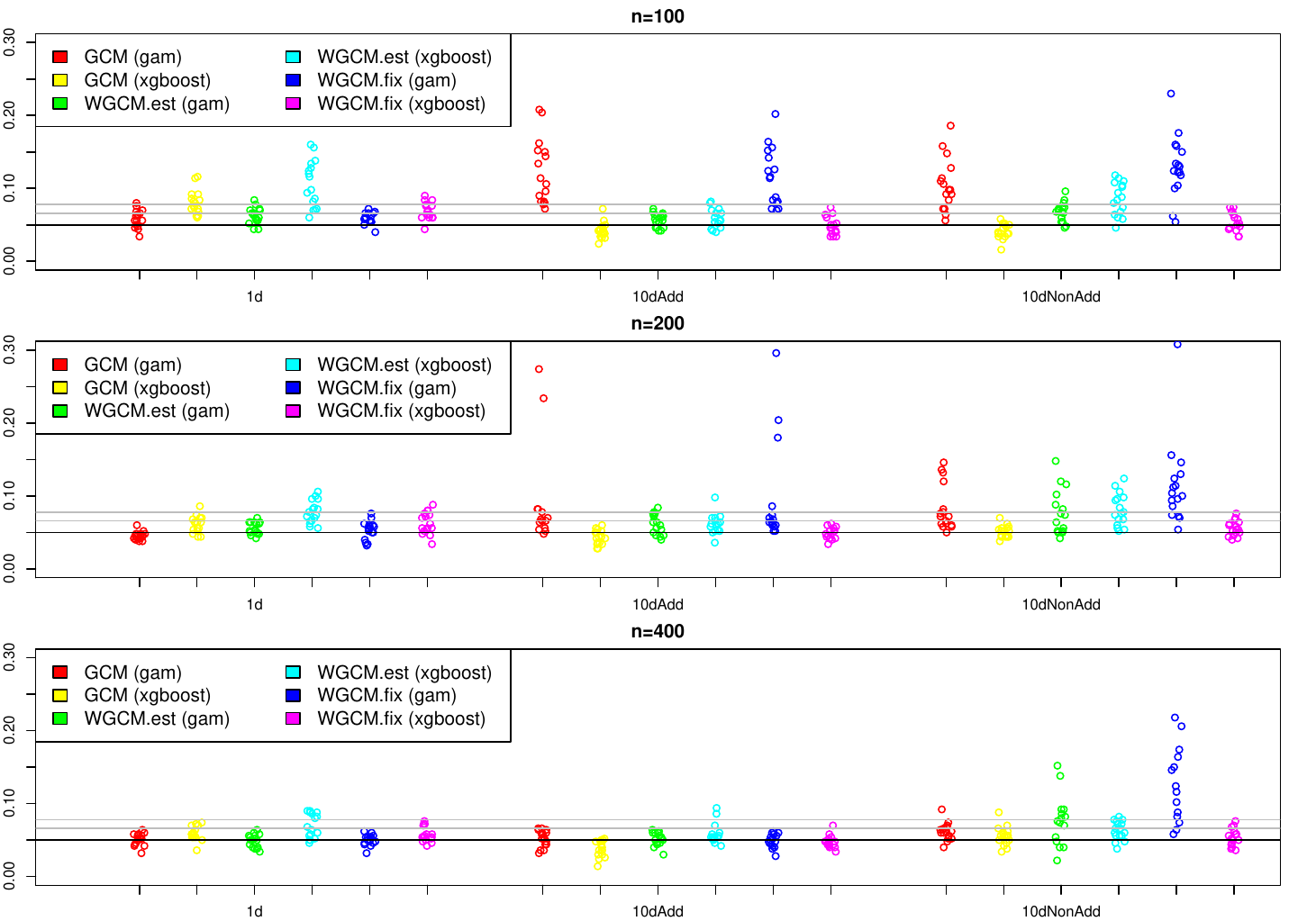}
	\end{center}
	\caption{Rejection rates under the null hypothesis. For each combination of setting and method, every point corresponds to one of $16$ combinations of $b_1, b_2\in \{0, \frac{1}{3}, \frac{2}{3}, 1\}$. The rejection rates are calculated using  500 independent simulations. The two horizontal gray lines are at $0.066$ and $0.078$ and give individual and joint one-sided test regions at level $0.05$ for the null hypothesis "the rejection rate of the test is less than or equal to $0.05$": The value $0.066$ is the $0.95$-quantile of $B/500$, where $B\sim {Bin}(500, 0.05)$ distribution, so the line denotes the test region of an individual dot. The value $0.078$ is the $0.95$-quantile of $\max(B_1,\ldots, B_{16})/500$, where $B_1,\ldots, B_{16}$ are i.i.d. ${Bin}(500, 0.05)$, so the line denotes the joint test region for each group of $16$ dots.}
	\label{fig_RejRatesH0}
\end{figure}

\subsubsection{Alternative Hypothesis}
For the alternative, we consider the same settings (1D), (10D.Add), (10D.NonAdd), but modify them by adding $h_{c_1, c_2}(X)$ to $Y$ for some $c_1, c_2 \in [0,1]$. We simulate 100 data sets for every combination of $b_1, b_2, c_1, c_2\in \{0, \frac{1}{2}, 1\}$ and $n \in \{100, 200, 400\}$ and calculate the rejection rates of the methods GCM, WGCM.est and WGCM.fix both with \texttt{gam} and \texttt{xgboost}.

A significant difference to the simulations in Section 5.2 of \citet{ShahPetersCondInd} is that the function $h_{c_1, c_2}(X)$ introducing the dependence is not just a linear function, but is also varied.

For each setting (1D), (10D.Add) and (10D.NonAdd), both regression methods \texttt{gam} and \texttt{xgboost} and sample sizes $n\in\{100,200, 400\}$, we plot the rejection rates of the three methods GCM, WGCM.est and WGCM.fix against each other. Thus, each subplot consists of $9\cdot 9 =81$ points, each corresponding to one combination of $b_1$, $b_2$, $c_1$ and $c_2$, see Figures \ref{fig_RRHA_1D}, \ref{fig_RRHA_10DAdd} and \ref{fig_RRHA_10DNonAdd}.

Let us first look at the first, second, fourth and fifth columns of the plots. These plot the rejection rates of GCM against the rejection rates of WGCM.est and WGCM.fix, respectively. We see that the behavior of the GCM and the two variants of the WGCM is asymmetric. The part on the bottom right of the corresponding plots is free, indicating that there are no situations where the GCM has a very high and the WGCM a very low rejection rate. In contrast, we see situations where the WGCM has a high, but the GCM a low rejection rate (points in the top left). The effect gets more pronounced for larger sample size. Nevertheless, there are also many points below the diagonal connecting $(0,0)$ and $(1,1)$. These indicate situations, where the GCM works better than the WGCM. To summarise, we see that indeed, the WGCM enlarges the space of alternatives against which the test has power. This comes at the cost of having less power in situations where the GCM already works well. In the setting (10D.NonAdd), the majority of the points lies below the diagonal in the plots comparing GCM to one of WGCM.fix and WGCM.est. As mentioned and observed in the simulations under the null hypothesis, we should not put too much trust in the results of (10D.NonAdd) with regression method \texttt{gam}. For the results using \texttt{xgboost}, it may also be possible that the picture would look more similar to the situations (1D) and (10D.Add) for larger sample size.

The third and sixth column of the plots compare the rejection rates of WGCM.est and WGCM.fix. We observe that WGCM.fix seems to perform better than WGCM.est, where the effect is more pronounced for (10D.Add) and (10D.NonAdd) than for (1D). However, by changing some parameters of the methods, for example the fraction of the samples used to estimate the weight function in WGCM.est and the number and type of weight functions for WGCM.fix, the picture could possibly look different.

\begin{figure}
	\begin{center}
	\includegraphics[width=0.99\textwidth]{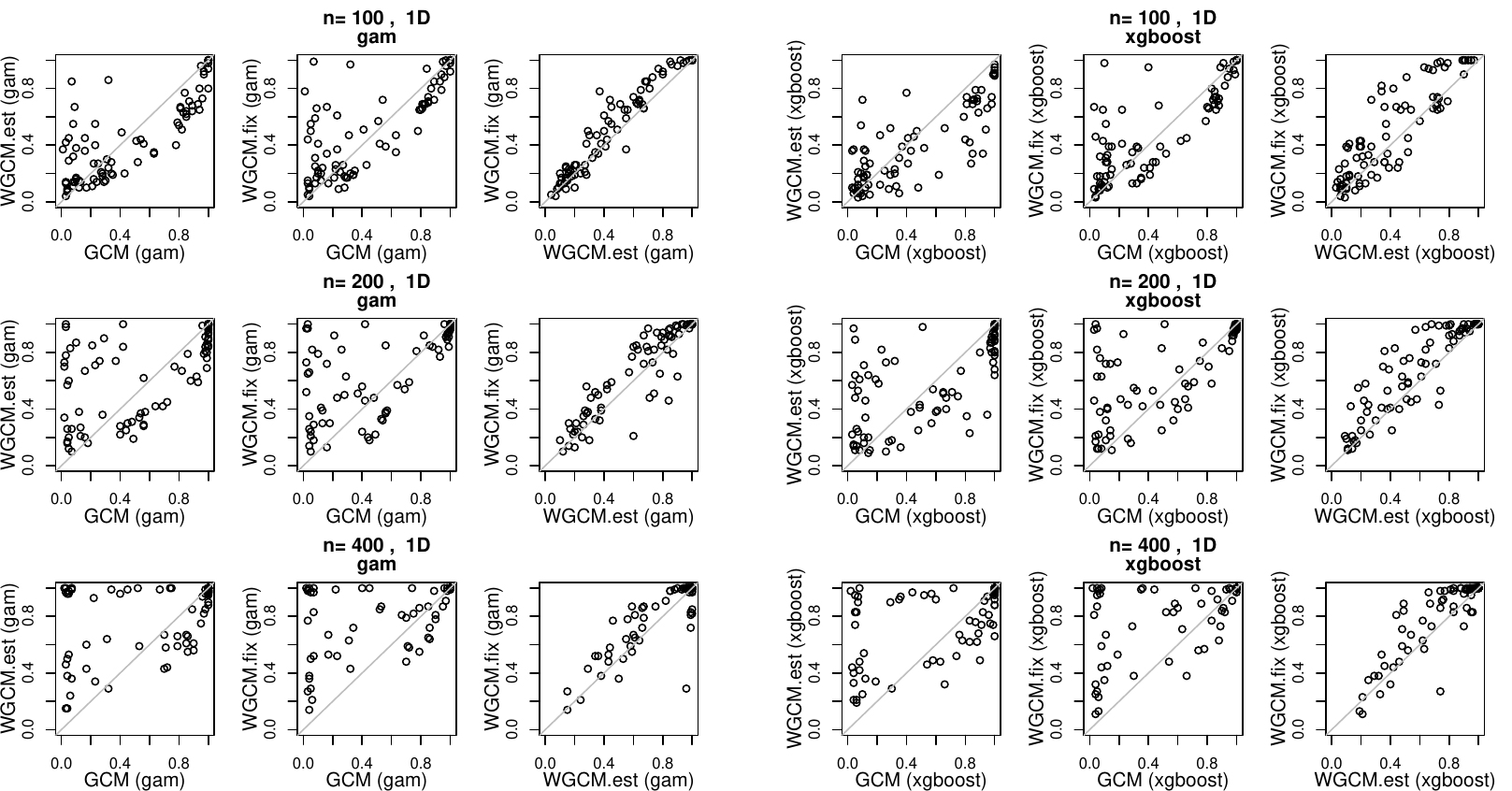}
	\end{center}
	\caption{Rejection rates under the alternative hypothesis in situation (1D). The rejection rates of GCM, WGCM.est and WGCM.fix are plotted against each other for regression methods \texttt{gam} (left) and \texttt{xgboost} (right) and $n\in\{100,200,400\}$. Each subplot consists of $9\cdot 9=81$ points, each corresponding to one combination of $b_1, b_2, c_1, c_1\in \{0,0.5, 1\}$.}
	\label{fig_RRHA_1D}
\end{figure}

\begin{figure}
	\begin{center}
	\includegraphics[width=0.99\textwidth]{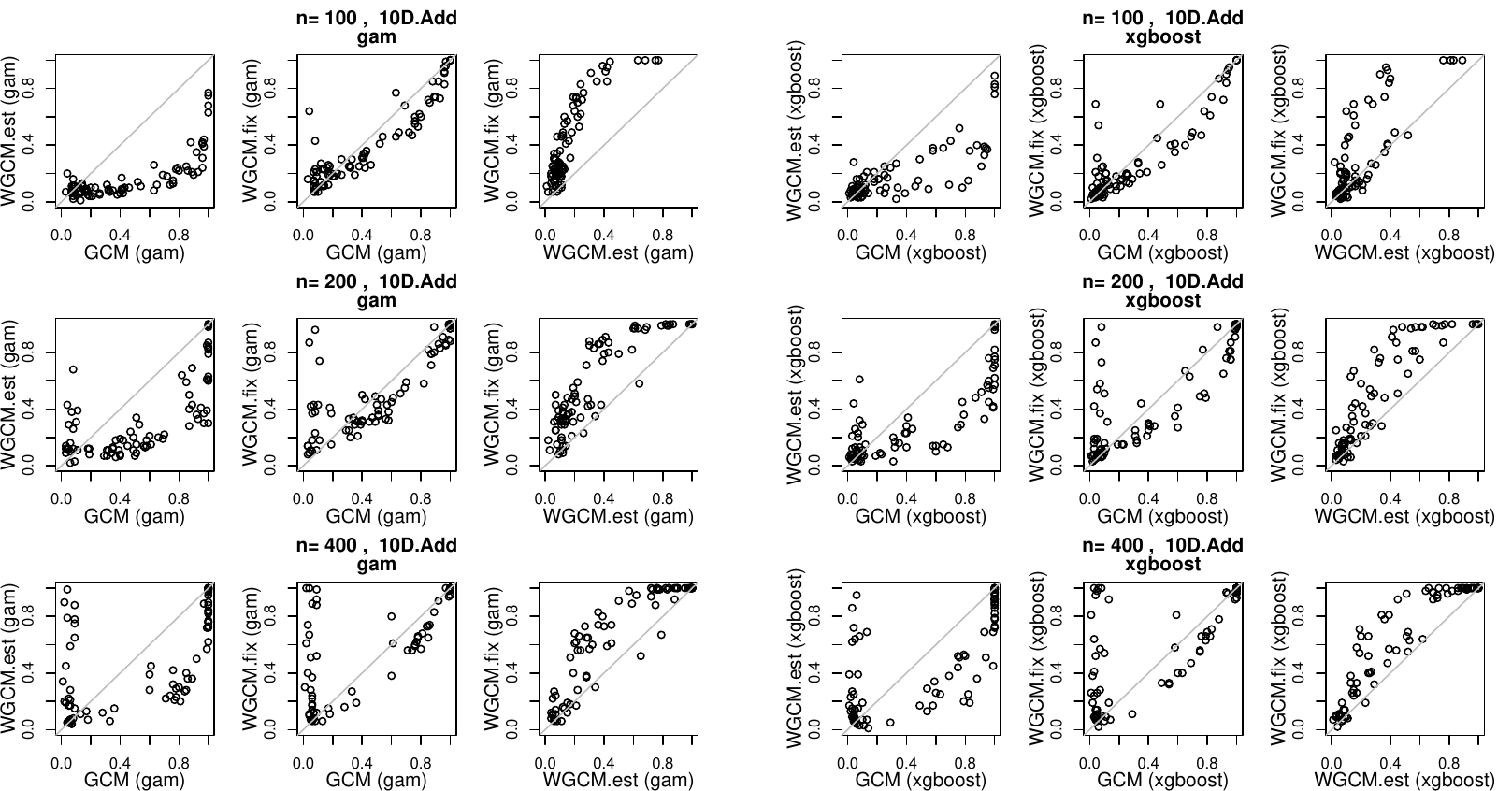}
	\end{center}
	\caption{Rejection rates under the alternative hypothesis in situation (10D.Add). The rejection rates of GCM, WGCM.est and WGCM.fix are plotted against each other for regression methods \texttt{gam} (left) and \texttt{xgboost} (right) and $n\in\{100,200,400\}$. Each subplot consists of $9\cdot 9=81$ points, each corresponding to one combination of $b_1, b_2, c_1, c_1\in \{0,0.5, 1\}$.}
	\label{fig_RRHA_10DAdd}
\end{figure}

\begin{figure}
	\begin{center}
	\includegraphics[width=0.99\textwidth]{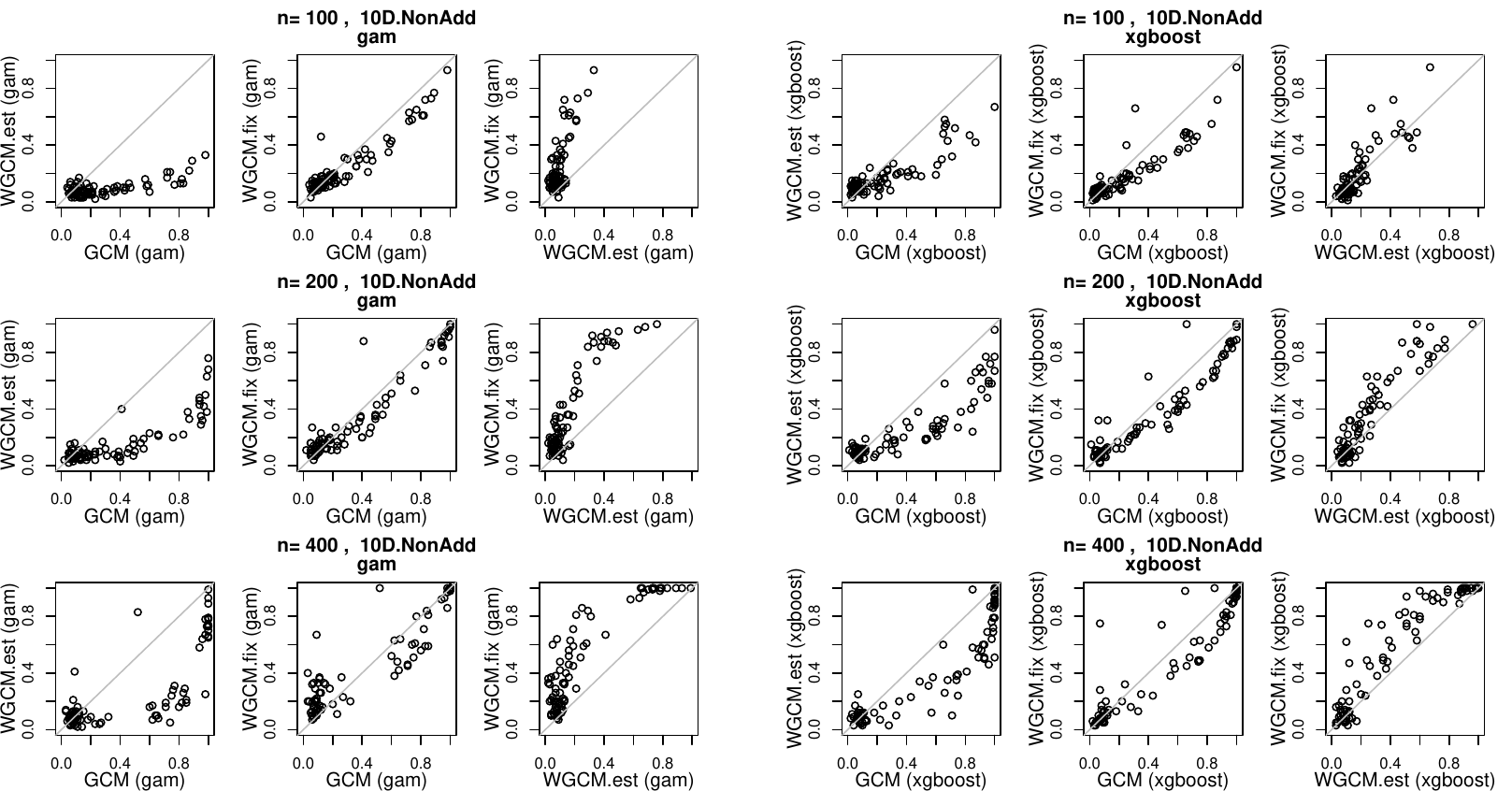}
	\end{center}
	\caption{Rejection rates under the alternative hypothesis in situation (10D.NonAdd). The rejection rates of GCM, WGCM.est and WGCM.fix are plotted against each other for regression methods \texttt{gam} (left) and \texttt{xgboost} (right) and $n\in\{100,200,400\}$. Each subplot consists of $9\cdot 9=81$ points, each corresponding to one combination of $b_1, b_2, c_1, c_1\in \{0,0.5, 1\}$.}
	\label{fig_RRHA_10DNonAdd}
\end{figure}

\subsection{Variable Selection and Importance}
In this section, we briefly sketch how conditional independence tests allow to perform variable selection tasks.

Suppose we have a response variable $Y$ and predictors $X_j, j=1,\ldots, d$, where all random variables take values in $\mathbb R$. For all $j=1, \ldots, d$, we can test
$$ H_0: X_j \indep Y|\mathbf X_{-j}.$$
We expect the corresponding $p$-value to be small if $X_j$ yields additional information for predicting $Y$ that is not contained in $\mathbf X_{-j}$. This can be seen as a generalization of the individual $t$-tests in linear regression. We can then look, for which $j$ the variable $X_j$ is significant after a multiple testing correction. This approach is for example described in \citet{WatsonWrightTesting}, where they compare it to a new method.

The following example is taken from \cite{AzadkiaSimple}, where it appears as Example 7.4, with the difference that we only use a 50-dimensional $X$ and $n=500$ samples instead of a 1000-dimensional $X$ with $n=2000$ samples. We use GCM, WGCM.est and WGCM.fix with \texttt{xgboost} for the regressions.

\begin{example}
Let $X_1,\ldots, X_{50}$ be i.i.d. $\mathcal N(0,1)$ and let $Y=X_1X_2+X_1-X_3+\epsilon$, where $\epsilon \sim \mathcal N(0,1)$ independent of $X_1, \ldots, X_{50}$. We simulate $100$ data sets with a sample size of $n=500$. For each data set and all $j=1,\ldots, 50$, we calculate a $p$-value for $ H_0: X_j \indep Y|\mathbf X_{-j}$ using the three tests. After a multiple testing correction using Holm's procedure \citep{HolmSeqRej}, we observe that GCM never finds the correct set of predictors at significance level $\alpha=0.05$, but most often, it just finds $X_1$ and $X_3$ as significant variables. WGCM.est finds the correct set of predictors in 86 out of 100 cases and WGCM.fix even finds the correct predictors in 99 of the 100 cases. We use the same type of weight functions with $k_0=7$ as in Section \ref{SubSec_ExpH0} for WGCM.fix and for WGCM.est, we use $30\%$ of the samples to estimate the weight functions.
\end{example}

Hence, this example illustrates that the two variants of the WGCM can find more dependencies than the GCM. In the following, we also look at some real data sets.

\subsubsection{Boston Housing Data}\label{Sec_Boston}
As a first example, we analyse the Boston housing data, see \citet{HarrisonRubinfeldBoston}. Among the set of 13 predictors, we want to find the most relevant ones to predict the target variable \texttt{medv}, which is the median value of owner-occupied homes. Plots of the $p$-values after multiple testing correction using Holm's method can be found in Figure \ref{fig_PValBoston}.
We performed all regressions using $\texttt{xgboost}$ for GCM, WGCM.est and WGCM.fix

\begin{figure}
	\begin{center}
	\includegraphics[width=0.8\textwidth]{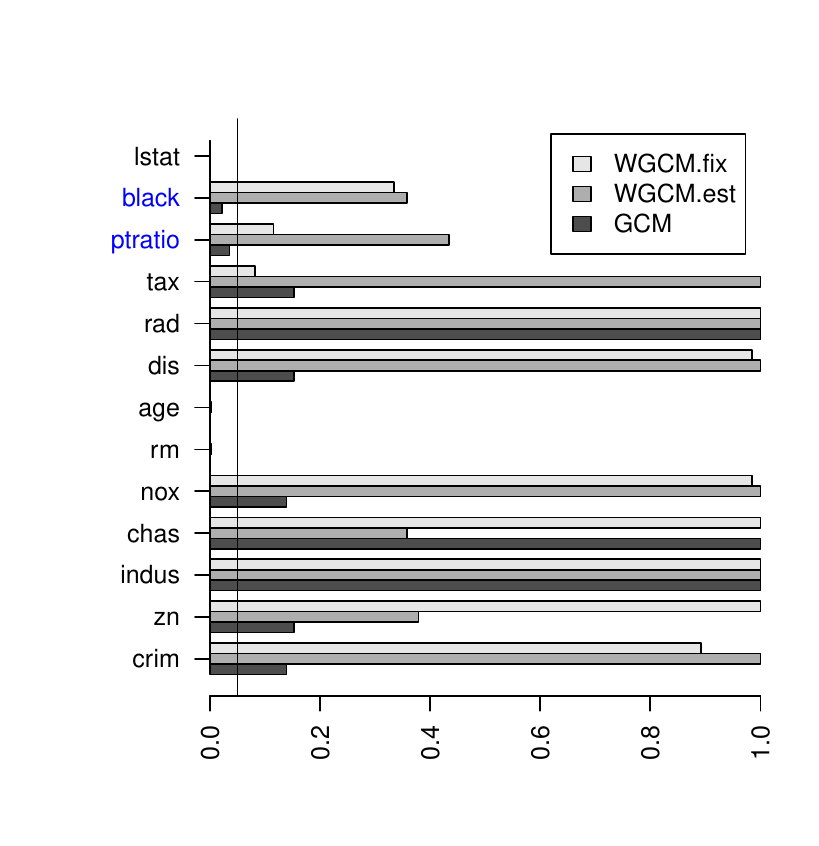}
	\end{center}
	\caption{The $p$-values for the Boston housing data set for the prediction of the variable \texttt{medv} adjusted for multiple testing using Holm. For the variables denoted in blue, GCM shows a significant effect at level $\alpha=0.05$, but WGCM.fix and WGCM.est do not.}
	\label{fig_PValBoston}
\end{figure}

We see that in this case, the GCM finds $5$ significant variables at significance level $\alpha=0.05$, whereas WGCM.fix and WGCM.est only find $3$. Hence, this is an example where the original GCM performs moderately better (in terms of power) than the new variants.

\subsubsection{Online News Popularity Data}\label{sec_ONP}
We analyse the online news popularity data set, see \cite{PaperONP}. The data can be obtained from the UCI Machine Learning Repository, see \cite{UCIRepository}. Removing missing values, we have $39644$ observations of $58$ predictors and one target variable. Each observations corresponds to one article published by \url{www.mashable.com}. The target variable is the number of shares of the article, whereas the predictors are various features of the article, ranging from the number of words to sentiment scores. For the analysis, we take the log-transform of the target variable. Plots of the $p$-values for the three methods can be found in Figure \ref{fig_PValONP}.
We performed the regressions using $\texttt{xgboost}$ for GCM, WGCM.est and WGCM.est.

\begin{figure}
	\begin{center}
	\includegraphics[height=0.89\textheight]{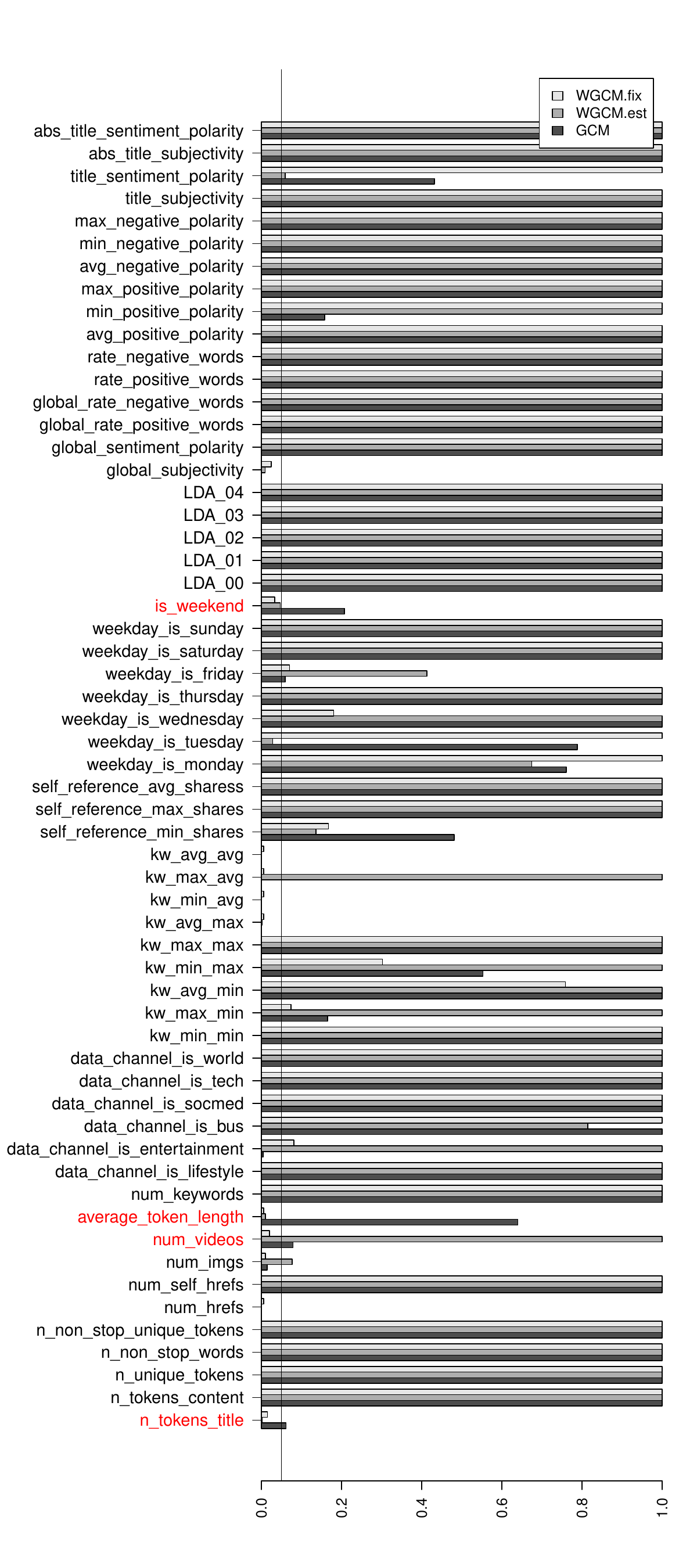}
	\end{center}
	\caption{$p$-values for the online news popularity data set for the prediction of the variable \texttt{shares} adjusted for multiple testing using Holm. For the variables denoted in red, WGCM.fix shows a significant effect at level $\alpha=0.05$, but GCM does not.}
	\label{fig_PValONP}
\end{figure}

We see that in this case, the GCM finds $8$ significant variables at significance level $\alpha=0.05$, whereas WGCM.est finds $9$ significant variables and WGCM.fix finds $11$ significant variables. Hence both versions of the WGCM perform slightly better (in terms of power) than the original GCM.

\subsubsection{Wave Energy Converters Data}
We look at the wave energy converters data set, available at the UCI Machine Learning Repository \citep{UCIRepository}. The data set displays the (simulated) power output for different configurations of wave energy converters in different real wave scenarios, see \citet{PaperWEC}. We restrict ourselves to the data of Tasmania. We have $32$ predictor variables, which consist of the $x$ and $y$-coordinates of $16$ wave energy converters forming a wave farm. The target is the total power output of the farm. We randomly sample $3000$ of the $72000$ configurations in the data set. By symmetry, it seems reasonable that either no or all predictor variables are significant. In fact, with Holm's method to adjust for multiple testing, GCM does not find any significant predictor at significance level $\alpha=0.05$, whereas WGCM.fix considers all $32$ predictors significant and WGCM.est finds $11$ significant predictors. Plots of the $p$-values for the three methods can be found in Figure \ref{fig_PValTWEC}. In this case, the two WGCM methods are clearly superior to GCM (in terms of power).

\begin{figure}
	\begin{center}
	\includegraphics[width=0.7\textwidth]{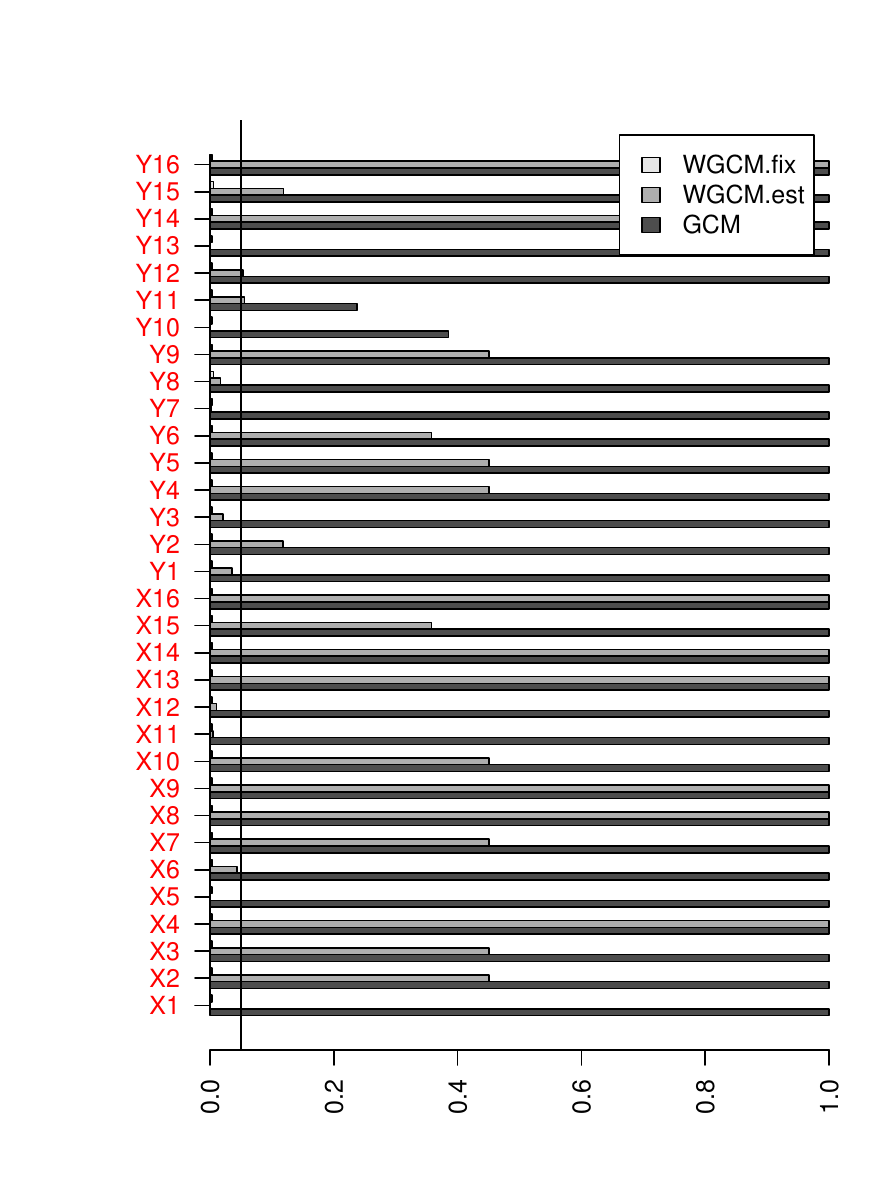}
	\end{center}
	\caption{$p$-values of the predictor variables for the Tasmanian wave energy converters data set for the prediction of the total power output adjusted for multiple testing using Holm. At level $\alpha=0.05$, WGCM.fix finds a significant effect for all $32$ variables and WGCM.est for $11$ variables. GCM does not find any significant effect.}
	\label{fig_PValTWEC}
\end{figure}

\section{Conclusion}\label{Sec_Conclusion}
We have introduced the \textit{weighted generalised covariance measure} (WGCM) as a new test for conditional independence and we provide an implementation in the R-package \texttt{weightedGCM}, which is available on CRAN. The WGCM is based on the \textit{generalised covariance measure} (GCM) by \cite{ShahPetersCondInd}. We gave two versions of the WGCM in the setting of univariate $X$ and $Y$. Their generalisations to the setting of multivariate $X$ and $Y$ can be found in Appendix \ref{sec_MultWGCM}. To give guarantees for the correctness of the tests under appropriate conditions, we could benefit from the work by \cite{ShahPetersCondInd}. We prove that WGCM.est and WGCM.fix have full asymptotic power against a broader class of alternatives than GCM. Finally, we compared our methods to the original GCM using simulation and on real data sets. We have seen for finite samples that our approach allows to enlarge the set of alternatives against which the test has power. This comes at the cost of having reduced power in settings where the GCM already performs well. An application in a variable selection task on real data sets confirmed that it depends on the data at hand which method is to be preferred.
If the sample size is small and the form of the dependence is simple, the GCM will typically be the better choice. However, if the sample size is moderately large, choosing WGCM.fix or WGCM.est is typically beneficial.

\subsection{Practical Issues}
Choosing the optimal test among the GCM and the two versions of the WGCM remains a challenging task in practice, though it is always possible to perform several tests and use a multiple testing correction. It is worth mentioning that in principle, one could combine WGCM.est and WGCM.fix by including an estimated weight function together with the fixed weight functions of WGCM.fix. Such a case is implicitly covered by Theorem \ref{thm_MultEstWMultDim} in Appendix \ref{Sec_mWGCMEst}. However, including an estimated weight function has the disadvantage that also the analysis of the fixed weight functions can only be done on half of the sample. For this reason, we think that performing the two tests WGCM.est and WGCM.fix individually and using Bonferroni correction is the better choice for combining the two tests.

Moreover, the randomness of the sample splitting for WGCM.est leads to the question how stable the test is with respect to this random split. We investigate this in Appendix \ref{App_Stability}. Methods to aggregate $p$-values obtained from multiple sample splits are for example treated in \cite{MeiMeiBuPValHDReg} and in \cite{DiCiccioMultDatSplit}. The $p$-values obtained using such methods are more stable and provably controlling type I error. However, it is not clear how they affect the power of the test, even though the methods were found to improve power as well in other problems, see \cite{MeiMeiBuPValHDReg}.

\subsection{Outlook}
There are many other open questions remaining. On the theoretical side, it would be desirable to give more concise results about the power properties of WGCM.est and WGCM.fix. For WGCM.est, Corollary \ref{cor_PowerWGCMest} relies on the consistency of Method \ref{meth_WGCMEst} to estimate $w(z)=\sign(\Ex_P[\epsilon\xi|Z=z])$, which is not straightforward to verify and depends on the regression method used. For WGCM.fix, it would be interesting to see, if there is a set of fixed weight functions with better properties than the ones described in Section \ref{sec_ChoiceWeightFun}.

On the practical side, there are many parameters that can be varied and whose effects could be inspected more closely. For WGCM.est, it would be desirable to have guidelines for the fraction of the data to be used to estimate the weight function. For WGCM.fix, the choice of the number of fixed weight functions is unknown for optimal power. However, our current default choice used in the empirical analysis seems to work reasonably well.

Also Method \ref{meth_WGCMEst} to estimate the weight function could be investigated further. We do not claim our procedure to estimate $w(z)=\sign(\Ex_P[\epsilon\xi|Z=z])$ to be optimal. It is simply a straightforward way how the conditions of Theorem \ref{thm_WGCM1DEst} can be satisfied. However, there might be more powerful procedures. It could be worth investigating if a smoothed version of the sign is beneficial. This would have the advantage of giving less weight to values of $z$ for which the estimate of $\Ex_P[\epsilon\xi|Z=z]$ is close to $0$ and which are thus more likely to obtain the wrong sign. However, the smoothed version of the sign still has to be scaled in such a way that the conditions for the correct null distribution are satisfied.


\acks{We are grateful to Rajen D. Shah for pointing out that WGCM.est has power against all alternatives for binary $X$ and $Y$. We thank Rajen D. Shah and Jonas Peters for answering our questions about their paper \cite{ShahPetersCondInd}. We also thank three referees and an action editor for helpful comments. The research of Julia H\"orrmann is supported by ETH Foundations of Data Science. Peter B\"uhlmann received funding from the European Research Council (ERC) under the European Union's Horizon 2020 research and innovation programme (grant agreement No. 786461).}


\newpage

\appendix
\section{The Multivariate WGCM}\label{sec_MultWGCM}
In this section, we show how the methods of Section \ref{SubSec_WGCM1dMultFix} can also be applied in the case of multivariate $X$ and $Y$. This can be done both in the context of fixed and of estimated weight functions.

\subsection{Multivariate WGCM With Fixed Weight Functions (mWGCM.fix)}\label{sec_MultFixMultDim}
The procedure is again very similar to Section 3.2 of \citet{ShahPetersCondInd}.
The idea in the case of multivariate $X$ and $Y$ stays the same, but we calculate the test statistic for every pair of $X_j$ and $Y_l$.

For any distribution $P\in\mathcal E_0$ and all $j=1,\ldots, d_X$ and $l=1,\ldots, d_Y$, define
$$f_{P,j}(z)=\Ex_P\left[X_j|Z=z\right]\quad \text{and}\quad g_{P,l}(z)=\Ex_P\left[Y_l|Z=z\right].$$
Then, we can write
$$X_j=f_{P,j}(Z)+\epsilon_{P,j}\quad \text{and} \quad Y_l=g_{P,l}(Z)+\xi_{P,l}.$$
Let $\epsilon_{P,ij}=x_{ij}-f_{P,j}(z_i)$ and $\xi_{P,il}=y_{il}-g_{P,l}(z_i)$.

Let $\hat f_{j}^{(n)}$ and $\hat g_{l}^{(n)}$ be estimates of $f_{P,j}$ and $g_{P,l}$, obtained by regression of $\mathbf{X}_j^{(n)}$  and $\mathbf{Y}_l^{(n)}$, on $\mathbf{Z}^{(n)}$. Note that $\mathbf{X}_j^{(n)}=(x_{1j},\ldots, x_{nj})^T$ is the $j$th column of the data matrix $\mathbf X\upind n$, or equivalently the column of samples from the random variable $X_j$.

For each $j,l$ let $K(j,l)\in\mathbb N$ and let $\left\{w_{jlk}:j=1,\ldots, d_X, l=1,\ldots, d_Y, k=1,\ldots, K(j,l)\right\}$ be functions from $\mathbb R^{d_Z}\to \mathbb R$. We allow $K(j,l)$ to grow with $n$. We will work under the assumption that the functions $w_{jlk}$ are uniformly bounded for all $j$, $l$ and $k$.
Let
$$\mathbf K=\mathbf K(n)=\sum_{j=1}^{d_X}\sum_{l=1}^{d_Y} K(j,l).$$
Let $\mathbf R_{jlk} \in \mathbb R^n$ be the vector of products of the residuals corresponding to $X_j$ and $Y_l$ weighted by $w_{jlk}$, that is,
$$\mathbf R_{jlk}=\left(\begin{array}{c}(x_{1j}-\hat f_j(z_1)) (y_{1l}-\hat g_l(z_1))w_{jlk}(z_1)\\ \vdots\\ (x_{nj}-\hat f_j(z_n)) (y_{nl}-\hat g_l(z_n))w_{jlk}(z_n)\end{array}\right).$$
Let $T_{jlk}\upind n$ be the test statistic of the WGCM based on the vector $\mathbf R_{jlk}$, that is,
\begin{equation}\label{eq_DefStatMultFixWMultDim}
T_{jlk}\upind n=\frac{\sqrt n \bar{\mathbf R}_{jlk}}{\left(\frac{1}{n}\|\mathbf R_{jlk}\|_2^2-\bar{\mathbf R}_{jlk}^2\right)^{1/2}}\eqqcolon \frac{\tau_{N,{jlk}}\upind n}{\tau_{D,{jlk}}\upind n},
\end{equation}
with $\bar{\mathbf R}_{jlk}$ being the sample average of the coordinates of ${\mathbf R_{jlk}}$. Finally, let $\mathbf T\upind n=\left(T_{jlk}\upind n\right)_{jlk}\in\mathbb R^{\mathbf K}$ be the vector of all test statistics. We consider the maximum absolute value of the vector $\mathbf T\upind n$ as a test statistic,
$$S_n=\max_{j=1,\ldots, d_X,\, l=1,\ldots, d_Y,\, k=1,\ldots, K(j,l)}|T_{jlk}\upind n|.$$
In summary, everything works similarly to the univariate case with the added complication of having expressions with three subscripts instead of one. Define $\hat \Sigma\in \mathbb R^{\mathbf K\times \mathbf K}$ by
$$\hat \Sigma_{jlk,j'l'k'}=\frac{\frac{1}{n}\mathbf R_{jlk}^T \mathbf R_{j'l'k'}-\bar{\mathbf R}_{jlk}\bar {\mathbf R}_{j'l'k'}}{\left(\frac{1}{n}\|\mathbf R_{jlk}\|_2^2-\bar{\mathbf R}_{jlk}^2\right)^{1/2}\left(\frac{1}{n}\|\mathbf R_{j'l'k'}\|_2^2-\bar{\mathbf R}_{j'l'k'}^2\right)^{1/2}}.$$
Let $\hat{\mathbf T}\upind n= \left(\hat T_{jlk}\upind n\right)_{jlk}\in\mathbb R^{\mathbf K}$ have multivariate normal distribution with covariance $\hat \Sigma$ and mean $0$ and let 
$$\hat S_n=\max_{j,l,k}|\hat T_{j,l,k}\upind n|.$$
Let $\hat G_n$ be the quantile function of $\hat S_n$ given $\hat \Sigma$. $\hat G_n$ is random, depends on the data and can be approximated by simulation. We need similar conditions to (A1a), (A1b) and (A2) for Theorem \ref{thm_MultFixedWOneDim}. Let
$$\sigma_{jl}^2=\sigma_{P,jl}^2=\Ex_P\left[\epsilon_j^2\xi_l^2\right].$$
Consider a sequence $(D_n)_{n\in\mathbb N}$ with $D_n\geq 1$.
\begin{itemize}[leftmargin=15mm]
	\item[(C1a)] $\max_{r=1,2}\Ex_P\left[\left|\frac{\epsilon_j\xi_l}{\sigma_{jl}}\right|^{2+r}/D_n^r\right]+\Ex_P\left[\exp\left(\left|\frac{\epsilon_j\xi_l}{\sigma_{jl}}\right|/D_n\right)\right]\leq 4$ for all $j=1,\ldots, d_X$, and $l=1,\ldots, d_Y$;
	\item[(C1b)] $\max_{r=1,2}\Ex_P\left[\left|\frac{\epsilon_j\xi_l}{\sigma_{jl}}\right|^{2+r}/D_n^{r/2}\right]+\Ex_P\left[\max_{j,l}\left|\frac{\epsilon_j\xi_l}{\sigma_{jl}}\right|^4/D_n^2\right]\leq 4$ for all $j=1,\ldots, d_X$ and $l=1,\ldots, d_Y$;
	\item[(C2)] $D_n^2\left(\log(\mathbf Kn)\right)^7/n\leq C n^{-c}$ for some constants $C,c>0$ that do not depend on $P\in\mathcal P$.
\end{itemize}

We obtain the analogue of Theorem \ref{thm_MultFixedWOneDim} (and also of Theorem 9 in \citealp{ShahPetersCondInd}). For $j=1,\ldots, d_X$ and $l=1,\ldots, d_Y$, let
\begin{align}
A_{f,j}=\frac{1}{n}\sum_{i=1}^n (f_{P,j}(z_i)-\hat f_j(z_i))^2, \label{eq_DefAfj}\\
A_{g,l}=\frac{1}{n}\sum_{i=1}^n (g_{P,l}(z_i)-\hat g_l(z_i))^2. \label{eq_DefAgl}
\end{align}

\begin{theorem}[mWGCM.fix]\label{thm_MultFixedWMultDim}
Let $\mathcal P\subset \mathcal P_0$. Assume that there exist $C,c> 0$ such that for all $n\in\mathbb N$ and $P\in\mathcal P$ there exists $D_n\geq 1$ such that either (C1a) and (C2) or (C1b) and (C2) hold. Furthermore, assume that there exist $C_1,c_1>0$ (independent of $n$) such that for all $j=1,\ldots, d_X$, $l=1,\ldots, d_Y$ and $k=1,\ldots, K(j,l)$ and for all $P\in\mathcal P$, we have $|w_{jlk}|\leq C_1$ and $\Ex_P\left[\epsilon_j^2\xi_l^2 w_{jlk}(Z)^2\right]\geq c_1 \sigma_{jl}^2$. Assume that 
\begin{align}
\max_{j,l}\frac{1}{\sigma_{jl}^2}A_{f, j} A_{g,l}=o_\mathcal P\left(n^{-1}\log(\mathbf K)^{-4}\right). \label{eq_MultFixedWMultDimCondA0}
\end{align}
Assume that there exist sequences $\left(\tau_{f,n}\right)_{n\in\mathbb N}$ and $\left(\tau_{g,n}\right)_{n\in\mathbb N}$ as well as positive real numbers $s_{g,jl}$, $t_{g,jl}$, $s_{f,jl}$ and  $t_{f,jl}$ possibly depending on $P\in\mathcal P$ such that for all $j=1,\ldots, d_X$, $l=1,\ldots, d_Y$
$$s_{f,jl}t_{f,jl}=\sigma_{jl},\quad s_{g,jl}t_{g,jl}=\sigma_{jl},$$
and such that
\begin{align}
\max_{i,j,l}|\epsilon_{P,ij}|/t_{g,jl}=O_\mathcal P(\tau_{g,n}),\quad \max_{j,l}A_{g, l}/s_{g,jl}^2=o_\mathcal P\left(\tau_{g,n}^{-2}\log(\mathbf K)^{-4}\right) \label{eq_MultFixedWMultDimCondC0}\\
\max_{i,j,l}|\xi_{P,il}|/t_{f,jl}=O_\mathcal P(\tau_{f,n}),\quad \max_{j,l}A_{f,j}/s_{f,jl}^2=o_\mathcal P\left(\tau_{f,n}^{-2}\log(\mathbf K)^{-4}\right). \label{eq_MultFixedWMultDimCondD0}
\end{align}
Then,
$$\sup_{P\in\mathcal P}\sup_{\alpha\in (0,1)}|\Prob_P(S_n\leq \hat G_n(\alpha))-\alpha|\to 0.$$
\end{theorem}
\begin{remark}\phantomsection\label{rmk_ThmMultFixedWMultDim}
\begin{enumerate}
\item The idea behind the sequences  $s_{g,jl}$, $t_{g,jl}$, $s_{f,jl}$ and  $t_{f,jl}$ is that they allow for a more general scaling than the simplified setting in the next item. Note that by $X\indep Y|Z$, we have that $\sigma_{jl}^2=\Ex_P\left[u_{P,j}(Z) v_{P,l}(Z)\right]$, where $u_{P,j}(Z)=\Ex_P[\epsilon_{P,j}^2|Z]$ and $v_{P,l}(Z)=\Ex_P[\xi_{P,l}^2|Z]$. If for example $u_{P,j}(Z)$ and $v_{P,l}(Z)$ are a.s. constant equal to some $u_j$ and $v_l\in \mathbb R$, we can take $t_{g,jl}^2=s_{f,jl}^2=u_j$ and $t_{f,jl}^2=s_{g,jl}^2=v_l$. In this case, conditions (\ref{eq_MultFixedWMultDimCondC0}) and (\ref{eq_MultFixedWMultDimCondD0}) are just conditions on the errors $\epsilon_j$, $\xi_l$ scaled by their standard deviation and on $A_{f,j}$, $A_{g,l}$ scaled by the error variance.

\item Assume that there exists $c_2>0$ independent of $n$ such that for all $P\in\mathcal P$, all $j=1,\ldots, d_X$ and  $l=1,\ldots, d_Y$ we have 
\begin{equation}\label{eq_VarBoundBelow}
\sigma_{P,jl}=\Ex_P\left[\epsilon_j^2\xi_l^2 \right]\geq c_2.
\end{equation}
Then, we can replace $\epsilon_j\xi_l/\sigma_{jl}$ by $\epsilon_j\xi_l$ in conditions (C1a), (C1b) and (C2). Furthermore, we can replace the sequences $s_{g,jl}$, $t_{g,jl}$, $s_{f,jl}$ and $t_{f,jl}$ by $1$ in conditions (\ref{eq_MultFixedWMultDimCondC0}) and (\ref{eq_MultFixedWMultDimCondD0}) and we can replace condition (\ref{eq_MultFixedWMultDimCondA0}) by
$$\max_{j,l} A_{f,j} A_{g,l}= o_\mathcal P\left(n^{-1} \log(\mathbf K)^{-2}\right),$$
see also the next item.

If in addition to (\ref{eq_VarBoundBelow}), we also have that the errors $\epsilon_{P,j}$ and $\xi_{P,l}$ have sub-Gaussian distributions with parameters uniformly bounded for all $j,l$ by some constant independent of $P\in\mathcal P$, then using the same arguments as in Remark \ref{rmk_MultFixedWOneDim}, condition (C1a) can be satisfied with $D_n$ constant independent of $n$. Moreover, we have that $\max_{i,j}|\epsilon_{P,ij}|=O_\mathcal P\left(\sqrt{\log(n d_X)}\right)$. If for example both
$$\max_j A_{f,j},\max_l A_{g,l}=o_\mathcal P\left(\log(\mathbf K)^{-4}\min(n^{-1/2},\log(\mathbf K)^{-1})\right),$$
then the modified versions of conditions (\ref{eq_MultFixedWMultDimCondA0}), (\ref{eq_MultFixedWMultDimCondC0}) and  (\ref{eq_MultFixedWMultDimCondD0}) are all satisfied.

\item If there exists $C_3 >0$ such that for all $P\in\mathcal P$ and all $j,j'=1,\ldots, d_X$ and $l,l'=1,\ldots, d_Y$ we have
$$\sigma_{jl}\sigma_{j'l'}\geq C_3 \sigma_{jl'}\sigma_{j'l},$$
then inspection of the proof shows that condition (\ref{eq_MultFixedWMultDimCondA0}) can be replaced by
$$\max_{j,l}\frac{1}{\sigma_{jl}^2} A_{f,j} A_{g,l}= o_\mathcal P\left(n^{-1} \log(\mathbf K)^{-2}\right).$$
\end{enumerate}
\end{remark}

In analogy to Corollary \ref{cor_PowerWGCMfix}, we obtain the following result about the power of mWGCM.fix if $d_X$, $d_Y$ and all $K(j,l)$ for $j=1,\ldots, d_X$ and $l=1,\ldots, d_Y$ are fixed. For this, let $u_{P,j}(Z)=\Ex_P[\epsilon_{P,j}^2|Z]$ and $v_{P,l}(Z)=\Ex_P[\epsilon_{P,l}^2|Z]$ and define
\begin{align}
B_{f,jl}=\frac{1}{n}\sum_{i=1}^n (f_{P,j}(z_i)-\hat f_j(z_i))^2v_{P,l}(z_i), \label{eq_DefBfjl}\\
B_{g,jl}=\frac{1}{n}\sum_{i=1}^n (g_{P,l}(z_i)-\hat g_l(z_i))^2u_{P,j}(z_i). \label{eq_DefBgjl}
\end{align}

\begin{corollary}[mWGCM.fix]\label{cor_Power_mWGCMfix}
Let $P\in\mathcal E_0$. Let $A_{f,j}$, $A_{g,l}$, $B_{f,jl}$ and $B_{g,jl}$ be defined as in (\ref{eq_DefAfj}), (\ref{eq_DefAgl}), (\ref{eq_DefBfjl}) and (\ref{eq_DefBgjl}) with the difference that all $\hat f_j$ and $\hat g_l$ have been estimated on an auxiliary data set independent of $(\mathbf X^{(n)}, \mathbf Y^{( n)},\mathbf Z^{( n)})$. Let $d_X$, $d_Y$ and all $K(j,l)$ for $j=1,\ldots, d_X, \, l=1,\ldots, d_Y$ be fixed. Assume that there exists $C>0$ such that for all $z\in \mathbb R^{d_Z}$ and all $j,l,k$, we have $|w_{jlk}(z)|\leq C$. Assume that for all $j,l$ we have $A_{f,j} A_{g,l}=o_P(n^{-1})$, $B_{f,jl}=o_P(1)$ and $B_{g,jl}=o_P(1)$ as well as $\Ex_P[\epsilon_{P,j}^2\xi_{P,l}^2w_{jlk}(Z)^2]>0$ for all $k=1,\ldots, K(j,l)$ and $\Ex_P\left[\epsilon_{P,j}^2\xi_{P,l}^2\right]<\infty$. If there exists $j\in\{1,\ldots, d_X\},\, l\in\{1,\ldots, d_Y\}$ and $k\in\{1,\ldots, K(j,l)\}$ such that $\Ex_P[\epsilon_{P,j}\xi_{P,l}w_k(Z)]\neq 0$, then for all $M>0$,
$$\Prob_P(S_n \geq M)\to 1,$$
that is, mWGCM.fix with fixed $d_X$, $d_Y$ and fixed number of weight functions $K(j,l)$ for all $j$ and $l$ has asymptotic power $1$ against alternative $P$ for any significance level $\alpha\in (0,1)$.
\end{corollary}

\subsubsection{Choice of weight functions}
As the simplest extension of the considerations in Section \ref{sec_ChoiceWeightFun}, we propose the following. For a fixed $k_0\geq 1$ and every combination of $j=1,\ldots, d_X$ and $l=1,\ldots, d_Y$, use the same $K(j,l)=k_0\cdot d_Z+1$ weight functions
$$w(\mathbf z)=1, \text{ and } w_{d,k}(\mathbf z)= \sign(z_d -a_{d,k}),\, d=1,\ldots, d_Z, \, k=1,\ldots, k_0,$$
where $a_{d,k}$ is the empirical $\frac{k}{k_0+1}$-quantile of $Z_d$. This yields a total of $\mathbf K=d_X d_Y\cdot (k_0\cdot d_Z+1)$ weight functions, but the weight functions do not depend on $j$ and $l$.

\subsection{Multivariate WGCM With Estimated Weight Functions}\label{Sec_mWGCMEst}
The same procedure can be applied with estimated weight functions. We consider the same setting as in Section \ref{sec_MultFixMultDim}.

The difference is that the weight functions have been estimated on an 
auxiliary data set $\mathbf A$ independent of $(\mathbf X\upind n,\mathbf Y\upind n,\mathbf Z\upind n)$, obtained for example by sample splitting as in Section \ref{SubSec_WGCMEst}.
For each $j,l$, let $K(j,l)\in\mathbb N$, and for each $n\in\mathbb N$, let
$$\left\{\hat w_{jlk}\upind n :j=1,\ldots, d_X, l=1,\ldots, d_Y, k=1,\ldots, K(j,l)\right\}$$
be functions from $\mathbb R^{d_Z}\to \mathbb R$ that have been estimated on $\mathbf A$. In general, we would recommend to set $K(j,l)=1$ and use Method \ref{meth_WGCMEst} to estimate $w_{j,l}(z)=\sign(\Ex_P[\epsilon_j\xi_l|Z=z]$. However, Theorem \ref{thm_MultEstWMultDim} even allows for $K(j,l)$ to grow with $n$.
Let
$$\mathbf K=\mathbf K(n)=\sum_{j=1}^{d_X}\sum_{l=1}^{d_Y} K(j,l).$$
As in Section \ref{sec_MultFixMultDim}, define based on the weight functions $\hat w_{jlk}\upind n$
\begin{itemize}
	\item the vectors of weighted products of residuals $\mathbf R_{jlk}$;
	\item the test statistics $T_{jlk}\upind n$ of the individual WGCM;
	\item the aggregated test statistic $S_n$;
	\item the estimated covariance matrix $\hat \Sigma$;
	\item the multivariate normal vector $\hat {\mathbf T} \upind n$ and $\hat S_n$, the maximum absolute value of the components of $\hat {\mathbf T}\upind n$;
	\item the quantile function  $\hat G_n$ of $\hat S_n$ given $\hat \Sigma$.
\end{itemize}
We have the following variant of Theorem \ref{thm_MultFixedWMultDim}. Remark \ref{rmk_ThmMultFixedWMultDim}  also applies to this theorem.

\begin{theorem}[mWGCM.est]\label{thm_MultEstWMultDim}
Let $\mathcal P\subset \mathcal P_0$ and assume there exist $C,c> 0$ such that for all $n\in\mathbb N$ and $P\in\mathcal P$ there exists $D_n\geq 1$ such that either (C1a) and (C2) or (C1b) and (C2) hold. Assume that the conditions (\ref{eq_MultFixedWMultDimCondA0}), (\ref{eq_MultFixedWMultDimCondC0}) and (\ref{eq_MultFixedWMultDimCondD0}) are satisfied. Assume that there exist $C_1, c_1 >0$ (independent of $n$) such that for all $P\in\mathcal P$ we have $P$-almost surely for all $j,l,k$ and $n$
$$|\hat w_{jlk}\upind n(z)|\leq C_1 \text{ for all } z\in\mathbb R^{d_z}$$
and
$$\Ex_P\left[\epsilon_j^2\xi_l^2 \hat w_{jlk}\upind n(Z)^2|\mathbf A\right]\geq c_1\Ex_P\left[\epsilon_j^2\xi_l^2\right].$$
Then,
$$\sup_{P\in\mathcal P}\sup_{\alpha\in (0,1)}|\Prob_P(S_n\leq \hat G_n(\alpha))-\alpha|\to 0.$$
\end{theorem}
A proof can be found in Appendix \ref{App_ProofMultWGCM}.

In the case of fixed $d_X$, $d_Y$ and if $K(j,l)=1$ for all $j=1,\ldots, d_X$ and $l=1,\ldots, d_Y$, we have a combination of Corollary \ref{cor_PowerWGCMest} and Corollary \ref{cor_Power_mWGCMfix}. We assume, that we use Method \ref{meth_WGCMEst} to estimate $w_{jl}(\cdot)=\sign(\Ex_P[\epsilon_j\xi_l|Z=\cdot])$. We denote the estimate of $w_{jl}$ by $\hat w\upind n_{jl}$. Note that we write $w_{jl}$ and not $w_{jlk}$, since $K(j,l)=1$ for all $j,l$.

\begin{corollary}[mWGCM.est]\label{Cor_mWGCMEst}
Let $P\in\mathcal E_0$. Let $A_{f,j}$, $A_{g,l}$, $B_{f,jl}$ and $B_{g,jl}$ be defined as in (\ref{eq_DefAfj}), (\ref{eq_DefAgl}), (\ref{eq_DefBfjl}) and (\ref{eq_DefBgjl}) with the difference that all $\hat f_j$ and $\hat g_l$ have been estimated on an auxiliary data set independent of $(\mathbf X^{(n)}, \mathbf Y^{( n)},\mathbf Z^{( n)})$ and $\mathbf A$. Let $d_X$, $d_Y$ be fixed and $K(j,l)=1$ for all $j=1, \ldots, d_X$ and $l=1, \ldots, d_Y$. Assume that there exists $C>0$ such that for all $z\in \mathbb R^{d_Z}$, all $n\in \mathbb N$ and all $j,l$, we have $|\hat w\upind n_{jl}(z)|\leq C$. Assume that for all $j,l$ we have $A_{f,j} A_{g,l}=o_P(n^{-1})$, $B_{f,jl}=o_P(1)$ and $B_{g,jl}=o_P(1)$ as well as $\Ex_P[\epsilon_{P,j}^2\xi_{P,l}^2\hat w\upind n_{jl}(Z)^2|\mathbf A]>0$. Assume that there exists $\eta>0$ such that $\Ex_P\left[|\epsilon_{P,j}\xi_{P,l}|^{2+\eta}\right]<\infty$ for all $j,l$. If there exists $j\in\{1,\ldots, d_X\}$ and $l\in\{1,\ldots, d_Y\}$ such that $\Ex_P[\epsilon_{P,j}\xi_{P,l}|Z]$ is not almost surely equal to $0$ and if $ w_{jl}(z)=\sign(\Ex_P[\epsilon_{P,j}\xi_{P,l}|Z=z]$ can be consistently estimated in the sense that
$$\Ex_P\left[(\hat w_{jl}\upind n (Z)-w_{jl}(Z))^2 \mathbbm 1\{\Ex_P[\epsilon_{P,j}\xi_{P,l}|Z]\neq 0\}|\mathbf A\right] \to 0 \text{ in probability}, $$
then
$$\Prob_P(S_n \geq M)\to 1,$$
that is, mWGCM.est with fixed $d_X$, $d_Y$ and fixed number of weight functions $K(j,l)=1$ for all $j$ and $l$ has asymptotic power $1$ against alternative $P$ for any significance level $\alpha\in (0,1)$.
\end{corollary}

\subsection{Categorical Variables}\label{Sec_MultDiscrete}
The methodology for multivariate $X$ and $Y$ can also be used to treat arbitrary categorical variables. This leads to a test that has asymptotic power 1 against any alternative, provided that the conditions of Corollary \ref{Cor_mWGCMEst} are satisfied. Assume that $X$ takes values in $\{x_1,\ldots, x_J\}$ and $Y$ takes values in $\{y_1,\ldots, y_L\}$. 

We can apply the methodology from Section \ref{sec_MultFixMultDim} and Section \ref{Sec_mWGCMEst} to the variables
\begin{align*}
X^\ast &= \left(\mathbbm 1\{X=x_1\},\ldots, \mathbbm 1\{X=x_J\}\right)\\
Y^\ast &= \left(\mathbbm 1\{Y=y_1\},\ldots, \mathbbm 1\{Y=y_L\}\right).
\end{align*}
Note that $X\indep Y |Z$ if and only if for all $j=1,\ldots, J$ and for all $l=1,\ldots, L$,
$$\Prob_P(X=x_j, Y=y_l|Z)=\Prob_P(X=x_j|Z)\Prob_P(Y=y_L|Z)\text{ a.s.},$$
or equivalently if for all $j=1,\ldots, J$ and for all $l=1,\ldots, L$, we have $X_j^\ast\indep Y_l^\ast|Z$.
In the same way as in the case of binary $X$ and $Y$ in Section \ref{Sec_WGCMDiscrete}, we obtain that $X_j^\ast\indep Y_l^\ast|Z$ if and only if $\Ex_P[\epsilon_j\xi_l|Z]= 0$ a.s., where $\epsilon_j=X_j^\ast-\Ex_P[X_j^\ast|Z]$ and $\xi_l=Y_l^\ast-\Ex_P[Y_l^\ast|Z]$. In particular, we obtain the following version of Corollary \ref{Cor_mWGCMEst}.

\begin{corollary}[mWGCM.est, categorical case]\label{Cor_mWGCMEstDiscrete}
Let $X$ and $Y$ be categorical and assume that the distribution $P$ of $(X,Y,Z)$ satisfies $X\notindep Y|Z$. Let $A_{f,j}$, $A_{g,l}$, $B_{f,jl}$ and $B_{g,jl}$ be defined as in (\ref{eq_DefAfj}), (\ref{eq_DefAgl}), (\ref{eq_DefBfjl}) and (\ref{eq_DefBgjl}) with the difference that all $\hat f_j$ and $\hat g_l$ have been estimated on an auxiliary data set independent of $(\mathbf X^{(n)}, \mathbf Y^{( n)},\mathbf Z^{( n)})$ and $\mathbf A$. Let $J$, $L$ be fixed and $K(j,l)=1$ for all $j=1, \ldots, J$ and $l=1, \ldots, L$. Assume that there exists $C>0$ such that for all $z\in \mathbb R^{d_Z}$, all $n\in \mathbb N$ and all $j,l$, we have $|\hat w\upind n_{jl}(z)|\leq C$. Assume that for all $j,l$ we have $A_{f,j} A_{g,l}=o_P(n^{-1})$, $B_{f,jl}=o_P(1)$ and $B_{g,jl}=o_P(1)$ as well as $\Ex_P[\epsilon_{P,j}^2\xi_{P,l}^2\hat w\upind n_{jl}(Z)^2|\mathbf A]>0$. Assume that there exists $\eta>0$ such that $\Ex_P\left[|\epsilon_{P,j}\xi_{P,l}|^{2+\eta}\right]<\infty$ for all $j,l$. Since $X\notindep Y|Z$, there exists $j\in\{1,\ldots, d_X\}$ and $l\in\{1,\ldots, d_Y\}$ such that $\Ex_P[\epsilon_{P,j}\xi_{P,l}|Z]$ is not almost surely equal to $0$. If $ w_{jl}(z)=\sign(\Ex_P[\epsilon_{P,j}\xi_{P,l}|Z=z]$ can be consistently estimated in the sense that
$$\Ex_P\left[(\hat w_{jl}\upind n (Z)-w_{jl}(Z))^2 \mathbbm 1\{\Ex_P[\epsilon_{P,j}\xi_{P,l}|Z]\neq 0\}|\mathbf A\right] \to 0 \text{ in probability}, $$
then
$$\Prob_P(S_n \geq M)\to 1,$$
that is, mWGCM.est with fixed $J$, $L$ and fixed number of weight functions $K(j,l)=1$ for all $j$ and $l$ has asymptotic power $1$ against alternative $P$ for any significance level $\alpha\in (0,1)$.
\end{corollary}

\section{Stability of WGCM.est with Respect to Sample Splitting}\label{App_Stability}
The procedure WGCM.est described in Section \ref{SubSec_WGCMEst}, depends on a random split of the sample. If one repeats the procedure several times, the $p$-values will typically differ. In this section, we revisit the Boston housing data set, see Section \ref{Sec_Boston}, to investigate the effect of multiple sample splits for WGCM.est.

We repeat the analysis from Section \ref{Sec_Boston} for $100$ independent splits of the sample and count for each of the $13$ predictors, how often the (Holm-corrected) $p$-value is significant at level $\alpha = 0.05$. A plot of the frequencies for each variable can be found in Figure \ref{fig_BostonMultSig}.

\begin{figure}
	\begin{center}
	\includegraphics[width=0.8\textwidth]{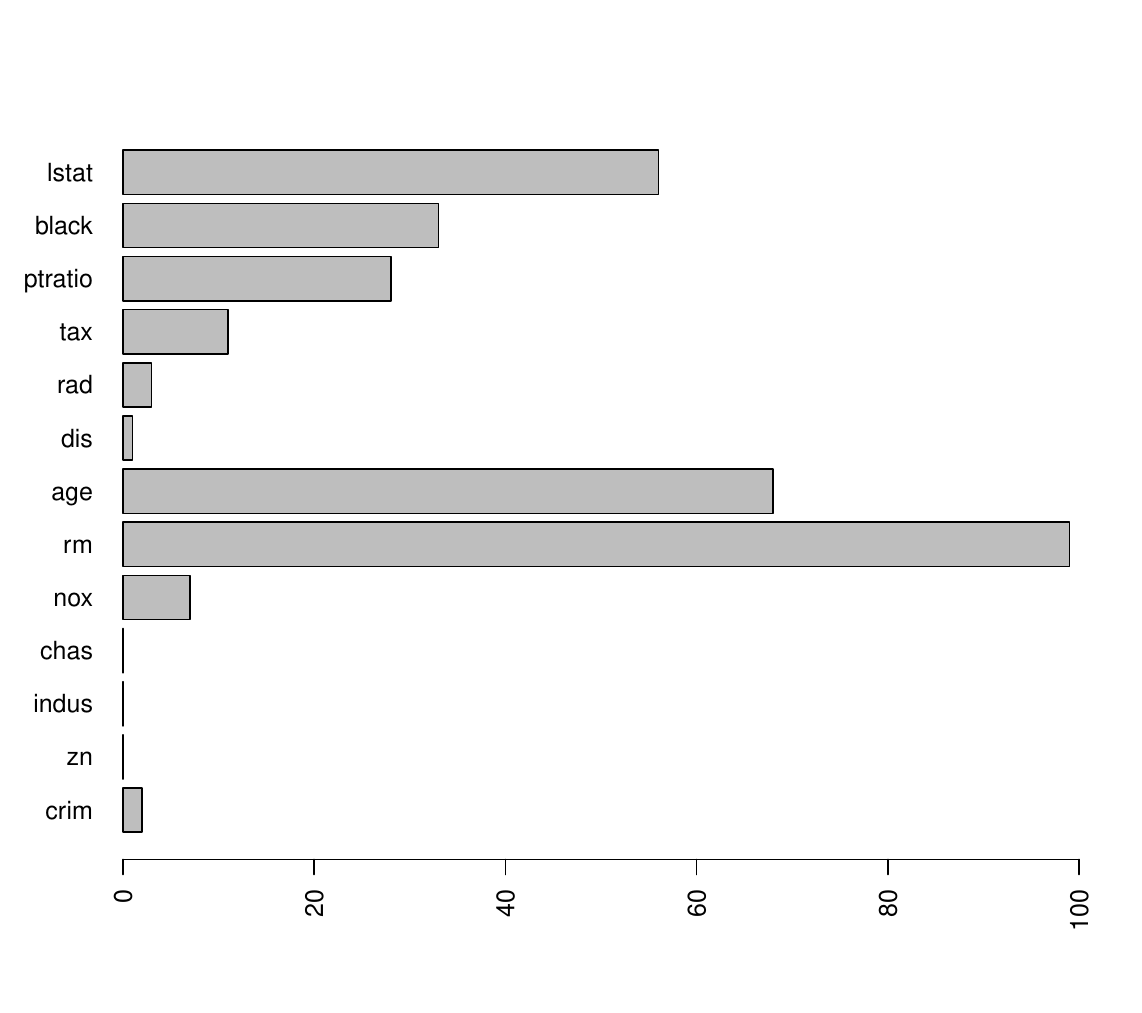}
	\end{center}
	\caption{Number of times, each predictor in the Boston housing data set is significant at level $\alpha = 0.05$ out of $100$ independent application of WGCM.est.}
	\label{fig_BostonMultSig}
\end{figure}

We see that the variable \texttt{rm} is almost always significant, and the variables \texttt{age} and \texttt{lstat} are significant in more than $50\%$ of the cases. These were also the significant variables for WGCM.est in the original analysis in Section \ref{Sec_Boston}. However, the analysis indicates that the $p$-values based on WGCM.est and hence also the number of significant variables is not very stable with respect to the randomness of the sample splitting. Depending on the split, it could well happen that only one or up to five variables are considered significant.

This leads to the question, how the $p$-values based on WGCM.est can be made more stable. Methods to aggregate $p$-values based on multiple sample splits have been developed in \cite{MeiMeiBuPValHDReg} and in \cite{DiCiccioMultDatSplit}. In the following, we apply the two approaches from  \cite{MeiMeiBuPValHDReg} to the Boston housing data.

In the context of WGCM.est, let $B\in \mathbb N$.
\begin{enumerate}
 \item Perform the test WGCM.est $B$ times with $B$ independent sample splits, obtaining $p$-values $P_1, \ldots, P_B$.
 \item For $\gamma \in (0,1)$, define the aggregated $p$-value $Q(\gamma)$ as
 $$Q(\gamma)=\min(1, q_\gamma(\{P_b/\gamma|b=1,\ldots, B\})),$$
 where $q_\gamma(\cdot)$ is the empirical $\gamma$-quantile.
\end{enumerate}
If the $p$-values $P_1, \ldots, P_B$ are asymptotically correct, then $Q(\gamma)$ is also an asymptotically correct $p$-value, see Theorem 3.1 \cite{MeiMeiBuPValHDReg}.

However, one cannot simply search for $\gamma$ yielding the lowest value of $Q(\gamma)$. Instead, for a fixed lower bound $\gamma_\text{min}\in(0,1)$ define
$$P=\min\left(1, (1-\log\gamma_\text{min})\inf_{\gamma\in (\gamma_\text{min}, 1)} Q(\gamma)\right).$$
This is also an asymptotically correct $p$-value, see Theorem 3.2 in \cite{MeiMeiBuPValHDReg}.

For the Boston housing data set, we again apply the procedure from Section \ref{Sec_Boston} for $100$ independent splits of the sample, but without applying Holm's method. For each of the $13$ predictors, we calculate the corresponding aggregated $p$-values $Q(\gamma)$ with $\gamma=\frac{1}{2}$ (Method 1) and $P$ with $\gamma_\text{min}=0.05$ (Method 2). At the end, we apply Holm's correction once to the $13$ $p$-values from Method 1 and once to the $13$ $p$-values from Method 2. Plots of the $p$-values based on this procedure can be found in Figure \ref{fig_BostonPAggr}. We see that for both aggregation methods, still the variables \texttt{rm}, \texttt{age} and \texttt{lstat} are significant.

\begin{figure}
	\begin{center}
	\includegraphics[width=0.8\textwidth]{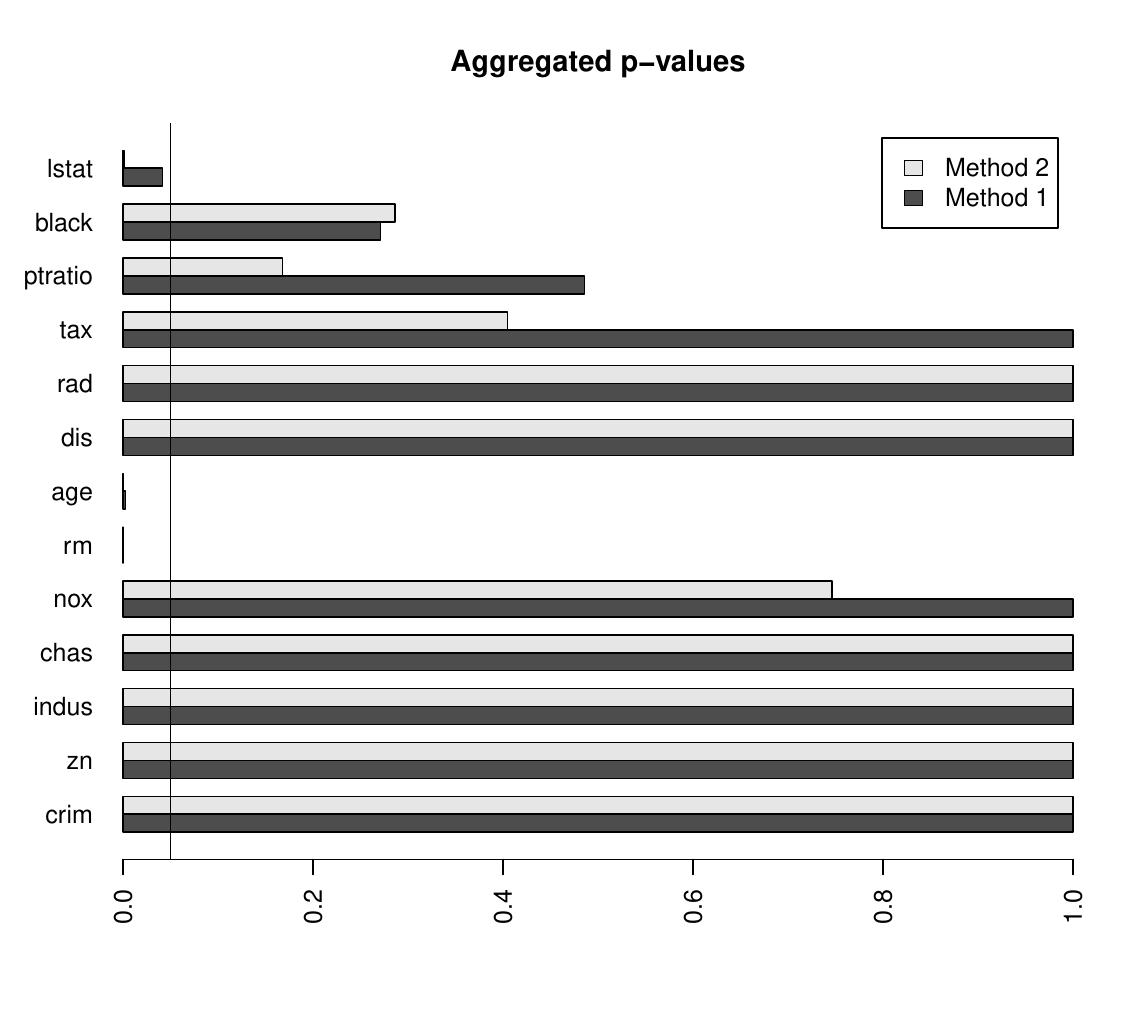}
	\end{center}
	\caption{Aggregated $p$-values for the Boston housing data set using Method 1 and Method 2 based on $100$ independent applications of WGCM.est. Note that the results are comparable to Figure \ref{fig_PValBoston}.}
	\label{fig_BostonPAggr}
\end{figure}

In principle, applying WGCM.est with multiple sample splits is a good idea, as it makes the $p$-values more stable. A caveat is however that especially in these variable importance examples, the runtime gets quite big, as there are already many tests involved in the procedure, even if only one sample split is used.

\section{Proofs of Section \ref{SubSec_WGCMEst}}\label{App_ProofUnivariateWGCM}
In this section, we give the proofs of the Theorems on the univariate WGCM with single fixed and single estimated weight function.
\subsection{A More General Result} \label{sec_WGCM1DMG}
To prove Theorem \ref{thm_WGCM1DEst}, we use the following more general result. Let $\mathcal P\subset \mathcal P_0$. For $P\in\mathcal P$ and $C,c>0$, define 
$$\mathcal W_{P,C,c}=\left\{w:\mathbb R^{d_Z}\to \mathbb R \,\Bigr\rvert\,|w|\leq C \land \Ex_P\left[\epsilon_P^2\xi_P^2w(Z)^2\right]\geq c\right\}.$$
To be more precise, $\mathcal W_{P,C,c}$ is the set of \textit{measurable} functions with those properties.
For each $w\in\mathcal W_{P,C,c}$, let
\begin{equation*}
R_i\upind n = R_{w,i}\upind n=\left(x_i-\hat f\upind n(z_i)\right)\left(y_i-\hat g\upind n(z_i)\right) w(z_i), \quad i=1,\ldots, n,
\end{equation*}
and let
\begin{equation}\label{eq_DefT1DMG}
T\upind n=T_w\upind n=\frac{\frac{1}{\sqrt n}\sum_{i=1}^n R_i\upind n}{\left(\frac{1}{n}\sum_{i=1}^n {\left(R_i\upind n\right)}^2-\left(\frac{1}{n}\sum_{r=1}^n R_r\upind n\right)^2\right)^{1/2}}=\frac{\tau_{N}\upind n}{\tau_{D}\upind n}.
\end{equation}
Then, we have the following result under the null hypothesis:
\begin{theorem}\label{thm_WGCM1DGenRes}
Let $A_f$, $A_g$, $B_f$ and $B_g$ be defined as in (\ref{eq_DefAfAg}) and (\ref{eq_DefBfBg}). Let $\mathcal P\subset \mathcal P_0$ and $C,c>0$. Assume that $A_f A_g=o_\mathcal P(n^{-1})$, $B_f=o_\mathcal P(1)$ and $B_g= o_\mathcal P(1)$. If there exists $\eta > 0$ such that $\sup_{P\in\mathcal P}\Ex_P\left[|\epsilon_P\xi_P|^{2+\eta}\right]<\infty$, and if for all $P\in\mathcal P$, the set $\mathcal W_{P,C,c}$ is nonempty, then
$$\sup_{P\in\mathcal P}\sup_{w\in\mathcal W_{P, C,c}}\sup_{t\in\mathbb R} |\Prob_P(T_w^{(n)}\leq t)-\Phi(t)|\to 0.$$
\end{theorem}
We also have a more general power result. Let $\mathcal P\subset \mathcal E_0$. For $P\in\mathcal P$ and $C,c>0$, let
$$\mathcal W_{P,C,c}=\left\{w:\mathbb R^{d_Z}\to\mathbb R \, \Bigr\rvert\, |w|\leq C \land \var_P\left(\epsilon_P\xi_Pw(Z)\right)\geq c\right\}.$$
\begin{theorem}\label{thm_PowerWGCM1DMG}
Let $A_f$, $A_g$, $B_f$ and $B_g$ be defined as in (\ref{eq_DefAfAg}) and (\ref{eq_DefBfBg}), but with the difference that $\hat f$ and $\hat g$ have been estimated on an auxiliary data set which is independent of $(\mathbf X\upind n, \mathbf Y\upind n,\mathbf Z\upind n)$. Let $\mathcal P\subset \mathcal E_0$ and $C,c>0$. Assume that $A_f A_g=o_\mathcal P(n^{-1})$, $B_f=o_\mathcal P(1)$ and $B_g= o_\mathcal P(1)$. If there exists $\eta > 0$ such that $\sup_{P\in\mathcal P}\Ex_P\left[|\epsilon_P\xi_P|^{2+\eta}\right]<\infty$, and if for all $P\in\mathcal P$, the set $\mathcal W_{P,C,c}$ is nonempty, then
$$\sup_{P\in\mathcal P}\sup_{w\in\mathcal W_{P, C,c}}\sup_{t\in\mathbb R} \left|\Prob_P\left(\frac{\tau_N\upind n-\sqrt n\rho_{P,w}}{\tau_D\upind n}\leq t\right)-\Phi(t)\right|\to 0,$$
where
$$\rho_{P,w}=\Ex_P\left[\epsilon_P\xi_P w(Z)\right].$$
\end{theorem}

\subsection{Proofs}\label{sec_WGCM1DProofs}
We first prove Theorem \ref{thm_WGCM1DGenRes} and then show that Theorem \ref{thm_WGCM1DEst} follows from Theorem \ref{thm_WGCM1DGenRes}.

\subsubsection{Proof of Theorem \ref{thm_WGCM1DGenRes}}
\begin{proof}
The proof closely follows the proof of Theorem 6 in Section D.1 in the supplementary material of \cite{ShahPetersCondInd}. We will sometimes omit the dependence on $n$ and $P$ in the notation. We will repeatedly use limit theorems from Appendix \ref{App_LimTheo}

Fix $C,c>0$ and write $\mathcal W_P$ instead of $\mathcal W_{P,C,c}$. For $n\in\mathbb N$, $P\in\mathcal P$ and $w\in\mathcal W_{P}$, define
$$\sigma_w^2=\sigma_{P,w}^2=\Ex_P\left[\epsilon_P^2\xi_P^2 w(Z)^2\right].$$
We first prove that 
\begin{equation*}
\sup_{P\in\mathcal P}\sup_{w \in\mathcal W_{P}}\sup_{t\in \mathbb R}\left|\Prob_P\left(\frac{\frac{1}{\sqrt n}\sum_{i=1}^n \epsilon_i\xi_i  w(z_i)}{\sigma_w}\leq t\right)-\Phi(t)\right|\to 0.
\end{equation*}
To simplify notation, we will abbreviate this in the following with
\begin{equation}\label{eq_ApplUnCLT}
\frac{\frac{1}{\sqrt n}\sum_{i=1}^n \epsilon_i\xi_i  w(z_i)}{\sigma_w} \todistunifPW \mathcal N(0,1).
\end{equation}

This is an application of Lemma \ref{lem_UnCLT}, where the random variable $\zeta$ corresponds to $\frac{\epsilon_P\xi_P w(Z)}{\sigma_w}$. Instead of the set $\mathcal P$ of distributions for $\zeta$ in Lemma \ref{lem_UnCLT}, we look at the set of distributions determined by $(P,w)$ for $\frac{\epsilon_P\xi_P w(Z)}{\sigma_w}$ where $P$ varies in $\mathcal P$ and $w$ varies in $\mathcal W_P$. We have
$$\Ex_P\left[\frac{\epsilon_P\xi_P w(Z)}{\sigma_w}\right]=\Ex_P\left[\Ex_P[\epsilon_P|Y,Z]\xi_Pw(Z)\right]/\sigma_w=0$$
and $\Ex_P[\epsilon_P^2\xi_P^2 w(Z)^2/\sigma_w^2]=1$ as well as
$$\sup_{P,w}\Ex_P\left[\left|\frac{\epsilon_P\xi_P w(Z)}{\sigma_w}\right|^{2+\eta}\right]\leq \left(\frac{C}{\sqrt c}\right)^{2+\eta}\sup_{P\in\mathcal P}\Ex_P\left[|\epsilon_P\xi_P|^{2+\eta}\right]<\infty$$
by assumption, so the lemma implies (\ref{eq_ApplUnCLT}).

In the following, we will repeatedly apply Lemmas \ref{lem_UnWLLN}, \ref{lem_UnSlutsky} and \ref{lem_BoundConv} over the class of distributions for $\frac{\epsilon_P\xi_P w(Z)}{\sigma_w}$ determined by $(P,w)$ in a similar fashion.

For $i=1,\ldots n$, define
\begin{align*}
\Delta f_i&=f(z_i)-\hat f(z_i),\\
\Delta g_i&=g(z_i)-\hat g(z_i).
\end{align*}
We first prove that 
\begin{equation}\label{eq_AsNormNum}
\frac{\tau_N}{\sigma_w}\todistunifPW \mathcal N(0,1).
\end{equation}
Observe that 
\begin{align}
\tau_N=\frac{1}{\sqrt n}\sum_{i=1}^n R_i&=\frac{1}{\sqrt n}\sum_{i=1}^n w(z_i)(f(z_i)-\hat f(z_i)+\epsilon_i)(g(z_i)-\hat g(z_i)+\xi_i)\nonumber\\
&=(b+\nu_f+\nu_g)+\frac{1}{\sqrt n}\sum_{i=1}^n  w (z_i)\epsilon_i\xi_i,\label{eq_DecompNum}
\end{align}
with
\begin{align*}
b&=\frac{1}{\sqrt n}\sum_{i=1}^n w(z_i)\Delta f_i\Delta g_i,\\
\nu_f&=\frac{1}{\sqrt n}\sum_{i=1}^n w(z_i)\xi_i\Delta f_i,\\
\nu_g&=\frac{1}{\sqrt n}\sum_{i=1}^n w(z_i)\epsilon_i\Delta g_i.
\end{align*}

In analogy to the notation of (\ref{eq_ApplUnCLT}), we write for a sequence $(V_n)_{n\in\mathbb N}$ of random variables depending on $w\in\mathcal W_P$,
$$V_n=o_{\mathcal P,\mathcal W}(1)$$
if for all $\delta >0$,
$$\sup_{P\in\mathcal P}\sup_{w\in\mathcal W_P}\Prob_P(|V_n|>\delta)\to 0.$$

For $b$, by the Cauchy-Schwarz inequality,
\begin{equation}\label{eq_TermB}
|b|\leq C \frac{1}{\sqrt n}\sum_{i=1}^n |\Delta f_i\Delta g_i| \leq C \sqrt n \sqrt{A_f A_g} = o_{\mathcal P,\mathcal W}(1),
\end{equation}
since $C \sqrt n \sqrt{A_f A_g}$ is independent of $w\in\mathcal W_P$ and $A_f A_g=o_\mathcal P\left(n^{-1}\right)$ by assumption.

Next, we want to control $\nu_f$ and $\nu_g$. For $\nu_g$, 
\begin{align*}
\Ex_P[\epsilon_i\Delta g_i|\mathbf Y,\mathbf Z]&=\Ex_P[\epsilon_i|\mathbf Z]\Delta g_i=0,\\
\Ex_P[\epsilon_i^2\Delta g_i^2|\mathbf Y,\mathbf Z]&=\Ex_P[\epsilon_i^2|\mathbf Z]\Delta g_i^2=\Delta g_i^2u(z_i).
\end{align*}
Thus,
\begin{align*}
\Ex_P[\nu_g^2|\mathbf Y,\mathbf Z]&\leq C^2\frac{1}{n}\sum_{i=1}^n \Ex_P[\epsilon_i^2\Delta g_i^2|\mathbf Y,\mathbf Z]+C^2\sum_{i\neq j}\Ex_P[\epsilon_i\Delta g_i\epsilon_j\Delta g_j|\mathbf Y,\mathbf Z]\\
&=C^2 \frac{1}{n}\sum_{i=1}^n \Delta g_i^2u(z_i)+C^2\sum_{i\neq j}\Ex_P[\epsilon_i\Delta g_i|\mathbf Y,\mathbf Z]\Ex_P[\epsilon_j\Delta g_j|\mathbf Y,\mathbf Z]\\
&=C^2\frac{1}{n}\sum_{i=1}^n \Delta g_i^2u(z_i)=C^2 B_g=o_{\mathcal P}(1)
\end{align*}
by assumption. Since $C^2 B_g$ is independent of $w\in\mathcal W_P$, we get that
\begin{equation}\label{eq_ExpNuGSq}
\Ex_P[\nu_g^2|\mathbf Y,\mathbf Z]=o_{\mathcal P, \mathcal W}(1).
\end{equation}
Thus, for all $\epsilon > 0$ using Markov's inequality, 
\begin{align*}
\Prob_P\left(\nu_g^2\geq\epsilon\right)&=\Prob_P\left(\nu_g^2\land \epsilon \geq\epsilon\right)\\
&\leq\epsilon^{-1}\Ex_P\left[\Ex_P\left[\nu_g^2\land \epsilon|\mathbf Y,\mathbf Z\right]\right]\\
&\leq \epsilon^{-1}\Ex_P\left[\Ex_P\left[\nu_g^2|\mathbf Y,\mathbf Z\right]\land \epsilon\right].
\end{align*}
Equation (\ref{eq_ExpNuGSq}) and Lemma \ref{lem_BoundConv} applied to the variables $\Ex[\nu_g^2|\mathbf Y,\mathbf Z]$ with distributions determined by $(P,w)$ therefore imply that $\nu_g=o_{\mathcal P,\mathcal W}(1)$. Similarly, $\nu_f=o_{\mathcal P,\mathcal W}(1)$.

By (\ref{eq_DecompNum}), we have
\begin{align*}
\frac{\tau_N}{\sigma_w}&=\frac{b+\nu_f+\nu_g}{\sigma_w}+\frac{\frac{1}{\sqrt n}\sum_{i=1}^n  w(z_i)\epsilon_i\xi_i}{\sigma_w}.
\end{align*}
Since $\sigma_w\geq \sqrt{c}$, the first term is $o_{\mathcal P,\mathcal W}(1)$. By (\ref{eq_ApplUnCLT}) and Lemma \ref{lem_UnSlutsky}, 1., we get
(\ref{eq_AsNormNum}).

Next, we aim to prove 
\begin{equation}\label{eq_AsConstDen}
\frac{\tau_D}{\sigma_w}=1+o_{\mathcal P,\mathcal W}(1).
\end{equation}
For this, it is enough to prove
\begin{equation}\label{eq_AsConstDenSq}
\frac{\tau_D^2}{\sigma_w^2}=1+o_{\mathcal P,\mathcal W}(1),
\end{equation}
by continuity of $t\mapsto \sqrt t$ at $t=1$.

Recall that 
$$\tau_D^2=\frac{1}{n}\sum_{i=1}^n {R_i}^2-\left(\frac{1}{n}\sum_{r=1}^n R_r\right)^2.$$
Since by (\ref{eq_AsNormNum}), $\frac{1}{\sqrt {n}\sigma_w} \sum_{i=1}^n R_i=\frac{\tau_N}{\sigma_w}\todistunifPW \mathcal N(0,1)$, it follows that $\frac{1}{n\sigma_w}\sum_{i=1}^n R_i=o_{\mathcal P,\mathcal W}(1)$. To prove (\ref{eq_AsConstDenSq}), it therefore suffices to prove
\begin{equation}\label{eq_AsConstDenReduced}
\frac{\frac{1}{n}\sum_{i=1}^n R_i^2}{\sigma_w^2}=1+o_{\mathcal P,\mathcal W}(1).
\end{equation}
For this, write
\begin{align}
|R_i^2-w(z_i)^2\epsilon_i^2\xi_i^2|&\leq|w(z_i)^2(\epsilon_i+\Delta f_i)^2(\xi_i+\Delta g_i)^2- w(z_i)^2\epsilon_i^2\xi_i^2|\nonumber \\
&\leq \underbrace{ w(z_i)^2(\Delta f_i^2+2|\epsilon_i\Delta f_i|)(\Delta g_i^2+2|\xi_i\Delta g_i|)}_{\eqqcolon I_i}\nonumber \\
&+\underbrace{ w(z_i)^2\epsilon_i^2(\Delta g_i^2+2|\xi_i\Delta g_i|)}_{\eqqcolon II_i}+\underbrace{ w(z_i)^2\xi_i^2(\Delta f_i^2+2|\epsilon_i\Delta f_i|)}_{\eqqcolon III_i} \label{eq_DecompDiffDen}
\end{align}
For $I_i$, observe that $2\Delta f_i^2|\xi_i\Delta g_i|\leq \Delta f_i^2(\xi_i^2+\Delta g_i^2)$ and $2\Delta g_i^2|\epsilon_i\Delta f_i|\leq \Delta g_i^2(\epsilon_i^2+\Delta f_i^2)$. Thus, we have
\begin{align*}
I_i&\leq C^2(\Delta f_i^2\Delta g_i^2+2\Delta f_i^2|\xi_i\Delta g_i|+2\Delta g_i^2|\epsilon_i\Delta f_i|+4|\epsilon_i\xi_i\Delta f_i\Delta g_i|)\\
&\leq C^2(3\Delta f_i^2\Delta g_i^2+\epsilon_i^2\Delta g_i^2+\xi_i^2\Delta f_i^2+4|\epsilon_i\xi_i\Delta f_i\Delta g_i|).
\end{align*}
 Observe that
$$\frac{1}{n}\sum_{i=1}^n \Delta f_i^2\Delta g_i^2\leq\frac{1}{n}\sum_{i=1}^n \Delta f_i^2\sum_{i=1}^n \Delta g_i^2 =n A_f A_g= o_{\mathcal P}(1)$$
by assumption.
Furthermore, for all $\delta > 0$
\begin{align*}
\Prob_P\left(\frac{1}{n}\sum_{i=1}^n\epsilon_i^2\Delta g_i^2 \geq \delta\right)&=\Prob_P\left(\frac{1}{n}\sum_{i=1}^n\epsilon_i^2\Delta g_i^2\wedge\delta\geq \delta\right)\\
&\leq \delta^{-1} \Ex_P\left[\Ex_P\left[\frac{1}{n}\sum_{i=1}^n\epsilon_i^2\Delta g_i^2\wedge\delta|\mathbf Y,\mathbf Z\right]\right]\\
&\leq \delta^{-1} \Ex_P\left[\Ex_P\left[\frac{1}{n}\sum_{i=1}^n\epsilon_i^2\Delta g_i^2|\mathbf Y,\mathbf Z\right]\wedge\delta\right]\\
&= \delta^{-1} \Ex_P\left[B_g\wedge\delta\right],
\end{align*}
similarly as before. Since $B_g= o_\mathcal P(1)$ by assumption, Lemma \ref{lem_BoundConv} implies that $\Ex_P[B_g\land \delta]\to 0$ and thus,
\begin{equation}\label{eq_BoundEpsG}
\frac{1}{n}\sum_{i=1}^n\epsilon_i^2\Delta g_i^2=o_{\mathcal P}(1).
\end{equation}
Similarly, also $n^{-1}\sum_{i=1}^n\xi_i^2\Delta f_i^2=o_{\mathcal P}(1)$. Moreover,
\begin{align}
\frac{1}{n}\sum_{i=1}^n|\epsilon_i\xi_i\Delta f_i\Delta g_i|&\leq\left(\frac{1}{n}\sum_{i=1}^n\epsilon_i^2\xi_i^2\right)^{1/2}\left(\frac{1}{n}\sum_{i=1}^n\Delta f_i^2\Delta g_i^2\right)^{1/2} \nonumber\\
&\leq \left(\frac{1}{n}\sum_{i=1}^n\epsilon_i^2\xi_i^2\right)^{1/2} \left(\frac{1}{n}\sum_{i=1}^n\Delta f_i^2\sum_{i=1}^n\Delta g_i^2\right)^{1/2} \label{eq_BoundMixedTerm}
\end{align}
By Lemma \ref{lem_UnWLLN}, we have that for all $\delta >0$
$$\sup_{P\in\mathcal P}\Prob_P\left(\left|\frac{1}{n}\sum_{i=1}^n\epsilon_i^2\xi_i^2-\Ex_P[\epsilon^2\xi^2]\right|>\delta \right)\to 0,$$
so
\begin{equation}\label{eq_BoundEpsXi}
\frac{1}{n}\sum_{i=1}^n\epsilon_i^2\xi_i^2=O_\mathcal P(1).
\end{equation}
The second factor in (\ref{eq_BoundMixedTerm}) is equal to $\sqrt{nA_f A_g}=o_{ \mathcal P}(1)$. This implies $n^{-1}\sum_{i=1}^n|\epsilon_i\xi_i\Delta f_i\Delta g_i|=o_{ \mathcal P}(1)$. In total, we get
\begin{equation}\label{eq_ConI}
\frac{1}{n}\sum_{i=1}^n I_i=o_{\mathcal P,\mathcal W}(1).
\end{equation}
For $II_i$, we have
\begin{align*}
\frac{1}{n}\sum_{i=1}^n  w(z_i)^2(\epsilon_i^2\Delta g_i ^2+2\epsilon_i^2|\xi_i\Delta g_i|)&\leq 2C^2\left(\frac{1}{n}\sum_{i=1}^n\epsilon_i^2\xi_i^2\right)^{1/2}\left(\frac{1}{n}\sum_{i=1}^n\epsilon_i^2\Delta g_i^2\right)^{1/2}+\frac{C^2}{n}\sum_{i=1}^n\epsilon_i^2\Delta g_i^2.
\end{align*}
By (\ref{eq_BoundEpsG}), $n^{-1}\sum_{i=1}^n\epsilon_i^2\Delta g_i^2=o_{ \mathcal P}(1)$, so using (\ref{eq_BoundEpsXi}), also
$$\frac{1}{n}\sum_{i=1}^n II_i=o_{ \mathcal P,\mathcal W}(1).$$
Similarly, $\frac{1}{n}\sum_{i=1}^n III_i=o_{ \mathcal P,\mathcal W}(1)$. Combining this with (\ref{eq_ConI}) and (\ref{eq_DecompDiffDen}), we get
\begin{equation*}\label{eq_ResultDenom}
\frac{1}{n}\sum_{i=1}^n R_i^2=\frac{1}{n}\sum_{i=1}^n w(z_i)^2\epsilon_i^2\xi_i^2+o_{ \mathcal P,\mathcal W}(1).
\end{equation*}
With Lemma \ref{lem_UnWLLN}, we obtain that for all $\delta>0$
$$\sup_{P\in\mathcal P}\sup_{w\in\mathcal W_P} \Prob_P\left(\left|\frac{1}{n}\sum_{i=1}^n w(z_i)^2\epsilon_i^2\xi_i^2-\sigma_w^2\right|>\delta\right)\to 0$$
 and thus,
$$\frac{1}{n}\sum_{i=1}^n R_i^2-\sigma_w^2=o_{\mathcal P,\mathcal W}(1).$$
Since $\sigma_w^2\geq c$, we obtain
$$\frac{\frac{1}{n}\sum_{i=1}^n R_i^2}{\sigma_w^2}-1=o_{\mathcal P,\mathcal W}(1) $$
so (\ref{eq_AsConstDenReduced}) and (\ref{eq_AsConstDen}) follow.
Since $T^{(n)}=\frac{\tau_N^{(n)}/\sigma_w}{\tau_D^{(n)}/\sigma_w}$, we get by Lemma \ref{lem_UnSlutsky} with (\ref{eq_AsNormNum})
$$T\upind n\todistunifPW \mathcal N(0,1).$$
This concludes the proof.
\end{proof}

\subsubsection{Proof of Theorem \ref{thm_WGCM1DEst}}
\begin{proof}
The statement 2. follows from Theorem \ref{thm_WGCM1DGenRes} in the following way: For $P\in\mathcal P$, $w\in \mathcal W_{P,C,c}$ and $t\in\mathbb R$, let 
$$\Gamma_n(P,w, t)=\Prob_P(T_w\upind n\leq t).$$
with $T_w\upind n$ defined as in (\ref{eq_DefT1DMG}). By Theorem \ref{thm_WGCM1DGenRes}, we have
$$\sup_{P\in\mathcal P}\sup_{w\in\mathcal W_{P,C,c}}\sup_{t\in\mathbb R} |\Gamma_n(P,w,t)-\Phi(t)|\to 0.$$
For the auxiliary data set $\mathbf A$, let $\left(\hat w\upind n\right)_{n\in\mathbb N}$ be the sequence of functions estimated on $\mathbf A$.
From the assumptions of Theorem \ref{thm_WGCM1DEst}, 2., we know that there exist $C,c>0$ such that for all $P\in\mathcal P$, $P$-almost surely for all $n\in\mathbb N$ we have $\hat w\upind n\in\mathcal W_{P,C,c}$. Since $\mathbf A$ is independent of $(\mathbf X\upind n,\mathbf Y\upind n,\mathbf Z\upind n)$, we have
$$\Prob_P(T\upind n\leq t|\mathbf A)=\Gamma_n(P, \hat w\upind n, t).$$
Using iterated expectations, it follows that
\begin{align*}
\sup_{P\in\mathcal P}\sup_{t\in\mathbb R}\left|\Prob_P(T\upind n\leq t)-\Phi(t)\right|&= \sup_{P\in\mathcal P}\sup_{t\in\mathbb R}\left|\Ex_P\left[\Gamma_n(P,\hat w\upind n,t)\right]-\Phi(t)\right|\\
&\leq \sup_{P\in\mathcal P}\Ex_P\left[\sup_{t\in\mathbb R}|\Gamma_n(P,\hat w\upind n,t)-\Phi(t)|\right]\\
&\leq \sup_{P\in\mathcal P}\Ex_P\left[\sup_{Q\in\mathcal P}\sup_{w\in\mathcal W_{Q,C,c}}\sup_{t\in\mathbb R} |\Gamma_n(Q,w,t)-\Phi(t)|\right]\\
&=\sup_{P\in\mathcal P}\sup_{w\in\mathcal W_{P,C,c}}\sup_{t\in\mathbb R} |\Gamma_n(P,w,t)-\Phi(t)|\to 0.
\end{align*}
This concludes the proof of 2. The statement 1. can be proven in a similar fashion, but since we do not require uniformity over a collection $\mathcal P$, it is enough to have
\begin{equation}\label{eq_ResCasei}
\Prob_P(T\upind n\leq t|\mathbf A)\to \Phi(t) \text{ a.s.}
\end{equation}
Then, one can conclude using dominated convergence. To show (\ref{eq_ResCasei}), it is enough to a.s. have $c>0$ such that $\sigma_n^2\geq c$, so $c$ is allowed to depend on $\mathbf A$.
\end{proof}

\subsubsection{Proof of Theorem \ref{thm_PowerWGCM1DMG}}
\begin{proof}
The proof is very similar to the proof of Theorem \ref{thm_WGCM1DGenRes}, see also the proof of Theorem 8 in Section D.3 in the supplementary material of \cite{ShahPetersCondInd}. We will therefore only sketch the main steps. Apart from not assuming the null hypothesis anymore, a main difference is that we have a data set $\mathbf B$ independent of $(\mathbf X\upind n,\mathbf Y\upind n,\mathbf Z\upind n)$ on which $\hat f$ and $\hat g$ have been estimated.

Let
$$\rho_w=\rho_{P,w}=\Ex_P[\epsilon_P\xi_P w(Z)]$$
and
$$\sigma_w^2=\sigma_{P,w}^2=\var_P(\epsilon_P\xi_P w(Z))=\Ex_P\left[\epsilon_P^2\xi_P^2 w(Z)^2\right]-\rho_{P,w}^2.$$

We first prove that (with the notation of (\ref{eq_ApplUnCLT}) and (\ref{eq_AsNormNum}))
\begin{equation}\label{eq_AsNormNumP}
\frac{\tau_N-\sqrt n \rho_w}{\sigma_w}\todistunifPW \mathcal N(0,1).
\end{equation}
For this, write
\begin{align}\label{eq_DecompNumPrPow}
\tau_N-\sqrt n  \rho_w&=(b+\nu_f+\nu_g)+\frac{1}{\sqrt n}\sum_{i=1}^n \left(\epsilon_i\xi_i w(z_i)-\rho_w\right),
\end{align}
with $b$, $\nu_f$ and $\nu_g$ defined as in (\ref{eq_DecompNum}). Similarly to (\ref{eq_ApplUnCLT}), one can show that
$$\frac{1}{\sqrt n \sigma_w}\sum_{i=1}^n\left(\epsilon_i\xi_i w(z_i)-\rho_w\right)\todistunifPW \mathcal N(0,1).$$
The term $b$ can be controlled as in (\ref{eq_TermB}). For the term $\nu_g$, replace conditioning on $\mathbf Y,\mathbf Z$ by conditioning on $\mathbf B, \mathbf Z$ and similarly for $\nu_f$. Since $\sigma_w^2\geq c$, we arrive at (\ref{eq_AsNormNumP}).

Next, we prove
\begin{equation}\label{eq_AsConstDenPrPow}
\frac{\tau_D}{\sigma_w}=1+o_{\mathcal P,\mathcal W}(1),
\end{equation}
which follows from ${\tau_D^2}/{\sigma_w^2}=1+o_{\mathcal P,\mathcal W}(1).$
Recall that
\begin{equation}\label{eq_DefTauDAlt}
\tau_D^2=\frac{1}{n}\sum_{i=1}^n {R_i}^2-\left(\frac{1}{n}\sum_{r=1}^n R_r\right)^2.
\end{equation}
Since by (\ref{eq_AsNormNumP}), $\frac{1}{\sqrt {n}\sigma_w} \sum_{i=1}^n (R_i-\rho_w)\todistunifPW\mathcal N(0,1)$, we have  
\begin{equation}\label{eq_RightTermDenumPow}
\frac{\frac{1}{n}\sum_{i=1}^n R_i-\rho_w}{\sigma_w}=o_{\mathcal P,\mathcal W}(1).
\end{equation}
Next, we show
\begin{equation}\label{eq_LeftTermDenumPow}
\frac{\frac{1}{n}\sum_{i=1}^n R_i^2-\Ex_P\left[\epsilon^2\xi^2 w(Z)^2\right]}{\sigma_w^2}=o_{\mathcal P,\mathcal W}(1).
\end{equation}
For this, write
$$|R_i^2- w(z_i)^2\epsilon_i^2\xi_i^2|\leq I_i+II_i+III_i$$
exactly as in the proof of Theorem \ref{thm_WGCM1DGenRes}. By replacing conditioning on $\mathbf Y,\mathbf Z$ with conditioning on $\mathbf B,\mathbf Z$, one can show
$$\sum_{i=1}^n (I_i+II_i+III_i)=o_{\mathcal P,\mathcal W}(1)$$
in the same way as there. 
Using Lemma \ref{lem_UnWLLN}, one can show
$$\frac{1}{n}\sum_{i=1}^n\epsilon_i^2\xi_i^2 w(z_i)^2-\Ex_P[\epsilon^2\xi^2 w(Z)^2]=o_{\mathcal P,\mathcal W}(1),$$
so this implies (\ref{eq_LeftTermDenumPow}).

Equations (\ref{eq_DefTauDAlt}), (\ref{eq_RightTermDenumPow}) and (\ref{eq_LeftTermDenumPow}) together imply
\begin{align*}
\frac{\tau_D^2}{\sigma_w^2}&=\frac{\Ex_P[\epsilon^2\xi^2 w(Z)^2]}{\sigma_w^2}+o_{\mathcal P,\mathcal W}(1)-\left(\frac{\rho_w}{\sigma_w}+o_{\mathcal P,\mathcal W}(1)\right)^2\\
&=\frac{\Ex_P[\epsilon^2\xi^2 w(Z)^2]-\rho_w^2}{\sigma_w^2}+o_{\mathcal P,\mathcal W}(1)\\
&=1+o_{\mathcal P,\mathcal W}(1),
\end{align*}
where we used that $|\rho_w|$ is bounded in the second line. Thus, we get (\ref{eq_AsConstDenPrPow}) and can conclude using Lemma \ref{lem_UnSlutsky}.
\end{proof}

\subsubsection{Proof of Theorem \ref{thm_PowerWGCM1DEst}}
\begin{proof}
Proving Theorem \ref{thm_PowerWGCM1DEst} can be reduced to Theorem \ref{thm_PowerWGCM1DMG}, just as Theorem \ref{thm_WGCM1DEst} could be reduced to Theorem \ref{thm_WGCM1DGenRes}.
\end{proof}

\section{Proof of Theorem \ref{thm_MultFixedWOneDim}}\label{App_ProofsWGCMMultFix}
To prove Theorem \ref{thm_MultFixedWOneDim}, we use the strategy of the proof of Theorem 9 in Section D.4 of \citet{ShahPetersCondInd}. We will heavily rely on results from \citet{ChernozhukovGaussianApprox}, which we summarise in the following.

\subsection{Summary of Results on Gaussian Approximation of Maxima of Random Vectors}\label{sec_ResChern}
We present the following results in the form they are also presented in Section D.4.1 in \citet{ShahPetersCondInd}, which are sometimes special cases of the corresponding results in \citet{ChernozhukovGaussianApprox}.
There is the following difference to the results given there: We consider maxima of absolute values of random vectors instead of maxima of random vectors. The results for the maxima translate to the corresponding results for the absolute value by considering the vector $(Y,-Y)$ instead of just the vector $Y$.

Assume $W\sim \mathcal N_p(0,\Sigma)$, with $\Sigma_{jj}=1$ for $j=1,\ldots, p$. Assume $p\geq 3$ possibly depending on $n$. Let $V=\max_{j=1,\ldots, p}|W_j|$. Let $\tilde W$ be a random vector taking values in $\mathbb R^p$ with $\Ex[\tilde W]=0$ and $\Cov(\tilde W)=\Sigma$ and let $\tilde w_1,\ldots, \tilde w_n \in \mathbb  R^p$ be i.i.d copies of $\tilde W$. Let
$$\tilde V=\max_{j=1,\ldots, p}\left|\frac{1}{\sqrt n}\sum_{i=1}^n \tilde w_{ij}\right|,$$
where $\tilde w_{ij}$ is the $j$th component of $\tilde w_i$.
We need the following conditions for some sequence $(B_n)_{n\in \mathbb N}$ with $B_n \geq 1$:
\begin{itemize}[leftmargin=15mm]
	\item[(B1a)] $\max_{k=1,2}\Ex\left[|\tilde W_j|^{2+k}/B_n^k\right]+\Ex\left[\exp\left(|\tilde W_j|/B_n\right)\right]\leq 4$ for all $j=1,\ldots, p$;
	\item[(B1b)] $\max_{k=1,2}\Ex\left[|\tilde W_j|^{2+k}/B_n^{k/2}\right]+\Ex\left[\max_{j=1,\ldots, p}|\tilde W_j|^4/B_n^2\right]\leq 4$ for all $j=1,\ldots, p$;
	\item[(B2)] There exist some constants $C_1,c_1>0$ such that $B_n^2\left(\log(pn)\right)^7/n\leq C_1 n^{-c_1}$. 
\end{itemize}

The following Lemma corresponds to Corollary 2.1 in \citet{ChernozhukovGaussianApprox} and Theorem 22 in \citet{ShahPetersCondInd}:
\begin{lemma}\label{lem_Cor21Chern}
Assume that either (B1a) and (B2) or (B1b) and (B2) hold. Then, there exist constants $c,C>0$ depending only on $c_1$ and $C_1$ such that
$$\sup_{t\in \mathbb R} |\Prob(\tilde V\leq t)-\Prob(V\leq t)|\leq C n^{-c}.$$
\end{lemma}

The next Lemma corresponds to Lemma 24 in \citet{ShahPetersCondInd} and follows from Lemma C.1 in Section C.5 of \citet{ChernozhukovGaussianApprox} (first statement) and is an application of Lemma A.1 in Section A.1 there (second statement).
\begin{lemma}\label{lem_LemC1Chern}
Again assume that either (B1a) and (B2) or (B1b) and (B2) hold. Let $\tilde \Sigma=\sum_{i=1}^n \tilde w_i \tilde w_i^T/n$ be the empirical covariance matrix of $\tilde w_1,\ldots, \tilde w_n$. Then, there exist constants $c,C$ depending only on $C_1$ and $c_1$ such that
\begin{align*}
\log(p)^2 \Ex\left[\|\tilde \Sigma-\Sigma\|_\infty\right]&\leq Cn^{-c}\\
\log(p)^2\Ex\left[\max_{j=1,\ldots, p}\left|\frac{1}{n}\sum_{i=1}^n \tilde w_{ij}\right|\right]&\leq Cn^{-c}.
\end{align*}
\end{lemma}

The following lemma corresponds to Lemma 2.1 in \citet{ChernozhukovGaussianApprox} and to Lemma 21 in \citet{ShahPetersCondInd}.
\begin{lemma}\label{lem_Lem21Chern}
Let $W=(W_1,\ldots, W_p)^T$ have a multivariate Gaussian distribution with $\Ex[W_i]=0$ and $\Ex[W_i^2]=1$ for all $i=1,\ldots,p$. Then, there exists a  universal constant $C>0$ such that for all $t\geq 0$
$$\sup_{w\geq 0}\Prob\left(\left|\max_{j=1,\ldots, p} |W_j|-w\right|\leq t\right)\leq C t\left(\sqrt{2\log(p)}+1\right).$$
\end{lemma}
Define the function
$$q(\theta)=\theta^{1/3}(2\lor \log(p/\theta))^{2/3}.$$
One may check by differentiating that $q$ is increasing in $\theta$.\footnote{In \citet{ChernozhukovGaussianApprox} and \citet{ShahPetersCondInd}, the definition is $q(\theta)=\theta^{1/3}(1\lor \log(p/\theta))^{2/3}$. Although it is not essential for the proof, we consider it to be convenient for the function $q$ to be increasing in $\theta$.}
The following Lemma corresponds to Lemma 3.1 in \citet{ChernozhukovGaussianApprox} and to Lemma 23 in \citet{ShahPetersCondInd}.
\begin{lemma}\label{lem_Lem31Chern}
Let $U$ and $W$ be centered multivariate Gaussian random vectors taking values in $\mathbb R^p$. Let $U$ have covariance matrix $\Theta$ and $W$ have covariance matrix $\Sigma$, such that for $i=1,\ldots, p$, we have $\Sigma_{ii}=1$. Let $\Delta_0=\max_{j,k=1,\ldots, p}|\Sigma_{jk}-\Theta_{jk}|$. Then, there exists a universal constant $C>0$ such that 
\begin{enumerate}
	\item $\sup_{t\in\mathbb R}|\Prob\left(\max_{j=1,\ldots,p}|U_j|\leq t\right)-\Prob\left(\max_{j=1,\ldots,p}|W_j|\leq t\right)|\leq C q(\Delta_0).$
	\item Let $G_\Sigma$ and $G_\Theta$ be the quantile functions of $\max_{j=1,\ldots,p}|W_j|$ and $\max_{j=1,\ldots,p}|U_j|$. Then, for all $\alpha \in (0,1)$
	$$G_\Theta(\alpha)\leq G_\Sigma(\alpha+C q(\Delta_0))\quad \text{and}\quad G_\Sigma(\alpha)\leq G_\Theta(\alpha+C q(\Delta_0)).$$
\end{enumerate}
\end{lemma}
Note that 2. follows from 1. by observing that 
$$\Prob\left(\max_j|U_j|\leq G_\Sigma\left(\alpha+ C q(\Delta_0)\right)\right)\stackrel{1.}{\geq} \Prob\left(\max_j|W_j| \leq G_\Sigma\left(\alpha+C q(\Delta_0)\right)\right)-C q(\Delta_0)=\alpha,$$
so $G_\Theta(\alpha)\leq G_\Sigma\left(\alpha+Cq(\Delta_0)\right)$. Similarly, the second statement of 2. follows.

\subsection{Proof of Theorem \ref{thm_MultFixedWOneDim}}
We prove a slightly more general result, similar to Section \ref{sec_WGCM1DMG}, since we later also want to apply a version of this theorem and its proof to estimated weight functions, see Appendix \ref{App_ProofMultWGCM}.

For $P\in\mathcal P$ and $C,c>0$, define
$$\mathcal W_{P,C,c}=\left\{w:\mathbb R^{d_Z}\to \mathbb R\, \Bigr\rvert\,|w|\leq C \land \Ex_P\left[\epsilon_P^2\xi_P^2w(Z)^2\right]\geq c\right\}.$$
For $\mathbf w=(w_1,\ldots, w_K)^T\in \mathcal W_{P,C,c}^K$, let $S_{n,\mathbf w}$ and $\hat G_{n,\mathbf w}$ be the versions of $S_n$ and $\hat G_n$ based on $\mathbf w=(w_1,\ldots, w_K)^T$.

\begin{theorem}\label{thm_MultFixedW1DMG}
Let $\mathcal P\subset \mathcal P_0$ and let $A_f$, $A_g$, $B_f$ and $B_g$ be defined as in (\ref{eq_DefAfAg}) and (\ref{eq_DefBfBg}). Assume that there exist $C, c\geq 0$ such that for all $n\in\mathbb N$ and $P\in\mathcal P$ there exists $B\geq 1$ such that either (A1a) and (A2) or (A1b) and (A2) hold. Let $C_1,c_1>0$ such that for all $P\in\mathcal P$ the set $\mathcal W_{P,C_1,c_1}$ is not empty. Assume that
\begin{align}
A_f A_g=o_\mathcal P\left(n^{-1}\log(K)^{-2}\right), \label{eq_MultFixedWOneDimCondA}\\
B_f=o_\mathcal P\left(\log(K)^{-4}\right), \quad B_g=o_\mathcal P\left(\log(K)^{-4}\right).\label{eq_MultFixedWOneDimCondB}
\end{align}
Assume that there exist sequences $\left(\tau_{f,n}\right)_{n\in\mathbb N}$ and $\left(\tau_{g,n}\right)_{n\in\mathbb N}$ such that
\begin{align}
\max_{i=1,\ldots, n}|\epsilon_{P,i}|=O_\mathcal P(\tau_{g,n}),\quad A_{g}=o_\mathcal P\left(\tau_{g,n}^{-2}\log(K)^{-2}\right), \label{eq_MultFixedWOneDimCondC}\\
\max_{i=1,\ldots, n}|\xi_{P,i}|=O_\mathcal P(\tau_{f,n}),\quad A_{f}=o_\mathcal P\left(\tau_{f,n}^{-2}\log(K)^{-2}\right). \label{eq_MultFixedWOneDimCondD}
\end{align}
Then,
$$\sup_{P\in\mathcal P}\sup_{ \mathbf w \in\mathcal W_{P,C_1,c_1}^K} \sup_{\alpha\in (0,1)}|\Prob_P(S_{n,\mathbf w}\leq \hat G_{n,\mathbf w}(\alpha))-\alpha|\to 0.$$
\end{theorem}

We see that Theorem \ref{thm_MultFixedWOneDim} follows from Theorem \ref{thm_MultFixedW1DMG}.

\begin{proof}[Proof of Theorem \ref{thm_MultFixedW1DMG}]
We closely follow Section D.4.2 in \citet{ShahPetersCondInd}.
In the following, we will often omit dependencies on $P$, $\mathbf w$ and $n$ and we will often write $x\lesssim y$ instead of \textit{there exists a constant $C>0$ independent of $n$, $P$ and  $\mathbf w$ such that $x\leq C y$.}

Fix $C_1,c_1 >0$ and write $\mathcal W_P=\mathcal W_{P,C_1,c_1}$. For $P\in\mathcal P$, $\mathbf w=(w_1,\ldots, w_K)^T\in\mathcal W_P^K$ and $k=1,\ldots, K$ let
$$\sigma_{k}^2=\sigma_{P,\mathbf w, k}^2=\Ex_P[\epsilon_P^2\xi_P^2 w_k(Z)^2].$$
We have for all $k=1,\ldots, K$
$$\Ex_P[\epsilon \xi w_k(Z)]=\Ex_P[\Ex_P[\epsilon|Z,Y]\xi w_k(Z)]=0.$$
We will need the conditions (B1a)/(B1b) and (B2) to hold for the random vectors
$$\left\{\left(\frac{\epsilon_i\xi_i w_k(z_i)}{\sigma_{k}}\right)_{k=1}^K\right\}_{i=1}^n.$$
This is guaranteed by (A1a)/(A1b) and (A2), since 
$$|w_k(z_i)/\sigma_k|\leq C_1/\sqrt{c_1}\eqqcolon \gamma.$$

Consider the definition (\ref{eq_DefStatMultFixWOneDim}) of $T_k\upind n$. Write $\tau_{N,k}\upind n/\sigma_k=\sqrt n \bar{\mathbf R}_k/\sigma_k=\delta_k+\tilde T_k$ with
$$\tilde T_k=\frac{1}{\sqrt n} \sum_{i=1}^n \frac{\epsilon_i \xi_i w_k(z_i)}{\sigma_k}.$$
Write
$$\frac{\tau_{D,k}\upind n}{\sigma_k}=\frac{\left(\frac{1}{n}\|\mathbf R_k\|_2^2-\bar{\mathbf R}_k^2\right)^{1/2}}{\sigma_k}=1+\Delta_k.$$
It follows that 
$$T_k=\frac{\tilde T_k+\delta_k}{1+\Delta_k}.$$
Define $\tilde S_n=\max_{k=1,\ldots, K} |\tilde T_k|$ and define the matrix $\Sigma\in \mathbb R^{K\times K}$ by
$$\Sigma_{kl}=\frac{\Ex_P[\epsilon^2\xi^2 w_k(Z)w_l(Z)]}{\sigma_k\sigma_l}.$$
Note that the diagonal of $\Sigma$ consists of ones.
Let $W=(W_1,\ldots, W_K)^T$ be a random vector with distribution $\mathcal N_K(0, \Sigma)$. Let $V_n=\max_{k=1,\ldots, K}|W_k|$ and let $G_n$ be the quantile function of $V_n$. Note that $\Sigma$, $W$, $V_n$ and $G_n$ all depend on $P$ and $\mathbf w$. Lemma \ref{lem_Lem21Chern} implies that the distribution function of $V_n$ is continuous, so for all $\alpha\in[0,1]$, we have $\Prob_P(V_n\leq G_n(\alpha))=\alpha$.
Our goal is to bound 
$$v_{P,\mathbf w}(\alpha)= |\Prob_P(S_{n,\mathbf w}\leq \hat G_{n,\mathbf w}(\alpha))-\alpha|.$$
For this, define 
$$\kappa_{P,\mathbf w}=\sup_{t\in \mathbb R}|\Prob_P(S_{n,\mathbf w}\leq t)-\Prob_P(V_n\leq t)|.$$
Fix $P\in\mathcal P$ and $\mathbf w\in\mathcal W_P$. Let $\triangle$ denote the symmetric difference. Using the triangle inequality and $|\Prob(A)-\Prob(B)|\leq \Prob(A\triangle B)$, we have
\begin{align*}
v(\alpha)&\leq|\Prob(S_n\leq \hat G_n(\alpha))-\Prob(S_n\leq G_n(\alpha))|+|\Prob(S_n\leq G_n(\alpha))-\Prob(V_n\leq G_n(\alpha))|\\
&\leq \Prob(\{S_n\leq \hat G_n(\alpha)\}\triangle \{S_n\leq G_n(\alpha)\})+\kappa.
\end{align*}
Let $u_\Sigma>0$. By Lemma \ref{lem_Lem31Chern}, $\|\Sigma-\hat\Sigma\|_\infty\leq u_\Sigma$ implies $G_n(\alpha-C' q(u_\Sigma))\leq \hat G_n(\alpha)\leq G_n(\alpha+C' q(u_\Sigma))$. This means that if $\|\Sigma-\hat\Sigma\|_\infty\leq u_\Sigma$, we have that $G_n(\alpha)<S_n\leq \hat G_n(\alpha)$ implies $G_n(\alpha)< S_n \leq G_n(\alpha+ C' q(u_\Sigma))$ and $\hat G_n(\alpha)< S_n\leq G_n(\alpha)$ implies $G_n(\alpha-C' q(u_\Sigma))\leq S_n\leq G_n(\alpha)$. We obtain
\begin{align*}
&\Prob\left(\{S_n\leq \hat G_n(\alpha)\}\triangle \{S_n\leq G_n(\alpha)\}\right)\\
&\leq\Prob\left(G_n(\alpha-C'q(u_\Sigma))\leq S_n\leq G_n(\alpha+C'q(u_\Sigma))\right)+\Prob(\|\Sigma-\hat \Sigma\|_\infty > u_\Sigma)\\
&\leq 2\kappa+\Prob(G_n(\alpha-C' q(u_\Sigma))\leq V_n\leq G_n(\alpha+C' q(u_\Sigma)))+\Prob(\|\Sigma-\hat \Sigma\|_\infty > u_\Sigma)\\
&= 2\kappa+2 C' q(u_\Sigma)+\Prob(\|\Sigma-\hat \Sigma\|_\infty > u_\Sigma).
\end{align*}
In total, we obtain 
$$v(\alpha)\lesssim \kappa+q(u_\Sigma)+\Prob(\|\Sigma-\hat \Sigma\|_\infty > u_\Sigma).$$
For $0<u_\delta, u_\Delta \leq 1$, define the event $\Omega=\{\max_{k=1,\ldots, K} |\delta_k|\leq u_\delta, \max_{k=1,\ldots, K}|\Delta_k|\leq u_\Delta\}$. Observe that $S_n = \max_k \left|\frac{\tilde T_k+\delta_k}{1+\Delta_k}\right|$, so on $\Omega$, we have that $S_n\leq t$ implies that $\max_k\frac{\ |\tilde T_k|-u_\delta}{1+u_\Delta}\leq t$ and thus $\tilde S_n=\max_k |\tilde T_k|\leq t(1+u_\Delta)+u_\delta$. Similarly, if $\tilde S_n\leq t(1-u_\Delta)-u_\delta$, then we have $S_n\leq t$.
This implies that for all $t\in \mathbb R$, we have
$$\left\{\tilde S_n\leq t\left(1-{u_\Delta}\right)-u_\delta\right\}\cap\Omega\subseteq\{S_n\leq t\}\cap\Omega\subseteq \left\{\tilde S_n\leq t\left(1+u_\Delta\right)+u_\delta\right\}\cap\Omega$$
Therefore, 
\begin{align*}
\kappa& \leq \sup_{t\in \mathbb R}\left\{\left|\Prob(\tilde S_n\leq t(1+u_\Delta)+u_\delta)-\Prob(V_n\leq t)\right|+\left|\Prob(\tilde S_n\leq t(1-u_\Delta)-u_\delta)-\Prob(V_n\leq t)\right|\right\}\\
&+\Prob(\Omega^c)\\
&\leq \sup_{t\in \mathbb R}\left|\Prob\left(\tilde S_n \leq t\right)-\Prob\left(V_n\leq \frac{t-u_\delta}{1+u_\Delta}\right)\right|+\sup_{t\in \mathbb R}\left|\Prob\left(\tilde S_n \leq t\right)-\Prob\left(V_n\leq \frac{t+u_\delta}{1-u_\Delta}\right)\right|+\Prob(\Omega^c),
\end{align*}
where the last line follows by reparametrisation. Now, we write
\begin{align*}
&\left|\Prob\left(\tilde S_n \leq t\right)-\Prob\left(V_n\leq \frac{t-u_\delta}{1+u_\Delta}\right)\right|\\
&\leq \left|\Prob(\tilde S_n\leq t)-\Prob(V_n\leq t)\right|+\left|\Prob\left(V_n\leq t-u_\delta\right)-\Prob\left(\left(1+u_\Delta\right)V_n\leq t-u_\delta\right)\right|\\
&+\Prob\left(t-u_\delta\leq V_n\leq t\right)\\
&=I+II+III.
\end{align*}
By Lemma \ref{lem_Cor21Chern}, it follows that $I\leq C_2 n^{-c_2}$ with $C_2,c_2>0$ independent of $P$ and $\mathbf w$. By Lemma \ref{lem_Lem31Chern}, we get that $II\leq C' q(u_\Delta^2+2 u_\Delta)$, since $\|\Cov(W)-\Cov((1+u_\Delta)W\|_\infty=u_\Delta^2+2 u_\Delta$. Since $u_\Delta \leq 1$, we have 
$$q(u_\Delta^2+2 u_\Delta)\leq q(3 u_\Delta)\leq 3^{1/3}q(u_\Delta).$$
By Lemma \ref{lem_Lem21Chern}, there exists $C''>0$ such that
$$III\leq \Prob(|V_n-t|\leq u_\delta)\leq C'' u_\delta \sqrt{\log(K)}.$$
Combining $I$, $II$ and $III$, we get that
$$\left|\Prob\left(\tilde S_n \leq t\right)-\Prob\left(V_n\leq \frac{t-u_\delta}{1+u_\Delta}\right)\right|\lesssim q(u_\Delta)+ u_\delta\sqrt{\log(K)}+n^{-c_2}.$$
Similarly, also
$$\left|\Prob\left(\tilde S_n \leq t\right)-\Prob\left(V_n\leq \frac{t+u_\delta}{1-u_\Delta}\right)\right|\lesssim q(u_\Delta)+ u_\delta\sqrt{\log(K)}+n^{-c_2}.$$
In total, we obtain
\begin{align*}
v(\alpha)&\lesssim \Prob\left(\|\Sigma-\hat \Sigma\|_\infty > u_\Sigma\right)+q(u_\Sigma)+\Prob\left(\max_{k=1,\ldots, K} |\delta_k| > u_\delta\right)+u_\delta\sqrt{\log(K)}\\
&+\Prob\left(\max_{k=1,\ldots, K} |\Delta_k|>u_\Delta\right)+q(u_\Delta)+ n^{-c_2}.
\end{align*}
Define 
\begin{equation} \label{eq_DefAn}
a_n=\log(K)^{-2}.
\end{equation}

We first show, how to conclude with the help of Lemma \ref{lem_Lem26SP} below. Take $u_\delta=o(a_n^{1/4})$, such that $\max_{k=1,\ldots, K} |\delta_k|=o_{\mathcal P,\mathcal W}(u_\delta)$. 
Then, we have $\sup_{P\in\mathcal P} \Prob(\max_k |\delta_k|>u_\delta)\to 0$ and $u_\delta\sqrt{\log(K)}=u_\delta a_n^{-1/4}\to 0$. Take $u_\Delta=o(a_n)$, such that $\max_{k=1,\ldots, K} |\Delta_k|=o_{\mathcal P,\mathcal W}(u_\Delta)$. Then, we have $\sup_{P\in\mathcal P} \Prob(\max_k |\Delta_k|>u_\Delta)\to 0$ and 
\begin{align*}
q(u_\Delta)&=u_\Delta^{1/3}\left(2\lor \left(\log(K)+\log(1/u_\Delta)\right)\right)^{2/3}\\
&\leq 2^{2/3} u_\Delta^{1/3}\left(1+\left(a_n^{-1/2}\right)^{2/3}+\log(1/u_\Delta)^{2/3}\right)\to 0,
\end{align*}
since $\left(u_\Delta a_n^{-1}\right)^{1/3}\to 0$ and $\left(u_\Delta\log(1/u_\Delta)^2\right)^{1/3}\to 0$.
The same argument works for the terms involving $u_\Sigma$.
It follows that
$$\sup_{P\in \mathcal P}\sup_{\mathbf w\in\mathcal W_P^K}\sup_{\alpha\in (0,1)}v_{P,\mathbf w}(\alpha)\to 0,$$
which concludes the proof.
\end{proof}

Next, we prove the following Lemma that corresponds to Lemma 26 in \citet{ShahPetersCondInd}.
\begin{lemma}\label{lem_Lem26SP}
For a sequence $(V_n)_{n\in\mathbb N}$ of random variables depending on $\mathbf w\in\mathcal W_P^K$ and a sequence $(b_n)_{n \in \mathbb N}$ of real numbers, let us write
$$V_n=o_{\mathcal P,\mathcal W}(b_n)$$
if for all $\delta >0$
$$\sup_{P\in\mathcal P}\sup_{\mathbf w\in\mathcal W_P^K}\Prob_P\left(\frac{|V_n|}{b_n}>\delta\right)\to 0.$$
Recall the sequence $a_n= \log(K)^{-2}$ from (\ref{eq_DefAn}). Within the setting of the proof of Theorem \ref{thm_MultFixedW1DMG}, we have
\begin{enumerate}
	\item $\max_{k=1,\ldots, K} |\delta_k|=o_{\mathcal P,\mathcal W}(a_n^{1/4})$;
	\item $\max_{k=1,\ldots, K} |\Delta_k|=o_{\mathcal P,\mathcal W}(a_n)$;
	\item $\|\Sigma-\hat \Sigma\|_\infty=o_{\mathcal P,\mathcal W}(a_n)$.
\end{enumerate}
\end{lemma}

\begin{proof}
Define for $i=1,\ldots n$,
\begin{align*}
\Delta f_i&=f(z_i)-\hat f(z_i),\\
\Delta g_i&=g(z_i)-\hat g(z_i).
\end{align*}
We start with 1.
As in the proof of Theorem \ref{thm_WGCM1DGenRes}, we can write $\sigma_k\delta_k=b_k+\nu_{g,k}+\nu_{f,k}$ with
\begin{align*}
b_k&=\frac{1}{\sqrt n}\sum_{i=1}^n w_k(z_i)\Delta f_i\Delta g_i,\\
\nu_{f,k}&=\frac{1}{\sqrt n}\sum_{i=1}^n w_k(z_i)\xi_i\Delta f_i,\\
\nu_{g,k}&=\frac{1}{\sqrt n}\sum_{i=1}^n w_k(z_i)\epsilon_i\Delta g_i.
\end{align*}
Remember that $\sigma_k^2\geq c_1$, $|w_k|\leq C_1$ and $\gamma=C_1/\sqrt{c_1}$.
Using Cauchy-Schwarz, we have $\max_k |b_k|/\sigma_k\leq \gamma\sqrt{n} A_f^{1/2} A_g^{1/2}$. By condition (\ref{eq_MultFixedWOneDimCondA}),
$$\sqrt n A_{f}^{1/2} A_{g}^{1/2}=o_\mathcal{P}(\log(K)^{-2})^{1/2}=o_\mathcal P(a_n^{1/2})= o_\mathcal P(a_n^{1/4}).$$

To control $\max_k|\nu_{g,k}|/\sigma_k$, we use Lemma \ref{lem_Lem29SP}, see below.
For $\delta >0$
\begin{align*}
\Prob\left(\max_{k}\frac{|\nu_{g,k}|}{\sigma_k a_n^{1/4}}> \delta\right)&=\Prob\left(\max_k \left|\frac{1}{\sqrt n} \sum_{i=1}^n \frac{w_k(z_i)\epsilon_i\Delta g_i}{\sigma_k a_n^{1/4}}\right|>\delta\right)\\
&\lesssim \frac{1}{\delta}\Ex\left[\delta\land \tau \sqrt{\log(K)}\left(\max_k\frac{1}{n a_n^{1/2}\sigma_k^2}\sum_{i=1}^n w_k(z_i)^2 \Delta g_i^2\right)^{1/2}\right]\\
&+\Prob(\max_i|\epsilon_i|\geq \tau)
\end{align*}
for all $\tau\geq 0$.
By condition (\ref{eq_MultFixedWOneDimCondC}), $\max_{i}|\epsilon_{i}|=O_\mathcal P(\tau_{g,n})$, so for any $\delta' >0$, there exists $D>0$ such that $\sup_{P\in\mathcal P}\Prob_P(\max_{i}|\epsilon_{i}|>D\tau_{g,n})<\delta'$ for all $n$. We have that
\begin{align*}
D\tau_{g,n} \sqrt{\log(K)}\left(\max_k\frac{1}{n a_n^{1/2}\sigma_k^2}\sum_{i=1}^n w_k(z_i)^2 \Delta g_i^2\right)^{1/2}&\leq D\tau_{g,n} \sqrt{\log(K)}\left(\frac{\gamma^2}{a_n^{1/2}}A_{g}\right)^{1/2}\\
&= D\gamma \left(\tau_{g,n}^2 \log(K)\frac{1}{\log(K)^{-1}}  A_{g}\right)^{1/2}\\
&= o_{\mathcal P,\mathcal W}(1)
\end{align*}
by condition (\ref{eq_MultFixedWOneDimCondC}).
Using bounded convergence (Lemma \ref{lem_BoundConv}), we get that 
$$\limsup_{n\to\infty}\sup_{P\in \mathcal P}\sup_{\mathbf w\in\mathcal W_P^K}\Prob_P\left(\max_{k}\frac{|\nu_{g,k}|}{\sigma_k a_n^{1/4}}\geq \delta\right) \lesssim \delta'.$$
Since $\delta,\delta'>0$ are arbitrary, we get $\max_{k}{|\nu_{g,k}|}/{\sigma_k }=o_{\mathcal P,\mathcal W}(a_n^{1/4})$. Similarly, we also have $\max_{k}{|\nu_{f,k}|}/{\sigma_k }=o_{\mathcal P,\mathcal W}(a_n^{1/4})$. This concludes the proof of 1.

For 2., note that by definition $1+\Delta_k=(\|\mathbf R_k\|_2^2/n-\bar{\mathbf R}_k^2)^{1/2}/\sigma_k$ and thus,
\begin{align*}
\max_k\left|\left(1+\Delta_k\right)^2-1\right|&\leq \max_k\left|\frac{\|\mathbf R_k\|_2^2}{n\sigma_k^2}-1\right|+\max_k \frac{\bar{\mathbf R}_k^2}{\sigma_k^2}\\
&\leq \max_k\left|\frac{\|\mathbf R_k\|_2^2}{n\sigma_k^2}-\tilde\Sigma_{kk}\right|+\max_k\left|\tilde\Sigma_{kk}-1\right|+\max_k \frac{\bar{\mathbf R}_k^2}{\sigma_k^2},
\end{align*}
with 
$$\tilde \Sigma_{kl}=\frac{1}{n\sigma_k\sigma_l}\sum_{i=1}^n \epsilon_i^2\xi_i^2 w_k(z_i) w_l(z_i)$$
for $k,l=1,\ldots, K$.
The first term is $o_{\mathcal P,\mathcal W}(a_n)$ by Lemma \ref{lem_Lem27SP} below. The second term is $o_{\mathcal P,\mathcal W}(a_n)$ using Lemma \ref{lem_LemC1Chern}.  For the third term, write
$$\max_k\frac{|\bar{\mathbf R}_k|}{\sigma_k}\leq \max_k\frac{|\delta_k|}{\sqrt n}+\max_k\left|\frac{1}{n\sigma_k}\sum_{i=1}^n \epsilon_i\xi_i w_k(z_i)\right|.$$
Observe that by condition (A2), we have $\log(K)^3/n\leq \log(Kn)^7/n\lesssim n^{-c}=o(1)$. Thus, $\log(K)^3=o(n)$. Using this and 1., we have 
$$\max_k \frac{|\delta_k|}{\sqrt n a_n}=\max_k\frac{|\delta_k|}{ a_n^{1/4}}\frac{1}{\sqrt n a_n^{3/4}}=o_{\mathcal P,\mathcal W}(1)\frac{\log(K)^{3/2}}{n^{1/2}}=o_{\mathcal P,\mathcal W}(1) o(1).$$
Together with Lemma \ref{lem_LemC1Chern}, we obtain
\begin{equation}\label{eq_OrderAvRk}
\max_k\frac{|\bar{\mathbf R}_k|}{\sigma_k}=o_{\mathcal P,\mathcal W}(a_n) = o_{\mathcal P,\mathcal W}(a_n^{1/2}).
\end{equation}
In total, we get
$$\max_k\left|\left(1+\Delta_k\right)^2-1\right|=o_{\mathcal P,\mathcal W}(a_n).$$
Using Lemma \ref{lem_Lem28SP} below with $f(x)=\sqrt{x+1}-1$, we obtain
$\max_k |\Delta_k|=o_{\mathcal P,\mathcal W} (a_n)$. This proves 2.

For 3., Lemma \ref{lem_LemC1Chern} implies $\|\Sigma-\tilde\Sigma\|_\infty=o_{\mathcal P,\mathcal W}(a_n)$, so it is enough to show $\|\hat\Sigma-\tilde\Sigma\|_\infty=o_{\mathcal P,\mathcal W}(a_n)$. Observe that
\begin{equation}\label{eq_BoundPartTildeSigma}
\max_{k,l}\left|\frac{\mathbf R_k^T \mathbf R_l}{n\sigma_k\sigma_l}-\frac{\bar{\mathbf R}_k\bar{\mathbf R}_l}{\sigma_k\sigma_l}-\tilde \Sigma_{kl}\right|=o_\mathcal P(a_n),
\end{equation}
using Lemma \ref{lem_Lem27SP} and (\ref{eq_OrderAvRk}).
By 2. and Lemma \ref{lem_Lem28SP}, it follows that
$$\max_k\left|\left(1+\Delta_k\right)^{-1}-1\right|=o_\mathcal P(a_n).$$
This implies that also
\begin{align}
&\max_{k,l}\left|\left(1+\Delta_k\right)^{-1}\left(1+\Delta_l\right)^{-1}-1\right|\nonumber\\ 
&\leq \max_{k,l}|1+\Delta_l|^{-1}\left|(1+\Delta_k)^{-1}-1\right|+\max_{l}\left|(1+\Delta_l)^{-1}-1\right|\nonumber\\ 
&=\left(1+o_{\mathcal P,\mathcal W}(a_n)\right)o_{\mathcal P,\mathcal W} (a_n)+o_{\mathcal P,\mathcal W}(a_n)=o_{\mathcal P,\mathcal W}(a_n). \label{eq_ProdDeltas}
\end{align}
Putting things together, we have
\begin{align*}
&\max_{k,l}\left|\hat \Sigma_{kl}-\tilde\Sigma_{kl}\right|\\
&=\max_{k,l}\left|\left(\frac{\mathbf R_k^T \mathbf R_l}{n\sigma_k\sigma_l}-\frac{\bar{\mathbf R}_k\bar{\mathbf R}_l}{\sigma_k\sigma_l}\right)\left(1+{\Delta_k}\right)^{-1}\left(1+{\Delta_l}\right)^{-1}-\tilde\Sigma_{kl}\right|\\
&\leq \max_{k,l}\left|\frac{\mathbf R_k^T \mathbf R_l}{n\sigma_k\sigma_l}-\frac{\bar{\mathbf R}_k\bar{\mathbf R}_l}{\sigma_k\sigma_l}\right|\left|\left(1+{\Delta_k}\right)^{-1}\left(1+{\Delta_l}\right)^{-1}-1\right|+\max_{k,l}\left|\frac{\mathbf R_k^T \mathbf R_l}{n\sigma_k\sigma_l}-\frac{\bar{\mathbf R}_k\bar{\mathbf R}_l}{\sigma_k\sigma_l}-\tilde \Sigma_{kl}\right|
\end{align*}

Using (\ref{eq_BoundPartTildeSigma}), Lemma \ref{lem_LemC1Chern} and (\ref{eq_ProdDeltas}), the first term is $O_ {\mathcal P,\mathcal W}(1)o_{\mathcal P,\mathcal W}(a_n)=o_{\mathcal P,\mathcal W}(a_n)$ and the second term is also $o_{\mathcal P,\mathcal W}(a_n)$ by (\ref{eq_BoundPartTildeSigma}). This proves 3.
\end{proof}

It remains to prove the following Lemma corresponding to Lemma 27 in \citet{ShahPetersCondInd}.

\begin{lemma}\label{lem_Lem27SP}
For $k,l=1,\ldots, K$, let
$$\tilde \Sigma_{kl}=\frac{1}{n\sigma_k\sigma_l}\sum_{i=1}^n \epsilon_i^2\xi_i^2 w_k(z_i) w_l(z_i).$$
Then,
$$\max_{k,l=1,\ldots, K}\left|\frac{\mathbf R_k^T \mathbf R_l}{n\sigma_k\sigma_l}-\tilde \Sigma_{kl}\right|=o_{\mathcal P,\mathcal W} (a_n).$$
\end{lemma}

\begin{proof}
With $\mathbf R_k=(R_{k1},\ldots, R_{kn})^T$, write for $k,l=1,\ldots, K$ and $i=1,\ldots, n$
\begin{align*}
&\frac{1}{\sigma_k\sigma_l}\left|R_{ki}R_{li}-\epsilon_i^2\xi_i^2 w_k(z_i) w_l(z_i)\right|\\
&=\frac{1}{\sigma_k\sigma_l}\left|(\Delta f_i+\epsilon_i)^2(\Delta g_i+\xi_i)^2w_k(z_i)w_l(z_i)-\epsilon_i^2\xi_i^2 w_k(z_i) w_l(z_i)\right|\\
&\leq \gamma^2\left(\Delta f_i^2\Delta g_i^2\right.\\
&+ 2|\Delta f_i^2\Delta g_i \xi_i|+2|\Delta f_i \epsilon_i \Delta g_i^2|\\
&+\Delta f_i^2\xi_i^2+\epsilon_i^2\Delta g_i^2\\
&+ 2|\epsilon_i^2 \Delta g_i \xi_i|+2|\Delta f_i \epsilon_i\xi_i^2|\\
&+\left. 4|\Delta f_i \epsilon_i\Delta g_i \xi_i|\vphantom{\Delta f_i^2\Delta g_i^2}\right).
\end{align*}
We show that the sum over each of the eight terms is $o_\mathcal P(a_n)$ individually. Since the terms on the right hand side of the inequality do not depend on $\mathbf w\in\mathcal W_P^K$ anymore, this implies that the sum over the left hand side is $o_{\mathcal P,\mathcal W}(a_n)$. Note that by symmetry it is enough to control only one term in each line.

To start, observe
$$\frac{1}{n}\sum_{i=1}^n \Delta f_i^2 \Delta g_i^2\leq n A_f A_g=o_\mathcal P(a_n).$$
For the second term, $2|\Delta f_i^2 \Delta g_i \xi_i|\leq \Delta f_i^2\Delta g_i^2+\Delta f_i^2\xi_i^2$. Observe that similarly to the proof of Theorem \ref{thm_WGCM1DGenRes}
\begin{align*}
\Prob\left(\frac{1}{n a_n^2}\sum_{i=1}^n \Delta f_i^2 \xi_i^2 \geq \delta\right)&\leq \frac{1}{\delta}\Ex\left[\Ex\left[\frac{1}{n a_n^2}\sum_{i=1}^n \Delta f_i^2 \xi_i^2|\mathbf X,\mathbf Z\right]\land \delta\right]\\
&\leq \frac{1}{\delta}\Ex\left[\frac{1}{n a_n^2}\sum_{i=1}^n \Delta f_i^2 v(z_i)\land \delta\right]\\
&\leq \frac{1}{\delta}\Ex\left[\frac{1}{a_n^2} B_f\land \delta\right] \to 0,
\end{align*}
using condition (\ref{eq_MultFixedWOneDimCondB}) and bounded convergence (Lemma \ref{lem_BoundConv}), so
\begin{equation}\label{eq_TermOrderAn2}
\frac{1}{n}\sum_{i=1}^n \Delta f_i^2 \xi_i^2=o_\mathcal P\left(a_n^2\right)=o_\mathcal P\left(a_n\right).
\end{equation}
This can also be used for the third line. For the fourth line, we use Cauchy-Schwarz to write
$$\frac{1}{n}\sum_{i=1}^n\Delta f_i \epsilon_i\xi_i^2\leq \left(\frac{1}{n}\sum_{i=1}^n\Delta f_i^2\xi_i^2 \right)^{1/2}\left(\frac{1}{n}\sum_{i=1}^n\epsilon_i^2\xi_i^2 \right)^{1/2}.$$
The first factor is $o_\mathcal P(a_n)$ by (\ref{eq_TermOrderAn2}) and the second factor is $O_\mathcal P(1)$ by Lemma \ref{lem_LemC1Chern}, so the product is $o_\mathcal P(a_n)$.
Finally,
$$\frac{1}{n}\sum_{i=1}^n |\Delta f_i\epsilon_i\Delta g_i\xi_i|\leq \left(\frac{1}{n}\sum_{i=1}^n\Delta f_i^2\xi_i^2 \right)^{1/2}\left(\frac{1}{n}\sum_{i=1}^n\Delta g_i^2\epsilon_i^2 \right)^{1/2}=o_\mathcal P(a_n)o_\mathcal P(a_n)=o_\mathcal P(a_n).$$
This completes the proof of Lemma \ref{lem_Lem27SP} and thus also the proof of Theorem \ref{thm_MultFixedWOneDim}.
\end{proof}

\subsection{Some Additional Lemmas}
The next two lemmas are also taken from \citet{ShahPetersCondInd}, where they appear as Lemma 28 and Lemma 29.

\begin{lemma}\label{lem_Lem28SP}
Let $\mathcal P$ be a collection of distributions and for all $n\in\mathbb N$, let $W\upind n$ be a random vector taking values in $\mathbb R^{p_n}$ and let $(a_n)_{n\in\mathbb N}$ be a bounded sequence of positive numbers. Assume that $\max_{j=1,\ldots, p_n} |W_j\upind n| = o_\mathcal P(a_n)$. Let $D\subset \mathbb R$ such that $0$ is in the interior of $D$ and let $f:D\to \mathbb R$ be continuously differentiable at $0$ with $f(0)=c$. Then,
$$\max_{j=1,\ldots, p_n}\left|f\left(W_j\upind n\right)-c\right|=o_\mathcal P(a_n).$$
\end{lemma}

\begin{lemma}\label{lem_Lem29SP}
Let $W\in \mathbb R^{n\times p}$, $V\in \mathbb R^{n\times p}$ be random matrices such that $\Ex[W|V]=0$ and the rows of $W$ are independent conditional on $V$. Then, for all $\epsilon > 0$
$$\epsilon \Prob\left(\max_j \left|\frac{1}{\sqrt n}\sum_{i=1}^n W_{ij} V_{ij}\right|>\epsilon\right)\lesssim \Ex\left[\epsilon\wedge \lambda\sqrt{\log p}\left(\max_{j} \frac{1}{n} \sum_{i=1}^n V_{ij}^2\right)^{1/2}\right]+\epsilon \Prob(\|W\|_\infty >\lambda)$$
for any $\lambda \geq 0$.
\end{lemma}

\section{Proofs of Appendix \ref{sec_MultWGCM}}\label{App_ProofMultWGCM}
In this section, we give the proofs of the results on the multivariate WGCM.
\subsection{Proof of Theorem \ref{thm_MultFixedWMultDim}}
As for Theorem \ref{thm_MultFixedWOneDim}, we prove a slightly more general result corresponding to Theorem \ref{thm_MultFixedW1DMG}. For a collection of weight functions $w_{jlk}:\mathbb R^{d_Z}\to\mathbb R$, write
$$\mathbf w=\left(w_{jlk}\right)_{j=1,\ldots, d_X, l=1,\ldots, d_Y, k=1,\ldots, K(j,l)}.$$
For $P\in\mathcal P$ and $C,c>0$, define
$$\mathcal W_{P,C,c}=\left\{\mathbf w = \left(w_{jlk}\right)_{j,l,k}\Bigr\rvert\,\forall j,l,k: |w_{jlk}|\leq C \land \Ex_P\left[\epsilon_j^2\xi_l^2w_{jlk}(Z)^2\right]\geq c\Ex_P\left[\epsilon_j^2\xi_l^2\right]\right\}.$$
For $\mathbf w\in \mathcal W_{P,C,c}$, let $S_{n,\mathbf w}$ and $\hat G_{n,\mathbf w}$ be the versions of $S_n$ and $\hat G_n$ based on $\mathbf w=\left(w_{jlk}\right)_{j,l,k}$.

\begin{theorem}\label{thm_MultFixedWMultDMG}
Let $\mathcal P\subset \mathcal P_0$ and let $A_{f,j}$ and $A_{g,l}$ be defined as in (\ref{eq_DefAfj}) and (\ref{eq_DefAgl}). Assume that there exist $C,c> 0$ such that for all $n\in\mathbb N$ and $P\in\mathcal P$ there exists $D_n\geq 1$ such that either (C1a) and (C2) or (C1b) and (C2) hold. Let $C_1,c_1>0$ such that for all $P\in\mathcal P$ the set $\mathcal W_{P,C_1,c_1}$ is not empty. Assume that 
\begin{align}
\max_{j,l}\frac{1}{\sigma_{jl}^2}A_{f, j} A_{g,l}=o_\mathcal P\left(n^{-1}\log(\mathbf K)^{-4}\right). \label{eq_MultFixedWMultDimCondA}
\end{align}
Assume that there exist sequences $\left(\tau_{f,n}\right)_{n\in\mathbb N}$ and $\left(\tau_{g,n}\right)_{n\in\mathbb N}$ as well as positive real numbers $s_{g,jl}$, $t_{g,jl}$, $s_{f,jl}$ and  $t_{f,jl}$ possibly depending on $P\in\mathcal P$ such that for all $j=1,\ldots, d_X$, $l=1,\ldots, d_Y$
$$s_{f,jl}t_{f,jl}=\sigma_{jl},\quad s_{g,jl}t_{g,jl}=\sigma_{jl},$$
and such that
\begin{align}
\max_{i,j,l}|\epsilon_{P,ij}|/t_{g,jl}=O_\mathcal P(\tau_{g,n}),\quad \max_{j,l}A_{g, l}/s_{g,jl}^2=o_\mathcal P\left(\tau_{g,n}^{-2}\log(\mathbf K)^{-4}\right) \label{eq_MultFixedWMultDimCondC}\\
\max_{i,j,l}|\xi_{P,il}|/t_{f,jl}=O_\mathcal P(\tau_{f,n}),\quad \max_{j,l}A_{f,j}/s_{f,jl}^2=o_\mathcal P\left(\tau_{f,n}^{-2}\log(\mathbf K)^{-4}\right). \label{eq_MultFixedWMultDimCondD}
\end{align}
Then,
$$\sup_{P\in\mathcal P}\sup_{\mathbf w\in \mathcal W_{P,C_1,c_1}}\sup_{\alpha\in (0,1)}|\Prob_P(S_{n,\mathbf w}\leq \hat G_{n,\mathbf w}(\alpha))-\alpha|\to 0.$$
\end{theorem}

We see that Theorem \ref{thm_MultFixedWMultDim} follows from Theorem \ref{thm_MultFixedWMultDMG}.

\subsubsection{Proof of Theorem \ref{thm_MultFixedWMultDMG}}
\begin{proof}
The proof is along the same lines as the proof of Theorem \ref{thm_MultFixedW1DMG} with the complication of having more indices. We therefore just present the parts that require extra care compared to the earlier proof.

Define 
$$\bar \sigma_{jlk}^2=\bar \sigma_{P,\mathbf w, jlk}^2=\Ex_P\left[\epsilon_j^2\xi_l^2 w_{jlk}(Z)^2\right].$$
We will need to apply the results from Section \ref{sec_ResChern} to the random vectors
$$\left\{\left(\frac{\epsilon_j\xi_l w_{jlk}(z_i)}{\bar \sigma_{jlk}}\right)_{j,l,k}\right\}_{i=1}^n.$$
The conditions (B1a)/(B1b) and (B2) are satisfied by (C1a)/(C1b) and (C2) using that
$$c_1 \sigma_{jl}^2\leq \bar \sigma_{jlk}^2\leq C_1^2\sigma_{jl}^2$$
In exactly the same way as in the proof of Theorem \ref{thm_MultFixedWOneDim}, the theorem can be reduced to the following lemma:
\begin{lemma}\label{lem_Lem26SPMultW}
Let $a_n=\log(\mathbf K)^{-2}$ and let $\delta_{jlk}$ and $\Delta_{jlk}$ be defined analogously to the proof of Theorem \ref{thm_MultFixedW1DMG}. Then,
\begin{enumerate}
	\item $\max_{j,l,k} |\delta_{jlk}|=o_{\mathcal P,\mathcal W}\left(a_n^{1/4}\right)$;
	\item $\max_{j,l,k} |\Delta_{jlk}|=o_{\mathcal P,\mathcal W}\left(a_n\right)$;
	\item $\|\Sigma-\hat \Sigma\|_\infty=o_{\mathcal P,\mathcal W}\left(a_n\right)$.
\end{enumerate}
\end{lemma}
The proof of this Lemma is similar to the proof of Lemma \ref{lem_Lem26SP}. Extra care has to be taken in part 1. for the control of $\max_{j,l,k}|\nu_{g,jlk}|/\bar\sigma_{jlk}$, where $\nu_{g,jlk}=\sqrt n^{-1}\sum_{i=1}^n w_{jlk}(z_i)\epsilon_{ij}\Delta g_{ij}$. For this, one also uses Lemma \ref{lem_Lem29SP} to write for all $\delta >0$
\begin{align*}
\Prob\left(\max_{j,l,k}\frac{|\nu_{g,jlk}|}{\bar\sigma_{jlk} a_n^{1/4}}\geq \delta\right)&= \Prob\left(\max_{j,l,k} \left|\frac{1}{\sqrt n} \sum_{i=1}^n \frac{w_{jlk}(z_i)\epsilon_{ij}\Delta g_{il}}{\bar\sigma_{jlk} a_n^{1/4}}\right|>\delta\right)\\
&\leq \Prob\left(\max_{j,l,k} \left|\frac{1}{\sqrt n} \sum_{i=1}^n \frac{w_{jlk}(z_i)\epsilon_{ij}\Delta g_{il}}{\sqrt {c_1}\sigma_{jl} a_n^{1/4}}\right|>\delta\right)\\
&\lesssim \frac{1}{\delta}\Ex\left[\delta\land \tau \sqrt{\log(\mathbf K)}\left(\max_{j,l}\frac{1}{n a_n^{1/2}s_{g,jl}^2 {c_1}}\sum_{i=1}^n C_1^2 \Delta g_{il}^2\right)^{1/2}\right]\\
&+\Prob\left(\max_{i,j,l}|\epsilon_{ij}| /t_{g,jl}\geq \tau\right)
\end{align*}
for all $\tau\geq 0$. From this, one can proceed as in the proof of Lemma \ref{lem_Lem26SP}.
The rest of the proof of Lemma \ref{lem_Lem26SPMultW} also works as before, with the difference that
$$\tilde\Sigma_{jlk,j'l'k'}=\frac{1}{n\bar \sigma_{jlk}\bar \sigma_{j'l'k'}}\sum_{i=1}^n\epsilon_{ij}\epsilon_{ij'}\xi_{il}\xi_{il'}w_{jlk}(z_i) w_{j'l'k'}(z_i).$$
With this definition, the equivalent of Lemma \ref{lem_Lem27SP} is the following:
\begin{lemma}\label{lem_Lem27SPMult}
$$\max_{jlk,j'l'k'}\left|\frac{\mathbf R_{jlk}^T \mathbf R_{j'l'k'}}{n\bar \sigma_{jlk}\bar \sigma_{j'l'k'}}-\tilde \Sigma_{jlk,j'l'k'}\right|=o_{\mathcal P,\mathcal W} (a_n).$$
\end{lemma}
The idea of the proof of this lemma is similar to the proof of Lemma \ref{lem_Lem27SP}, but due to the slight difference in the definition of $\tilde \Sigma$, we redo the proof. For all $j,l,k,j',l',k'$ and $i=1,\ldots, n$ and omitting the dependence on $i$ from the second line on, we can write
\begin{align*}
&\frac{1}{\bar\sigma_{jlk}\bar\sigma_{j'l'k'}}\left|R_{jlk,i}R_{j'l'k',i}-\epsilon_{ij}\epsilon_{ij'}\xi_{il}\xi_{il'}w_{jlk}(z_i) w_{j'l'k'}(z_i)\right|\\
&=\frac{1}{\bar\sigma_{jlk}\bar\sigma_{j'l'k'}}\left|\left[(\Delta f_{j}+\epsilon_{j})(\Delta f_{j'}+\epsilon_{j'})(\Delta g_{l}+\xi_{l})(\Delta g_{l'}+\xi_{l'})-\epsilon_{j}\epsilon_{j'}\xi_{l}\xi_{l'}\right]w_{jlk}(z_i) w_{j'l'k'}(z_i)\right|\\
&\leq\frac{C_1^2}{c_1\sigma_{jl}\sigma_{j'l'}}\left(|\Delta f_j\Delta f_{j'}\Delta g_l\Delta g_{l'}|\right.\\
&+|\Delta f_j\Delta f_{j'}\Delta g_l\xi_{l'}|+|\Delta f_j\Delta f_{j'}\xi_l\Delta g_{l'}|+|\Delta f_j\epsilon_{j'}\Delta g_l\Delta g_{l'}|+|\epsilon_j\Delta f_{j'}\Delta g_l\Delta g_{l'}|\\
&+|\Delta f_j\Delta f_{j'}\xi_l\xi_{l'}|+|\epsilon_j\epsilon_{j'}\Delta g_l\Delta g_{l'}|\\
&+|\Delta f_j\epsilon_{j'}\xi_l\Delta g_{l'}|+|\epsilon_j\Delta f_{j'}\Delta g_l\xi_{l'}|\\
&+|\Delta f_j\epsilon_{j'}\Delta g_l\xi_{l'}|+|\epsilon_j\Delta f_{j'}\xi_l\Delta g_{l'}|\\
&\left. +|\Delta f_j\epsilon_{j'}\xi_l\xi_{l'}|+|\epsilon_j\Delta f_{j'}\xi_l\xi_{l'}|+|\epsilon_j\epsilon_{j'}\Delta g_l\xi_{l'}|+|\epsilon_j\epsilon_{j'}\xi_l\Delta g_{l'}|\right).
\end{align*}
We control the sum over all fifteen terms individually. By symmetry, it is enough to control one term in each line.
For the first term, observe that 
$$\frac{2}{\sigma_{jl}\sigma_{j'l'}}|\Delta f_j\Delta f_{j'}\Delta g_l\Delta g_{l'}|\leq \frac{1}{\sigma_{jl}^2}\Delta f_j^2\Delta g_l^2+\frac{1}{\sigma_{j'l'}^2}\Delta f_{j'}^2\Delta g_{l'}^2.$$
We have that
\begin{equation}\label{eq_Term1OrderAn2}
\max_{j,l}\frac{1}{n}\sum_{i=1}^n \frac{1}{\sigma_{jl}^2}\Delta f_{ij}^2\Delta g_{il}^2\leq\max_{j,l}\frac{n}{\sigma_{jl}^2}A_{f,j} A_{g,l}=o_\mathcal P(a_n^2)=o_\mathcal P(a_n)
\end{equation}
by condition (\ref{eq_MultFixedWMultDimCondA}).

For the second line, we have
$$\frac{2}{\sigma_{jl}\sigma_{j'l'}}|\Delta f_j \Delta f_{j'} \Delta g_l\xi_{l'}|\leq \frac{1}{\sigma_{jl}^2} \Delta f_j^2\Delta g_l^2+ \frac{1}{\sigma_{j'l'}^2} \Delta f_{j'}^2\xi_{l'}^2.$$
We show that the maximum of the sum over the second term is $o_\mathcal P(a_n^2)$. Let $\delta > 0$. Then,
\begin{align*}
\Prob\left(\max_{j',l'}\frac{1}{n a_n^2}\sum_{i=1}^n\frac{1}{\sigma_{j'l'}^2}\Delta f_{ij'}^2\xi_{il'}^2\geq \delta\right)&\leq \frac{1}{\delta}\Ex\left[\max_{j',l'}\frac{1}{n a_n^2}\sum_{i=1}^n\frac{1}{\sigma_{j'l'}^2}\Delta f_{ij'}^2\xi_{il'}^2\land \delta\right]
\end{align*}
Using condition (\ref{eq_MultFixedWMultDimCondD}), we have
\begin{align*}
\max_{j',l'}\frac{1}{n a_n^2}\sum_{i=1}^n\frac{1}{\sigma_{j'l'}^2}\Delta f_{ij'}^2\xi_{il'}^2&\leq \max_{i,j',l'}\frac{\xi_{il'}^2 }{t_{f,j'l'}^2}\max_{j',l'}\frac{1}{n a_n^2 s_{f, j'l'}^2}\sum_{i=1}^n \Delta f_{ij'}^2\\
&=O_\mathcal P(\tau_{f,n}^2)\max_{j',l'} \frac{A_{f,j'}}{a_n^2 s_{f,j'l'}^2}\\
&=o_\mathcal P(1).
\end{align*}
We can conclude using bounded convergence (Lemma \ref{lem_BoundConv}) that
\begin{equation}\label{eq_Term2OrderAn2}
\max_{j',l'}\frac{1}{n}\sum_{i=1}^n\frac{1}{\sigma_{j'l'}^2}\Delta f_{ij'}^2\xi_{il'}^2=o_\mathcal P(a_n^2)=o_\mathcal P(a_n).
\end{equation}

For the third and the fourth line, we can write
\begin{align*}
\frac{2}{\sigma_{jl}\sigma_{j'l'}}|\Delta f_j \Delta f_{j'} \xi_l\xi_{l'}|&\leq \frac{1}{\sigma_{jl}^2} \Delta f_j^2\xi_l^2+ \frac{1}{\sigma_{j'l'}^2} \Delta f_{j'}^2\xi_{l'}^2\\
\frac{2}{\sigma_{jl}\sigma_{j'l'}}|\Delta f_j \epsilon_{j'}  \xi_l\Delta g_{l'}|&\leq \frac{1}{\sigma_{jl}^2} \Delta f_j^2\xi_{l}^2+ \frac{1}{\sigma_{j'l'}^2} \epsilon_{j'}^2\Delta g_{l'}^2,
\end{align*}
so this can be controlled as the term before.
For the fifth line, by Cauchy-Schwarz . 
\begin{align*}
\max_{jl,j'l'}\frac{1}{n}\sum_{i=1}^n\frac{1}{\sigma_{jl}\sigma_{j'l'}}|\Delta f_{ij}\epsilon_{ij'}\Delta g_{il}\xi_{il'}|
&\leq \max_{j,l}\sqrt{\frac{1}{n}\sum_{i=1}^n\frac{1}{\sigma_{jl}^2}\Delta f_{ij}^2\Delta g_{il}^2}\max_{j',l'}\sqrt{\frac{1}{n}\sum_{i=1}^n\frac{1}{\sigma_{j'l'}^2}\epsilon_{ij'}^2\xi_{il'}^2}.
\end{align*}
The first factor is $o_\mathcal P(a_n)$ by (\ref{eq_Term1OrderAn2}). The second factor is $O_\mathcal P(1)$ by Lemma \ref{lem_LemC1Chern}, so the product is $o_\mathcal P(a_n)$.

Note that if we are in the setting where $\sigma_{jl}\sigma_{j'l'}\geq C_3 \sigma_{jl'}\sigma_{j'l}$, as described in Remark \ref{rmk_ThmMultFixedWMultDim}, 3.,  it is better pair together $\Delta f_j\xi_{l'}$ and $\epsilon_{j'}\Delta g_l$ and treat it as the third and fourth line. In this way, we only need to have $\max_{j,l}A_{f,j} A_{g,l}=o_\mathcal P(n^{-1}\log(\mathbf K)^{-2})$ instead of an exponent of $-4$, see equation (\ref{eq_Term1OrderAn2}).

For the last line, observe
\begin{align*}
\max_{jl,j'l'}\frac{1}{n}\sum_{i=1}^n\frac{1}{\sigma_{jl}\sigma_{j'l'}}|\Delta f_{ij}\epsilon_{ij'}\xi_{il}\xi_{il'}|
&\leq \max_{j,l}\sqrt{\frac{1}{n}\sum_{i=1}^n\frac{1}{\sigma_{jl}^2}\Delta f_{ij}^2\xi_{il}^2}\max_{j',l'}\sqrt{\frac{1}{n}\sum_{i=1}^n\frac{1}{\sigma_{j'l'}^2}\epsilon_{ij'}^2\xi_{il'}^2},
\end{align*}
from which we conclude as before using (\ref{eq_Term2OrderAn2}). This concludes the proof of Lemma \ref{lem_Lem27SPMult}.
\end{proof}

\subsection{Proof of Theorem \ref{thm_MultEstWMultDim}}
\begin{proof}
Theorem \ref{thm_MultEstWMultDim} follows from Theorem \ref{thm_MultFixedWMultDMG} just as Theorem \ref{thm_WGCM1DEst} followed from Theorem \ref{thm_WGCM1DGenRes}. In the setting of Theorem \ref{thm_MultFixedWMultDMG}, for $P\in\mathcal P$ and $C_1,c_1>0$, $\mathbf w\in\mathcal  W_{P,C_1,c_1}$ and $\alpha\in (0,1)$, let
$$\Gamma_n(P,\mathbf w, \alpha)=\Prob_P\left(S_{n,\mathbf w}\leq \hat G_{n,\mathbf w}(\alpha)\right).$$
Then, we have
$$\sup_{P\in\mathcal P}\sup_{\mathbf w\in\mathcal W_{P,C_1,c_1}}\sup_{\alpha\in(0,1)}|\Gamma_n(P,\mathbf w,\alpha)-\alpha|\to 0.$$
For the functions $\left(\hat w_{jlk}\upind n\right)_{j,l,k}$ estimated on the auxiliary data set $\mathbf A$, we know that for all $P\in\mathcal P$, $P$-almost surely for all $n\in\mathbb N$, we have $\left(\hat w_{jlk}\upind n \right)_{j,l,k}\in\mathcal W_{P,C_1,c_1}$. Since $\mathbf A$ is independent of $(\mathbf X\upind n, \mathbf Y\upind n,\mathbf Z\upind n)$, we have
$$\Prob_P(S_n\leq \hat G_n(\alpha)|\mathbf A)=\Gamma_n(P,\hat{\mathbf w}\upind n, \alpha),$$
with $\hat{\mathbf w}\upind n=\left(\hat w_{jlk}\upind n\right)_{j,l,k}$.
Using iterated expectations, we have
\begin{align*}
\sup_{P\in\mathcal P}\sup_{\alpha\in (0,1)}\left|\Prob_P\left(S_n\leq \hat G_n(\alpha)\right)-\alpha\right|&=\sup_{P\in\mathcal P}\sup_{\alpha\in(0,1)}\left|\Ex_P\left[\Gamma_n(P,\hat{\mathbf w}\upind n, \alpha)\right]-\alpha\right|\\
&\leq \sup_{P\in\mathcal P}\Ex_P\left[\sup_{\alpha\in(0,1)}\left|\Gamma_n(P,\hat{\mathbf w}\upind n, \alpha)-\alpha\right|\right]\\
&\leq \sup_{P\in\mathcal P}\Ex_P\left[\sup_{Q\in\mathcal P}\sup_{\mathbf w\in\mathcal W_{Q,C_1,c_1}}\sup_{\alpha\in(0,1)}|\Gamma_n(Q,\mathbf w,\alpha)-\alpha|\right]\\
&=\sup_{P\in\mathcal P}\sup_{\mathbf w\in\mathcal W_{P,C_1,c_1}}\sup_{\alpha\in(0,1)}|\Gamma_n(P,\mathbf w,\alpha)-\alpha|\to 0.
\end{align*}
\end{proof}

\section{Limit Theorems}\label{App_LimTheo}
%
The following three results are taken from Section D.2 in the supplementary material of \cite{ShahPetersCondInd}. They are versions of the central limit theorem, the weak law of large numbers and Slutsky's Lemma that hold uniformly over a collection of distributions $\mathcal P$.
\begin{lemma}[Lemma 18 in \cite{ShahPetersCondInd}]\label{lem_UnCLT}
Let $\mathcal P$ be a family of distributions such that for all $P\in\mathcal P$ the random variable $\zeta$ satisfies $\Ex_P[\zeta]=0$ and $\Ex_P[\zeta^2]=1$. Assume that there exists $\eta>0$ such that $\sup_{P\in\mathcal P} \Ex_P\left[|\zeta|^{2+\eta}\right]<\infty$. Let $(\zeta_k)_{k\in\mathbb N}$ be i.i.d. copies of $\zeta$ and define $S_n=\frac{1}{\sqrt n}\sum_{k=1}^n \zeta_k$. Then, we have
$$\lim_{n\to\infty}\sup_{P\in\mathcal P}\sup_{t\in\mathbb R}\left|\Prob_P(S_n\leq t)-\Phi(t)\right|=0.$$
\end{lemma}

\begin{lemma}[Lemma 19 in \cite{ShahPetersCondInd}]\label{lem_UnWLLN}
Let $\mathcal P$ be a family of distributions. For $P\in\mathcal P$, let $\zeta\in\mathbb R$ be a random variable with law determined by $P$ and $\Ex_P[\zeta]=0$ for all $P\in\mathcal P$. Let $\zeta_1,\zeta_2,\ldots$ be i.i.d. copies of $\zeta$ and let $S_n=\frac{1}{n}\sum_{i=1}^n \zeta_i$.  Assume that there exists $\eta>0$ such that $\sup_{P\in\mathcal P}\Ex_P\left[|\zeta|^{1+\eta}\right]<\infty$. Then, for all $\epsilon >0$
$$\lim_{n\to\infty}\sup_{P\in\mathcal P}\Prob_P\left(\left|S_n\right|>\epsilon\right)=0.$$
\end{lemma}
\begin{lemma}[Lemma 20 in \cite{ShahPetersCondInd}]\label{lem_UnSlutsky}
Let $\mathcal P$ be a family of distributions that determine the law of the random variables $(V_n)_{n\in\mathbb N}$ and $(W_n)_{n\in\mathbb N}$. Assume that
$$\lim_{n\to\infty}\sup_{P\in\mathcal P}\sup_{t\in\mathbb R}|\Prob_P(V_n\leq t)-\Phi(t)|=0.$$
Then, the following holds:
\begin{enumerate}
	\item If $W_n=o_{\mathcal P}(1)$, then
	$$\lim_{n\to\infty}\sup_{P\in\mathcal P}\sup_{t\in\mathbb R}|\Prob_P(V_n+W_n\leq t)-\Phi(t)|=0.$$
	\item If $W_n=1+o_{\mathcal P}(1)$, then
	$$\lim_{n\to\infty}\sup_{P\in\mathcal P}\sup_{t\in\mathbb R}|\Prob_P(V_n/W_n\leq t)-\Phi(t)|=0.$$
\end{enumerate}
\end{lemma}

The next lemma is taken from Section D.5 in \cite{ShahPetersCondInd}.
\begin{lemma}[Lemma 25 in \cite{ShahPetersCondInd}]\label{lem_BoundConv}
Let $\mathcal P$ be a family of distributions that determine the law of the random variables $(W_n)_{n\in\mathbb N}$. If $W_n=o_{\mathcal P}(1)$ and if there exists $C>0$ such that for all $n\in\mathbb N$ we have $|W_n|\leq C$, then
$$\lim_{n\to\infty}\sup_{P\in\mathcal P}\Ex_P\left[|W_n|\right]=0.$$
\end{lemma}

\section{Sub-Gaussian and Sub-Exponential Distributions}\label{App_SubGaussExp}
We summarise some results on sub-Gaussian and sub-exponential distributions, see for example Sections 2.5 and 2.7 in \cite{HDProbab}.

\begin{definition}[Definition 2.5.6 and Proposition 2.5.2 in \cite{HDProbab}]
A random variable $X$ with $\Ex[X]=0$ is called a \textit{sub-Gaussian} random variable if one of the following equivalent conditions is satisfied. For the parameters $K_1,K_2,K_3>0$, there exists an absolute constant $C>0$ such that for all $i,j\in\{1,2,3\}$, property $i$ implies property $j$ with parameter $K_j\leq C K_i$.
\begin{enumerate}
	\item There exists $K_1>0$ such that for all $t\geq 0$
	$$\Prob(|X|\geq t)\leq 2\exp\left(-t^2/K_1^2\right).$$
	\item There exists $K_2>0$ such that
	$$\Ex\left[\exp\left(X^2/K_2^2\right)\right]\leq 2.$$
	\item There exists $K_3>0$ such that for all $\lambda \in \mathbb R$
	$$\Ex[\exp(\lambda X)]\leq \exp\left(K_3^2 \lambda^2\right).$$
\end{enumerate}
The \textit{sub-Gaussian norm} of $X$ is defined as
$$\|X\|_{\psi_2}=\inf\left\{t>0\bigr\vert \Ex\left[\exp\left(X^2/t^2\right)\right]\leq 2\right\}.$$
\end{definition}
\begin{example}[Example 2.5.8 in \cite{HDProbab}]
A random variable $X\sim \mathcal N(0,\sigma^2)$ is sub-Gaussian with
$$\|X\|_{\psi_2}\leq C\sigma,$$
where $C>0$ is an absolute constant.

A bounded random variable $X$ is sub-Gaussian with
$$\|X\|_{\psi_2}\leq C \|X\|_{\infty}, \text{ for } C=1/\sqrt{\log 2}.$$
\end{example}

\begin{lemma}[Exercise 2.5.10 in \cite{HDProbab}]\label{lem_ExpMaxSubGauss} 
Let $X_1, \ldots, X_n$ be sub-Gaussian random variables and let $K=\max_{i=1,\ldots, n}\|X_i\|_{\psi_2}$. Then, there exists an absolute constant $C>0$ such that
$$\Ex\left[\max_{i=1,\ldots, n} |X_i|\right]\leq C K \sqrt{\log(n)}.$$
\end{lemma}

\begin{corollary}\label{cor_MaxSubGauss}
Let $X_1,X_2,\ldots$ be sub-Gaussian random variables and assume that $K=\sup_{i\in\mathbb N} \|X_i\|_{\psi_2}<\infty$. Then, 
$$\max_{i=1,\ldots, n}|X_i|=O_P\left(\sqrt{\log(n)}\right).$$
\end{corollary}
\begin{proof}
For all $M>0$, by Markov's inequality
\begin{align*}
\Prob\left(\max_{i=1,\ldots, n}|X_i|\geq M \sqrt{\log n}\right)&\leq\frac{\Ex\left[\max_{i=1,\ldots, n} |X_i|\right]}{M\sqrt{\log n}}\\
&\leq \frac{CK\sqrt{\log n}}{M\sqrt{\log n}}\\
&=CK/M\to 0,\text{ as } M\to \infty.
\end{align*}
\end{proof}

\begin{definition}[Definition 2.7.5 and Proposition 2.7.1 in \cite{HDProbab}]\label{def_SubExp}
A random variable $X$ is called a \textit{sub-exponential} random variable if one of the following equivalent conditions is satisfied. For the parameters $K_1,K_2,K_3>0$, there exists an absolute constant $C>0$ such that for all $i,j\in\{1,2,3\}$, property $i$ implies property $j$ with parameter $K_j\leq C K_i$.
\begin{enumerate}
	\item There exists $K_1>0$ such that for all $t\geq 0$
	$$\Prob(|X|\geq t)\leq 2\exp\left(-t/K_1\right).$$
	\item There exists $K_2>0$ such that
	$$\Ex\left[\exp\left(|X|/K_2\right)\right]\leq 2.$$
	\item There exists $K_3>0$ such that for all $\lambda \in [0,1/K_3]$
	$$\Ex[\exp(\lambda |X|)]\leq \exp\left(K_3 \lambda\right).$$
\end{enumerate}
The \textit{sub-exponential norm} of $X$ is defined as
$$\|X\|_{\psi_1}=\inf\left\{t>0\bigr\vert \Ex\left[\exp\left(X/t\right)\right]\leq 2\right\}.$$
\end{definition}

\begin{lemma}[Lemma 2.7.7 in \cite{HDProbab}]\label{lem_ProdSubGauss}
If $X$ and $Y$ are sub-Gaussian random variables, then $XY$ is a sub-exponential random variable and 
$$\|XY\|_{\psi_1}\leq \|X\|_{\psi_2} \|Y\|_{\psi_2}.$$
\end{lemma}

\vskip 0.2in
\bibliography{ReferencesWGCM}

\end{document}